\documentclass[10pt,a4paper,pra,superscriptaddress,twocolumn,nofootinbib,accepted=2023-07-02]{quantumarticle}
\pdfoutput=1

\usepackage[utf8]{inputenc}
\usepackage[english]{babel}
\usepackage[T1]{fontenc}
\usepackage{CJKutf8}
\usepackage{multirow}
\usepackage{subcaption}
\usepackage{amsmath,amssymb,amsthm}
\usepackage{times,ragged2e,soul}
\usepackage[dvipsnames]{xcolor}

\usepackage{tikz}
\newtheorem{theorem}{Theorem}[section]
\newtheorem{lemma}[theorem]{Lemma}

\newtheorem{fact}[theorem]{Fact}
\newtheorem{observation}[theorem]{Observation}
\newtheorem{corollary}[theorem]{Corollary}
\newtheorem{definition}[theorem]{Definition}
\usepackage{braket}
\usepackage{graphicx}
\usepackage{mathtools}
\usepackage{bbold}
\usepackage{enumitem}
\usepackage[hyperindex,breaklinks,colorlinks=true]{hyperref}
\usepackage{hhline}

\renewcommand{\H}{\mathcal{H}}
\newcommand{\B}{\mathcal{B}}
\newcommand{\NS}{\mathcal{NS}}
\newcommand{\M}{\mathcal{M}}
\newcommand{\Q}{\mathcal{Q}}
\renewcommand{\S}{\mathcal{S}}
\newcommand{\Mp}{\mathcal{M}^\text{\tiny P}}
\newcommand{\vecP}{\vec{P}}

\newcommand{\SCHSH}{\mathcal{S}_\text{\tiny CHSH}}
\newcommand{\vecPH}{\vec{P}_\text{Hardy}}
\newcommand{\vecPQ}{\vec{P}_\text{Q}}
\newcommand{\vecPQii}{\vec{P}_\text{Q,2}}
\newcommand{\vecPQiii}{\vec{P}_\text{Q,3}}
\newcommand{\vecPQiv}{\vec{P}_\text{Q,4}}
\newcommand{\vecPC}{\vec{P}_\text{Cabello}}

\newcommand{\Exp}[1]{\langle #1 \rangle}
\newcommand{\MESd}{\Phi^+_d}
\newcommand{\MESqb}{\Phi^+_2}
\newcommand{\tp}{^\text{\,\tiny T}}

\newcommand{\proj}[1]{\ket{#1}\!\!\bra{#1}}

\renewcommand{\L}{\mathcal{L}}

\newcommand{\ketA}[1]{\ket{e_{#1}}}
\newcommand{\ketB}[1]{\ket{f_{#1}}}

\newcommand{\tr}{{\rm tr}}

\usepackage[capitalize]{cleveref}

\begin{document}

\begin{CJK*}{UTF8}{bsmi}

\title{Quantum correlations on the no-signaling boundary: self-testing and more}
\author{Kai-Siang Chen}
\affiliation{Department of Physics and Center for Quantum Frontiers of Research \& Technology (QFort), National Cheng Kung University, Tainan 701, Taiwan}

\author{Gelo Noel M. Tabia}
\affiliation{Department of Physics and Center for Quantum Frontiers of Research \& Technology (QFort), National Cheng Kung University, Tainan 701, Taiwan}
\affiliation{Physics Division, National Center for Theoretical Sciences, Taipei 10617, Taiwan}
\affiliation{Center for Quantum Technology, National Tsing Hua University, Hsinchu 300, Taiwan}

\author{Chellasamy Jebarathinam }
\affiliation{Center for Theoretical Physics, Polish Academy of Sciences, Aleja Lotnik\'{o}w 32/46, 02-668 Warsaw, Poland}
\affiliation{Department of Physics and Center for Quantum Information Science, National Cheng Kung University, Tainan 70101, Taiwan}

\author{Shiladitya Mal}
\affiliation{Department of Physics and Center for Quantum Frontiers of Research \& Technology (QFort), National Cheng Kung University, Tainan 701, Taiwan}
\affiliation{Physics Division, National Center for Theoretical Sciences, Taipei 10617, Taiwan}

\author{Jun-Yi Wu}
\affiliation{Department of Physics, Tamkang University, Tamsui, New Taipei 251301, Taiwan}

\author{Yeong-Cherng Liang}
\email{ycliang@mail.ncku.edu.tw}
\affiliation{Department of Physics and Center for Quantum Frontiers of Research \& Technology (QFort), National Cheng Kung University, Tainan 701, Taiwan}
\affiliation{Physics Division, National Center for Theoretical Sciences, Taipei 10617, Taiwan}

\begin{abstract}
In device-independent quantum information, correlations between local measurement outcomes observed by spatially separated parties in a Bell test play a fundamental role. Even though it is long-known that the set of correlations allowed in quantum theory lies strictly between the Bell-local set and the no-signaling set, many questions concerning the geometry of the quantum set remain unanswered. Here, we revisit the problem of when the boundary of the quantum set coincides with the no-signaling set in the simplest Bell scenario. In particular, for each Class of these common boundaries containing $k$ zero probabilities, we provide a $(5-k)$-parameter family of quantum strategies realizing these (extremal) correlations. We further prove that self-testing is possible in {\em all} nontrivial Classes beyond the known examples of Hardy-type correlations, and provide numerical evidence supporting the robustness of these self-testing results. Candidates of one-parameter families of self-testing correlations from some of these Classes are identified. As a byproduct of our investigation, if the qubit strategies leading to an extremal nonlocal correlation are local-unitarily equivalent, a self-testing statement based on this correlation provably follows. Interestingly, all these self-testing correlations found on the no-signaling boundary are provably non-exposed. An analogous characterization for the set $\mathcal{M}$ of quantum correlations arising from  finite-dimensional maximally entangled states is also provided. En route to establishing this last result, we show that all correlations of $\mathcal{M}$ in the simplest Bell scenario are attainable as convex combinations of those achievable using a Bell pair and projective measurements. In turn, we obtain the maximal Clauser-Horne-Shimony-Holt Bell inequality violation by any maximally entangled two-qudit state and a no-go theorem regarding the self-testing of such states.
\end{abstract}
\maketitle

\section{Introduction}

As remarked by Popescu and Rohrlich~\cite{Popescu:FP:1994}, the principle of relativistic causality demands that no signals can be transmitted faster than light. In the context of a Bell test, this means that spatially separated parties cannot alter the outcome distribution observed by a remote experimenter by performing a different measurement. Interestingly, it was shown in ~\cite{Popescu:FP:1994} that this principle alone (or even in conjunction with some other principles~\cite{Brassard_NCCT,Navascues_MacroscopicLocal,Pawlowski_IC,Navascues_AlmostQuantum}) is insufficient in demarcating the boundary of the quantum set $\Q$ of correlations precisely. In contrast, Bell's principle of local causality~\cite{Bell04,Norsen:2011aa} demands that physical influences propagate continuously through space. Even though also very well-motivated, it is too restrictive since quantum correlations arising from certain entangled quantum states are known~\cite{Bell64} to be incompatible with the constraints (known as Bell inequalities) derived therefrom.

By now, it is well-known that Bell-nonlocality (hereafter abbreviated as nonlocality), i.e., the possibility of exhibiting correlations stronger than that allowed by local causality, is an indispensable resource for device-independent (DI) quantum information (QI)~\cite{Brunner_RevModPhys_2014,Scarani_DIQI_12}. In other words, the fact that some entangled states can generate nonlocal correlation guarantees its usefulness in some DI quantum information processing protocols, such as quantum key distribution~\cite{Acin07} and random number generation~\cite{Colbeck09,Pironio10}.  More generally, even though we make no assumptions about the internal workings of the devices in a DI analysis, we can infer nontrivial conclusions about the employed devices (e.g.,~\cite{Gallego:PRL:2010,Bancal11,Moroder13,Liang:PRL:2015,SLCHen16,Baccari17,Bancal:PRL:2018,Zwerger_DICGME_18,Sekatski2018,Chen_18,Arnon-Friedman:2019aa,Wagner2020,Chen2021robustselftestingof,GSD2022}) or the underlying physical theory~\cite{LZ:Ent:2019} directly from the observed correlations. In some cases, one can even perform self-testing~\cite{Mayers04} of the devices, confirming that the devices work as expected, up to local isometries. See~\cite{Supic19} for a recent review on this topic and~\cite{Wang:2018tg,Zhang:PRL:2019,Zhang:npj:2019,Gomez:PRA:2019,Bancal:Quantum:2021,Dian:PRL:2022} for some recent experimental demonstrations and challenges.

In the bipartite scenario, the principle of relativistic causality~\cite{Popescu:FP:1994} gives rise to the no-signaling (NS) conditions~\cite{Barrett_05}, hence the set of $\NS$ correlations. However, in the multipartite scenario~\cite{Horodecki:2019ub}, the relationship between the two notions is more subtle. Since $\Q$ is a subset of $\NS$, the quantum advantage offered by any entangled state is necessarily constrained by the NS conditions. Still, there are known instances (e.g.,~\cite{Pearle:PRD:1970,BC:AP:1990,BKP:PRL:2006}) where the nonlocality of quantum correlations is constrained {\em only} by these conditions. Geometrically~\cite{Goh2018}, these correspond to situations when the boundary of $\Q$ meets that of $\NS$, beyond trivial instances when they also meet the boundary of the Bell-local set. 

Interestingly, correlations on these common boundaries have found applications in various contexts. For example, they have been used to illustrate the monogamous~\cite{BKP:PRL:2006} nature (a desired feature for secret key distributions) of specific quantum correlations, to show that no extension of quantum theory can have improved predictive power~\cite{CR:NC:2011}, free randomness can be amplified~\cite{CR:NatPhys:2012}, for manifesting quantum nonlocality with an arbitrarily small amount of measurement independence~\cite{PRB:PRL:2014}, efficient verification and certification of quantum devices~\cite{GSD2022}, etc. See also~\cite{KA:IEEE:2020} for a discussion about some of these applications.

A well-known nonlocal example of such a correlation is that introduced in the context of the Hardy paradox~\cite{Hardy1993}, which demonstrates nonlocality without resorting to any Bell inequality. Intriguingly, Hardy's argument applies to {\em almost all} entangled states of two qubits, except the maximally entangled ones. In fact, as we see in this work, even a maximally entangled state of an arbitrary finite Hilbert space dimension cannot exhibit the original Hardy paradox. See, however, Proposition 9 of~\cite{Ravi18}, where the authors proposed  a Hardy-type paradox for maximally entangled states of {\em almost all dimensions} except two.

The above observation manifests the difference between $\Q$ and $\M$ (the convex hull of the set of correlations attainable by finite-dimensional maximally entangled states).  In this work, we revisit the problem of characterizing---in the simplest Bell scenario---when the boundary of $\Q$, as well as $\M$ meet the boundary of $\NS$. The former characterization can also be deduced from the findings of~\cite{Rai:PRA:2019}, but we provide in addition explicit strategies for their quantum realization.

Another interesting feature of Hardy's correlation is that it can be used to self-test~\cite{Rabelo12} and hence certify that the underlying quantum strategy is essentially unique. Since then, other examples of correlations~\cite{KCA+06,Liang:2005vl} manifesting a Hardy-type paradox were also shown~\cite{Rai:PRA:2021,Rai:PRA:2022} to be a self-test. Here, we show that one can similarly find self-testing correlations for all the other nontrivial Classes of common boundaries between $\Q$ and $\NS$. Moreover, as with the Hardy correlation~\cite{Goh2018}, all these examples are provably non-exposed points.

We structure the rest of this paper as follows. In~\cref{sec2}, we formally introduce the various sets of correlations alluded to above and the other basic knowledge required for our analysis. Then, in~\cref{Sec:Q-NS}, we provide an alternative characterization of the common boundary of $\Q$ and $\NS$, quantum strategies for their realization, and examples of nonlocal quantum correlations lying on them that can be self-tested.  In~\cref{sec4}, we prove a Lemma concerning the extreme points of the set $\M$ of correlations due to finite-dimensional maximally entangled states in the simplest Bell scenario. We then use this to prove, for such states, the maximal Clauser-Horne-Shimony-Holt~\cite{CHSH_1969} Bell inequality violation and a no-go theorem for their self-testing. Finally, \cref{sec5} summarizes the results obtained and outlines some possible directions for future research.

\section{Preliminaries}~\label{sec2}
\subsection{Sets of correlations}\label{Sec:Sets}

In the simplest Bell scenario, two experimenters (Alice and Bob) may perform two dichotomic measurements each. The correlation between the local measurement outcomes may  be summarized~\cite{Brunner_RevModPhys_2014} by a collection of joint conditional probability distributions $\vecP=\{P(a,b|x,y)\}_{x,y,a,b=0,1}$ where $a$ and $b$ ($x$ and $y$) are, respectively, the label for the measurement outcome (setting) for the first and the second party. 

Note that entries of $\vecP$ have to satisfy certain constraints. For example, a probability distribution has to be non-negative and normalized, so
\begin{subequations}\label{Eq:Prob}
\begin{gather}
    P(a,b|x,y)\ge 0 \quad \forall\,\, a,b,x,y, \label{nonneg}\\
    \sum_{a,b}P(a,b|x,y)=1,\quad \forall\,\, x,y.\label{normalization}
\end{gather}
\end{subequations}
 
In addition, we are exclusively interested only in the set $\NS$, i.e., the set of correlations that satisfy the NS conditions~\cite{Barrett_05}
 \begin{equation}\label{Eq:NS}
 \begin{split}
     \sum_a P(a,b|x,y) =\sum_a P(a,b|x',y)\,\, \forall\,\, x,x',b,y\\
     \sum_b P(a,b|x,y) =\sum_b P(a,b|x,y')\,\, \forall\,\, y,y',a,x.
    \end{split}
 \end{equation}
The set of all such correlations forms a convex polytope. In particular, within the subspace of $\vecP$ satisfying \cref{Eq:NS}, $\NS$ is fully characterized by ``positivity facets'',\footnote{A facet of a polytope $\mathcal{P}$ is a boundary of $\mathcal{P}$ having maximal dimension.} i.e., those inequalities given in \cref{nonneg}.  In the simplest Bell scenario, $\NS$ is spanned by 24 vertices, 16 of which are Bell-local (or simply {\em local}), while the remaining 8 take the form of~\cite{Barrett_05}:
 \begin{equation}\label{Eq:PRBoxes}
    P(a,b|x,y)=\frac{1}{2}\delta_{a\oplus b, xy\oplus \alpha x\oplus \beta y\oplus \gamma}
 \end{equation}
 where $\delta_{c,d}$ is the Kronecker delta between $c$ and $d$, $\oplus$ means addition modulo $2$, and $\alpha,\beta,\gamma\in\{0,1\}$ are free parameters characterizing these nonlocal extremal points. In honor of the pioneering work of Popescu and Rohrlich (PR)~\cite{Popescu:FP:1994}, these nonlocal extremal points are referred to as PR boxes~\cite{Barrett_05}.
 
In contrast, the convex hull of the 16 local extremal points gives the so-called Bell polytope, or the local polytope, which we denote by $\L$. This is the set of correlations allowed by a local hidden variable theory~\cite{Bell64}, or equivalently the set of correlations that respects the principle of local causality~\cite{Bell04,Norsen:2011aa}
 \begin{equation}\label{Eq:LHV}
     P(a,b|x,y)=\sum_\lambda q_\lambda\, p(a|x,\lambda)p(b|y,\lambda) \quad \forall\, x,y,
 \end{equation}
 where the local response functions $p(a|x,\lambda), p(b|y,\lambda)$ can be taken, without loss of generality, to be either $0$ or $1$, and $\lambda$ can be understood as the local hidden variable (or computationally as a label for the extreme points of the local polytope). 
 
In this simplest Bell scenario, the Bell polytope is equivalently characterized by 16 positivity facets and 8 facet-defining Bell inequalities~\cite{Collins04}, first discovered by Clauser-Horne-Shimony-Holt (CHSH)~\cite{CHSH_1969}. Consequently, this Bell scenario is often referred to as the CHSH Bell scenario. For subsequent discussions, the following CHSH (Bell) inequality is of particular relevance:
\begin{equation}\label{Eq:CHSH}
    \SCHSH = \sum_{x,y,a,b=0}^1 (-1)^{xy+a+b+1} P(a,b|x,y)   \overset{\L}{\le} 2
\end{equation}

 The very basis of DIQI is that not all quantum correlations, namely, those obtained by performing local measurements on a shared quantum state $\rho$ can be decomposed in the form of ~\cref{Eq:LHV}. Formally, the set of quantum correlations $\Q$ consists of all $\vecP$ that respect Born's rule:
 \begin{equation}\label{Eq:Born}
     P(a,b|x,y)=\tr[(M^A_{a|x}\otimes M^B_{b|y})\rho],
 \end{equation}
 where ${M^A_{a|x}}$ and ${M^B_{b|y}}$ are, respectively, the positive operator-valued measure (POVM) elements associated with Alice's and Bob's local measurements. Notice that there is no restriction on the local Hilbert space dimension in~\cref{Eq:Born}. Thus, by Neumark's theorem~\cite{Neumark},  we can always take the POVM elements to be projectors when deciding whether a given $\vecP$ is in $\Q$.

Within $\Q$, a set of particular interest is the set of correlations attainable by finite-dimensional maximally entangled states, i.e., $\rho=\proj{\MESd}$ with
\begin{equation}\label{Eq:MES}
    \ket{\MESd}=\tfrac{1}{\sqrt{d}}\sum_{i=0}^{d-1}\ket{i}\!\ket{i},\quad 2 \le d < \infty.
\end{equation}
For each $d\ge 2$, we denote by $\M_d$ the convex hull of the set of all correlations attainable by $\ket{\MESd}$ via \cref{Eq:Born} and by $\M$ the convex hull of the union $\cup_{d=2,3,\ldots} \M_d$. From \cref{Eq:Born}, it follows that for any such states
\begin{equation}\label{Eq:MESCorrelation}
\begin{split}
    P(a,b|x,y)
        =\tfrac{1}{d}\tr\!\left[M^A_{a|x} \left(M^B_{b|y}\right)^\text{\tiny T}\right]
        =\tfrac{1}{d}\tr\!\left[\left(M^A_{a|x}\right)^\text{\tiny T}\! M^B_{b|y}\right],
\end{split}
\end{equation}
where $\tp$ denotes the transposition operation. Clearly, $\M\subsetneq\Q$, and the fact that the inclusion is generally strict was shown independently in~\cite{Liang2011}, ~\cite{Vidick2011}, and ~\cite{Junge:2011aa} (see also~\cite{Christensen15} and~\cite{Lin:Quantum:2022}).

In the CHSH Bell scenario, it is expedient to consider also the correlator, which is the probability of getting the same outcome minus the probability of getting different outcomes:
\begin{equation}\label{Eq:Correlator}
	\langle A_xB_y \rangle:=\sum_{a,b=0}^1 (-1)^{a+b} P(a,b|x,y).
\end{equation}
In quantum theory, this corresponds to the expectation value $\langle A_xB_y \rangle = \tr(A_x\otimes B_y\,\rho)$ where
\begin{equation}\label{Eq:GeneralObservables}
	A_x:= M^A_{0|x} - M^A_{1|x}, \quad B_y:= M^B_{0|y} - M^B_{1|y}
\end{equation}
are both dichotomic observables (i.e., with $\pm1$ eigenvalues). From \cref{Eq:GeneralObservables}, we can also express the POVM elements for all $x,y,a,b\in\{0,1\}$ in terms of the observables as
\begin{equation}\label{Eq:Obs->POVM}
        M^A_{a|x} = \frac{\mathbb{1}_2+(-1)^a A_x}{2}\,\text{ and }\,
        M^B_{b|y} =\frac{\mathbb{1}_2+(-1)^b B_y}{2}.
\end{equation}

\subsection{Self-testing}

One of our interests is to show that nonlocal correlations lying on the common boundaries of $\Q$ and $\NS$ can be used to self-test~\cite{Supic19} certain reference state and measurements. For completeness, we now recall from~\cite{Goh2018} the following definition for the self-testing of states and measurements.

\begin{definition}[Self-testing]\label{Dfn:Self-testing}
A quantum correlation $\vecP\in\Q$ is said to self-test the reference quantum realization $\{\ket{\widetilde{\psi}}, \{\widetilde{M^A_{a|x}}\}, \{\widetilde{M^B_{b|y}}\}\}$ of $\vecP$ if {\em for all} states $\ket{\psi}_{AB}$ and measurements $\{M^A_{a|x}\}$, $\{M^B_{b|y}\}$ giving the same correlation $\vecP$, one can find $\Phi:=\Phi_A \otimes \Phi_B$ with local isometries $\Phi_A: \H_A\mapsto\H_{A'}\otimes \H_{A''}$ and $\Phi_B: \H_B\mapsto\H_{B'}\otimes \H_{B''}$ such that
\begin{subequations}\label{Eq:Selftest}
\begin{equation}
    \Phi (M^A_{a|x}\otimes M^B_{b|y}\ket{\psi}_{AB}) =\ket{\varsigma}_{A'B'}\otimes \widetilde{M^A_{a|x}}\otimes \widetilde{M^B_{b|y}}\ket{\widetilde{\psi}}_{A''B''},\label{st-s&m}
\end{equation}
holds for some auxiliary ``junk" state $\ket{\varsigma}_{A'B'}$. 
By summing over all $a$, $b$ and using the linearity of $\Phi$, one also obtains the usual self-testing requirement for a state
\begin{equation}
    \Phi (\ket{\psi}_{AB}) =\ket{\varsigma}_{A'B'}\otimes \ket{\widetilde{\psi}}_{A''B''}.\label{st-s}
\end{equation}
\end{subequations}
\end{definition}
Notice that there are examples of $\vecP\in\Q$ that self-test a reference state but not the underlying measurements, see, e.g.,~\cite{Jeba:PRR:2019,Kaniewski:PRR:2020,Gigena:PRA:2022}. Furthermore, in the CHSH Bell scenario, a  self-testing quantum correlation is necessarily extremal~\cite{Goh2018} in $\Q$. 

Let us now recall from~\cite{Masanes06} and~\cite{Franz11} the following characterization of the extreme points of $\Q$ applicable to this Bell scenario. 
\begin{fact}[Masanes~\cite{Masanes06}]\label{CHSH:ExtremeQ}
In an $N$-partite Bell scenario involving only two dichotomic measurements, all extreme points of $\Q$ are achievable by measuring $N$-qubit pure states with projective measurements.
\end{fact}

\begin{fact}[Franz {\em et al.}~\cite{Franz11}]\label{CHSH:ExtremeQ2}
All extreme points of $\Q$ in the CHSH Bell scenario are achievable by measuring a two-qubit pure state $\sum_{i,j=0}^1 c_{ij}\ket{i}\ket{j}$, $\sum_{i,j} c^2_{ij}=1$, 
with the observables:
\begin{equation}\label{Eq:Meas:Class1}
\begin{split}
        &A_0 = B_0 = \sigma_z,\\
        A_1 = C_\alpha\, \sigma_z +&S_\alpha\, \sigma_x,\quad
        B_1 = C_\beta\, \sigma_z +S_\beta\, \sigma_x,
\end{split}
\end{equation}
where $C_\alpha:=\cos\alpha$, $C_\beta:=\cos\beta$, $S_\alpha:=\sin\alpha$, $S_\beta:=\sin\beta$, and $\alpha,\beta\in[0,2\pi)$, $c_{ij}\in\mathbb{R}$ are free parameters.
\end{fact}

With these facts in mind, it seems natural to first understand whether seemingly different qubit strategies for realizing a given correlation $\vecP$ are indeed different. To this end, we introduce the following definition to capture equivalence classes of quantum strategies.
\begin{definition}[$d$-equivalent]\label{Dfn:Equivalent}
A quantum realization of $\vecP$ consisting of $\left\{\ket{\psi}, \{M^A_{a|x}\}, \{M^B_{b|y}\}\right\}$ is said to be $d$-equivalent to another quantum realization of $\vecP$ given by $\left\{\ket{\tilde{\psi}}, \{\widetilde{M^A_{a|x}}\}, \{\widetilde{M^B_{b|y}}\}\right\}$ if there exist $d$-dimensional local unitary operators $u_A$ and $u_B$ such that
\begin{subequations}\label{Eq:Dfn:d-equivalent}
\begin{align}
	&u_A\otimes u_B \ket{\psi} = \ket{\tilde{\psi}},\label{Eq:Equivalent:State}\\
	u_A\otimes u_B (M^A_{a|x}\otimes &M^B_{b|y})u_A^\dag\otimes u_B^\dag = \widetilde{M^A_{a|x}}\otimes\widetilde{M^B_{b|y}}.\label{Eq:Equivalent:Meas}
\end{align}
\end{subequations}
\end{definition} 
From \cref{Eq:Born}, $d$-equivalent quantum strategies clearly result in the same correlation $\vecP$. The converse, however, is not necessarily true. To show that a given $\vecP$ in the CHSH Bell scenario is a self-test, it thus seems natural--- given Fact~\ref{CHSH:ExtremeQ}---to first determine if all qubit strategies realizing $\vecP$ are $2$-equivalent. If so, the following Theorem allows one to promote such an observation to a self-testing statement.
\begin{theorem}~\label{Thm:self-test:upgrade}
In the CHSH Bell scenario, if all qubit strategies $\left\{\ket{\psi}, \{M^A_{a|x}\}, \{M^B_{b|y}\}\right\}$ giving an extremal nonlocal correlation $\vecP$ are 2-equivalent to a specific reference strategy $\left\{\ket{\tilde{\psi}}, \{\widetilde{M^A_{a|x}}\}, \{\widetilde{M^B_{b|y}}\}\right\}$, then $\vecP$ self-tests the quantum realization given by this reference strategy.
\end{theorem}
For a proof of the Theorem, see~\cref{App:Proof:SelfTestingTheorem}.

\subsection{Exposed vs non-exposed point }

Often, instead of using the full measurement statistics given by $\vecP$, self-testing can also achieved by observing the maximal quantum violation of a Bell inequality. In this case, the corresponding extremal $\vecP=\vecP_0$ must be a {\em unique} maximizer of some Bell function, violating some Bell inequality maximally. Following~\cite{Goh2018}, we refer to such points as being exposed in $\Q$. More formally, since any Bell function may be specified by a vector of real numbers $\vec{B}$, an exposed point $\vecP_0$ in $\Q$ is a correlation that satisfies
\begin{equation}\label{Dfn:Exposed}
	\max_{\vecP\in\Q | \vecP\neq\vecP_0} \vec{B}\cdot\vecP < \vec{B}\cdot\vecP_0.
\end{equation}	
Conversely, if there is no Bell function $\vec{B}$ such that \cref{Dfn:Exposed} holds for $\vecP_0$, then the correlation $\vecP_0$ is said to be non-exposed (in $\Q$). In~\cite{Goh2018}, motivated by the findings presented in~\cite{Mancinska:2014uc}, the correlation leading to the Hardy paradox~\cite{Hardy1993} was shown to be non-exposed. Below, we show that the same holds also for several other extremal quantum correlations lying on the boundary of $\NS$.

\subsection{Correlation tables and relabeling}

A correlation $\vecP$ is conveniently represented using a table showing all entries of $\vecP$, as shown in \cref{tab:1}. Note that the non-negativity requirement of~\cref{nonneg} demands that all entries in the correlation table are non-negative. The normalization requirement of~\cref{normalization}, on the other hand, means that within each block corresponding to a fixed value of $x,y$, the sum of all entries gives one.

\begin{table}[h]
    \small\addtolength{\tabcolsep}{-1pt}
    \captionsetup{justification=RaggedRight,singlelinecheck=off}
    \centering
    \begin{tabular}{c|c||c|c||c|c||}
    \hhline{~~|t|~~|t|~~|t|}
     \multicolumn{2}{c||}{}  & \multicolumn{2}{c||}{$x=0$} & \multicolumn{2}{c||}{$x=1$} \\ 
     \hhline{~~--||--||}
     \multicolumn{2}{c||}{}  & $a=0$ & $a=1$ & $a=0$ & $a=1$ \\
    \hhline{|=t=::==::==:|}
    \multirow{2}{*}{$y=0$} & $b=0$ &  & $\delta$ &  &  \\ \hhline{~-||--||--||}
                             & $b=1$ & &  &  &  \\
    \hhline{|==::==::==:|}
    \multirow{2}{*}{$y=1$} & $b=0$ & & &  & \\
    \hhline{~-||--||--||}
                             & $b=1$ &  &  & &  \\
    \hhline{|=b=:b:==:b:==:b|}
         \end{tabular}
\caption{\label{tab:1} A correlation table showing the entries of a correlation $\vecP$. For instance, if the probability of getting outcomes $a=1,b=0$ for the given settings $x=0,y=0$ is $\delta$, we write $P(1,0|0,0)=\delta$, or equivalently, a $\delta$ as the corresponding entry in the correlation table. To simplify the presentation, we omit ``$a=$", ``$b=$", and the lines separating the different outcomes in all the correlation tables presented below. }
\end{table}

For the rest of the paper, it is worth remembering that a relabeling of the measurement settings $x$ ($y$), the measurement outcomes $a$ ($b$), or the parties Alice $\leftrightarrow$ Bob does not lead to any change in the nature of the correlation. In other words, if $\vecP$ is a member of $\L$, $\Q$, or $\NS$, it remains so after any such relabeling or even combinations thereof. In the context of the correlation table, cf.~\cref{tab:1}, a relabeling of the measurement settings $x=0\leftrightarrow x=1$ amounts to interchanging the two block columns (consisting of two columns ) whereas a relabeling of the measurement outcomes $a=0\leftrightarrow a=1$ for some given value of $x$ amounts to swapping the two columns associated with that value of $x$. Similarly, a relabeling of Bob's measurement settings $y=0\leftrightarrow y=1$ or his outcome $b=0\leftrightarrow b=1$ entails a permutation of the corresponding rows in the table. Finally, a relabeling of the parties amounts to performing a transposition of the table.

\subsection{A signature of Bell-local correlations}

When seen as a convex combination of the extreme points of $\NS$, {\em any} correlation $\vecP$ lying outside the local polytope, i.e., $\vecP\not\in\L$ must have contribution(s) from at least one PR box, such as that obtained from \cref{Eq:PRBoxes} by setting $\alpha=\beta=0$ and $\gamma=1$; see~\cref{Eq:vecPR} for the corresponding correlation table. 
\begin{equation} \label{Eq:vecPR}
    \vec{P}_{\text{PR}} := ~
    \begin{tabular}{cc|cc|cc|}
         &  & \multicolumn{2}{c|}{$x=0$} & \multicolumn{2}{c|}{$x=1$} \\
         &  & $0$ & $1$ & $0$ & $1$ \\
        \hline
        \multirow{2}{*}{$y=0$} & $0$ & $0$ & $\frac{1}{2}$ & $0$ & $\frac{1}{2}$\\
                               & $1$ & $\frac{1}{2}$  & $0$  & $\frac{1}{2}$ &$0$ \\ [1.02ex]
        \hline
        \multirow{2}{*}{$y=1$} & $0$ & $0$  & $\frac{1}{2}$ & $\frac{1}{2}$ & $0$ \\
                               & $1$ & $\frac{1}{2}$ &$0$  & $0$& $\frac{1}{2}$\\ [1.02ex]
        \hline
        \end{tabular}
\end{equation}

A distinctive feature of the PR box correlation shown in \cref{Eq:vecPR} is that it consists of three anti-diagonal blocks satisyfing $xy=0$ and one diagonal block for $x=y=1$. Or equivalently, $\vecP_{\text{PR}}$ consists of three blocks of diagonal zeros (DZs) and one of block anti-diagonal zeros (ADZs).  Moreover, all non-vanishing entries are $\frac{1}{2}$. With some thought, it is easy to see that all the other PR boxes, which can be obtained via a relabeling from~\cref{Eq:vecPR}, must have either (i) three blocks of DZs and one block of ADZs [like the one shown in~\cref{Eq:vecPR}] or (ii) three blocks of ADZs and one block of DZs.
 
\begin{table*}
         \small\addtolength{\tabcolsep}{-1pt}
         \captionsetup{justification=RaggedRight,singlelinecheck=off}
    \begin{subtable}[h]{0.23\textwidth}
        \centering
        \begin{tabular}{cc|cc|cc|}
         &  & \multicolumn{2}{c|}{$x=0$} & \multicolumn{2}{c|}{$x=1$} \\
         &  & $0$ & $1$ & $0$ & $1$ \\
        \hline
        \multirow{2}{*}{$y=0$} & $0$ & $0$ & & & \\
         & $1$ &  & $0$ & & \\
        \hline
        \multirow{2}{*}{$y=1$} & $0$ & & & & \\
         & $1$ & & & & \\
        \hline
        \end{tabular}
        \caption{Zeros at the diagonal positions (DZs). }
         \label{tab:3a}
    \end{subtable}
    \quad
    \begin{subtable}[h]{0.23\textwidth}
        \centering
        \begin{tabular}{cc|cc|cc|}
         &  & \multicolumn{2}{c|}{$x=0$} & \multicolumn{2}{c|}{$x=1$} \\
         &  & $0$ & $1$ & $0$ & $1$ \\
        \hline
        \multirow{2}{*}{$y=0$} & $0$ &  & $0$ & & \\
         & $1$ & $0$ &   & & \\
        \hline
        \multirow{2}{*}{$y=1$} & $0$ &  &  & & \\
         & $1$ & & & & \\
        \hline
        \end{tabular}
        \caption{Zeros at the anti-diagonal positions (ADZs). }
         \label{tab:3b}
    \end{subtable}
    \quad
    \begin{subtable}[h]{0.23\textwidth}
        \centering
        \begin{tabular}{cc|cc|cc|}
         &  & \multicolumn{2}{c|}{$x=0$} & \multicolumn{2}{c|}{$x=1$} \\
         &  & $0$ & $1$ & $0$ & $1$ \\
        \hline
        \multirow{2}{*}{$y=0$} & $0$ & $0$ & $0$ &  &  \\
         & $1$ &  &  &  &  \\
        \hline
        \multirow{2}{*}{$y=1$} & $0$ & & & & \\
         & $1$ & & & & \\
        \hline
        \end{tabular}
        \caption{Zeros on the same row.}
        \label{tab:3c}
    \end{subtable}
    \quad
    \begin{subtable}[h]{0.23\textwidth}
        \centering
        \begin{tabular}{cc|cc|cc|}
         &  & \multicolumn{2}{c|}{$x=0$} & \multicolumn{2}{c|}{$x=1$} \\
         &  & $0$ & $1$ & $0$ & $1$ \\
        \hline
        \multirow{2}{*}{$y=0$} & $0$ & $0$ & & & \\
         & $1$ & $0$ & &  & \\
        \hline
        \multirow{2}{*}{$y=1$} & $0$ & & & & \\
         & $1$ & & & & \\
        \hline
        \end{tabular}
        \caption{Zeros on the same column.}
         \label{tab:3d}
    \end{subtable}
    
    \caption{All possibilities of correlation tables with two zeros in the same block (shown here for $x=y=0$).}
    \label{tab:3}
\end{table*}

Since any convex combination of non-negative numbers is non-vanishing, the above observations imply that a $\vecP$ containing two neighboring zeros within a block (such as that shown in~\cref{tab:3c} and ~\cref{tab:3d}) cannot have a nontrivial contribution from a PR box. Hence, as was already noted in~\cite{Fritz_2011}, we have the following simple sufficient condition for a given $\vecP$ to be in $\L$.
\begin{observation}\label{SufficiencyL}
In the simplest Bell scenario, if the correlation table for a certain $\vecP\in\NS$ has two neighboring zeros within a block, then $\vecP$ is necessarily Bell-local, i.e., $\vecP\in\L$.
\end{observation}
For example, the local deterministic correlation $P(a,b|x,y)=\delta_{a,x}\delta_{b,y}$, which is an extreme point of both $\L$ and $\NS$, has neighboring zeros in every block. Observation~\ref{SufficiencyL}, of course, allows one to identify  $\vecP\in\L$ beyond such trivial examples.

\subsection{Correlations lying on the $\NS$ boundary}

Finally, note the following fact regarding correlations lying on the boundary of $\NS$.
\begin{fact}
	Each (positivity) facet of $\NS$ is characterized by having {\em exactly} one joint conditional probability distribution $P(a,b|x,y)$ being zero.
\end{fact}
In the CHSH Bell scenario, this corresponds to having one and only one out of the 16 inequalities of~\cref{nonneg} saturated. Consequently, when $k$ of the entries of $\vecP$ vanish, the correlation $\vecP$ lies at the intersection of $k$ such positivity facets of $\NS$. Thus, we have the following observation, which we rely on for the rest of the paper.
\begin{observation}
	A given  $\vecP\in\Q$ belongs to the boundary of $\NS$ if and only if at least one of its entries in  $\{P(a,b|x,y)\}_{x,y,a,b=0,1}$ is zero.
\end{observation}

\section{When the boundary of $\Q$ meets the boundary of $\NS$}~\label{Sec:Q-NS}

In this section, we proceed to classify the different possibilities of when a correlation $\vecP$ lying on the quantum boundary also sits on the no-signaling boundary. Our classification, inspired by the work of Fritz~\cite{Fritz_2011}, is based on the different number of zeros appearing in the correlation table. Since $\NS$ and $\Q$ are both eight-dimensional in the CHSH Bell scenario, it is worth noting that their common boundary with $k$ zeros is $(8-k)$-dimensional, as noted in~\cite{Rai:PRA:2019}.

For characterizing the common boundary, we prove, independently of~\cite{Rai:PRA:2019}, whether such correlations are necessarily local and provide families of quantum strategies realizing the extremal correlations from each Class. To this end, we start from the five-parameter family of (extremal) quantum strategies given in Fact~\ref{CHSH:ExtremeQ2}. By imposing the $k$ zero constraints of each Class, where $k\in\{1,2,3,4\}$, we then arrive at our $(5-k)$-parameter family of quantum strategies for each Class. Among them, we further identify the one from each Class that maximally violates the CHSH Bell inequality of~\cref{Eq:CHSH}.

Furthermore, let us remind that $\Q$ is the convex hull of its extreme points. It thus follows from Fact~\ref{CHSH:ExtremeQ}  that {\em if no} two-qubit pure state with projective measurements can generate a correlation table having zeros at the designated positions, then the corresponding correlation $\vecP$ must lie outside $\Q$. Similarly, {\em if all} two-qubit states with projective measurements giving rise to $k$ zeros at some designated positions in a correlation table are such that the corresponding $\vecP\in\L$, then there is {\em no} quantum $\vecP$ outside $\L$ giving the same set of zeros in the correlation table. Therefore, for the present purpose, it suffices to restrict to pure qubit strategy in conjunction with projective measurements, even though the set of dimension-restricted quantum correlations is known generally to be concave (see, e.g., Refs.~\cite{Pal:PRA:2009,DW15}).

\subsection{Further signature of local quantum correlations}

With the above observations, we arrive at our first result.
\begin{lemma}\label{lem1}
In the CHSH Bell scenario, if there are two (or more) zeros in the same row (or column) of the correlation table, then the corresponding quantum correlation is local.
\end{lemma}
\begin{proof}
Given Observation~\ref{SufficiencyL}, it remains to consider correlation tables having two zeros in the same row (or column) but different blocks. By relabeling, they take the form of~\cref{tab:4}.
\begin{table}[h]
    \captionsetup{justification=RaggedRight,singlelinecheck=off}
    \centering
    \begin{tabular}{cc|cc|cc|}
     &  & \multicolumn{2}{c|}{$x=0$} & \multicolumn{2}{c|}{$x=1$} \\
     &  & $0$ & $1$~ & $0$ & $1$ \\
    \hline
    \multirow{2}{*}{$y=0$} & $0$ & $0$&  & $0$ &  \\
                             & $1$ & & & & \\
    \hline
    \multirow{2}{*}{$y=1$} & $0$ & & & &  \\
                             & $1$ & & & & \\
    \hline
    \end{tabular}
     \caption{A correlation table with two zeros in the same row but different blocks. }
     \label{tab:4}
\end{table}\\
Then, it suffices to show that all extreme points of $\Q$ having the same two zeros as in Table~\ref{tab:4} are local.

\begin{table*}[]
\small\addtolength{\tabcolsep}{-1pt}
    \captionsetup{justification=RaggedRight}
    \begin{subtable}[h]{0.23\textwidth}
        \centering
        \begin{tabular}{cc|cc|cc|}
         &  & \multicolumn{2}{c|}{$x=0$} & \multicolumn{2}{c|}{$x=1$} \\
         &  & $0$ & $1$~ & $0$ & $1$ \\
        \hline
        \multirow{2}{*}{$y=0$} & $0$ & $0$ &  & & \\
                               & $1$ &  &   & &$0$ \\
        \hline
        \multirow{2}{*}{$y=1$} & $0$ &  &  &$0$ & \\
                               & $1$ & &$0$ & & \\
        \hline
        \end{tabular}
        \caption{Class 4a: All four blocks have one zero each.}
         \label{tab:class4a}
    \end{subtable}
    \quad
    \begin{subtable}[h]{0.23\textwidth}
        \centering
        \begin{tabular}{cc|cc|cc|}
         &  & \multicolumn{2}{c|}{$x=0$} & \multicolumn{2}{c|}{$x=1$} \\
         &  & $0$ & $1$~ & $0$ & $1$ \\
        \hline
        \multirow{2}{*}{$y=0$} & $0$ &  & $0$ & & \\
         & $1$ & $0$ &  & & \\
        \hline
        \multirow{2}{*}{$y=1$} & $0$ & & & & $0$\\
         & $1$ & & &$0$ & \\
        \hline
        \end{tabular}
        \caption{Class 4b: Two non-neighboring blocks have two DZs  each.}
         \label{tab:class4b}
    \end{subtable}
    \quad
    \begin{subtable}[h]{0.23\textwidth}
        \centering
        \begin{tabular}{cc|cc|cc|}
         &  & \multicolumn{2}{c|}{$x=0$} & \multicolumn{2}{c|}{$x=1$} \\
         &  & $0$ & $1$~ & $0$ & $1$ \\
        \hline
        \multirow{2}{*}{$y=0$} & $0$ & $0$ &  & & \\
                               & $1$ &  &   & &$0$ \\
        \hline
        \multirow{2}{*}{$y=1$} & $0$ &  &  & & \\
                               & $1$ & &$0$ & & \\
        \hline
        \end{tabular}
        \caption{Class 3a: Three blocks having one zero each. }
         \label{tab:class3a}
    \end{subtable}
    \quad
    \begin{subtable}[h]{0.23\textwidth}
        \centering
        \begin{tabular}{cc|cc|cc|}
         &  & \multicolumn{2}{c|}{$x=0$} & \multicolumn{2}{c|}{$x=1$} \\
         &  & $0$ & $1$~ & $0$ & $1$ \\
        \hline
        \multirow{2}{*}{$y=0$} & $0$ & $0$ &  & & \\
         & $1$ &  & $0$ & & \\
        \hline
        \multirow{2}{*}{$y=1$} & $0$ & & & & $0$\\
         & $1$ & & & & \\
        \hline
        \end{tabular}
        \caption{Class 3b: One block has two DZs, the remaining zero is in a non-neighboring block.}
         \label{tab:class3b}
    \end{subtable}
    \quad
    \begin{subtable}[h]{0.23\textwidth}
        \centering
        \begin{tabular}{cc|cc|cc|}
         &  & \multicolumn{2}{c|}{$x=0$} & \multicolumn{2}{c|}{$x=1$} \\
         &  & $0$ & $1$~ & $0$ & $1$ \\
        \hline
        \multirow{2}{*}{$y=0$} & $0$ & $0$ & & & \\
         & $1$ &  & $0$ & & \\
        \hline
        \multirow{2}{*}{$y=1$} & $0$ & & & & \\
         & $1$ & & & & \\
        \hline
        \end{tabular}
        \caption{Class 2a: The two zeros appear in the same block at the (anti) diagonal positions.}
	\label{tab:class2a}
    \end{subtable}
    \quad
    \begin{subtable}[h]{0.23\textwidth}
        \centering
        \begin{tabular}{cc|cc|cc|}
         &  & \multicolumn{2}{c|}{$x=0$} & \multicolumn{2}{c|}{$x=1$} \\
         &  & $0$ & $1$~ & $0$ & $1$ \\
        \hline
        \multirow{2}{*}{$y=0$} & $0$ & $0$ &  & & \\
         & $1$ &  &   & & $0$ \\
        \hline
        \multirow{2}{*}{$y=1$} & $0$ &  &  & & \\
         & $1$ & & & & \\
        \hline
        \end{tabular}
        \caption{Class 2b: The two zeros appear in blocks with the same value of $x$ (or $y$).}
         \label{tab:class2b}
    \end{subtable}
    \quad
    \begin{subtable}[h]{0.23\textwidth}
        \centering
        \begin{tabular}{cc|cc|cc|}
         &  & \multicolumn{2}{c|}{$x=0$} & \multicolumn{2}{c|}{$x=1$} \\
         &  & $0$ & $1$~ & $0$ & $1$ \\
        \hline
        \multirow{2}{*}{$y=0$} & $0$ & $0$ &  &  &  \\
         & $1$ &  &  &  &  \\
        \hline
        \multirow{2}{*}{$y=1$} & $0$ & & & & $0$ \\
         & $1$ & & & & \\
        \hline
        \end{tabular}
        \caption{Class 2c: The two zeros appear in non-neighboring blocks. }
        \label{tab:class2c}
    \end{subtable}
    \quad
    \begin{subtable}[h]{0.23\textwidth}
        \centering
        \begin{tabular}{cc|cc|cc|}
         &  & \multicolumn{2}{c|}{$x=0$} & \multicolumn{2}{c|}{$x=1$} \\
         &  & $0$ & $1$~ & $0$ & $1$ \\
        \hline
        \multirow{2}{*}{$y=0$} & $0$ & $0$ & & & \\
         & $1$ &  & &  & \\
        \hline
        \multirow{2}{*}{$y=1$} & $0$ & & & & \\
         & $1$ & & & & \\
        \hline
        \end{tabular}
        \caption{Class 1: There is only zero in the table. }
         \label{tab:class1}
    \end{subtable}
    \caption{Modulo the freedom of relabeling, there are only eight distinct Classes of correlation tables having four or fewer zeros. Classes having more than one zero in a row or column are omitted as they cannot contain nonlocal quantum correlations (see Lemma~\ref{lem1}.)}
    \label{tab:Classes}
\end{table*}

For the realization of these extreme points, we may consider, by Fact~\ref{CHSH:ExtremeQ}, a 2-qubit pure state in conjunction with projective measurements. Note, however, that if any of these POVM elements is full rank, the resulting quantum correlation is necessarily local (see, e.g., Appendix B.3.1 of~\cite{Liang:PhDthesis}). Hence, we consider hereafter only rank-one projective measurements in the orthonormal bases 
\begin{equation}\label{Eq_MeasBases}
\begin{split}
	&\{\ketA{0|x},\ketA{1|x}\}_{x=0,1} \text{ for Alice},\\
	\text{and \ \ } &\{\ketB{0|y},\ketB{1|y}\}_{y=0,1} \text{ for Bob},
\end{split}
\end{equation}
or equivalently, dichotomic observables [\cref{Eq:GeneralObservables}]:
\begin{equation}\label{Eq_Observables}
\begin{split}
	A_x = \proj{e_{0|x}}-\proj{e_{1|x}},\\ B_y = \proj{f_{0|y}}-\proj{f_{1|y}}.
\end{split}
\end{equation}
By the completeness of these bases, we may express a generic bipartite two-qubit pure state as :
\begin{equation}
    \ket{\Psi}=\sum_{a,b=0}^1c_{ab}\ketA{a|0}\!\ketB{b|0}
\end{equation}
where $c_{ab}\in\mathbb{C}$ and $\sum_{a,b} |c_{ab}|^2=1$. To satisfy $P(0,0|0,0)=0$, we must have $c_{00}=0$, then 
\begin{equation}\label{Eq_SimplifiedState}
    \ket{\Psi}=c_{10}\ketA{1|0}\!\ketB{0|0}+c_{01}\ketA{0|0}\!\ketB{1|0}+c_{11}\ketA{1|0}\!\ketB{1|0}.
\end{equation}
Now, either $c_{10}=0$ or $1\ge |c_{10}|>0$. For the former, the quantum state is separable state and the resulting correlation $\vecP$ must be~\cite{Werner:PRA:1989} local. For the latter, the requirement of $P(0,0|1,0)=0$ together with \cref{Eq:Born}, \cref{Eq_MeasBases}, and \cref{Eq_SimplifiedState} imply that $\braket{e_{0|1}|e_{1|0}}=0$.
This means that Alice's two measurements are identical up to a relabeling of the outcome, i.e., they are jointly measurable~\cite{Liang:PRep}. Hence, there exists a joint distribution involving both settings of Alice. By the celebrated result of Fine~\cite{Fine:PRL:1982}, the resulting correlation is again local. Since all extreme points of $\Q$ having the same two zeros are local, so must their convex mixtures, which completes the proof of the Lemma.
\end{proof}
Lemma ~\ref{lem1} implies that a nonlocal quantum correlation {\em cannot} have more than one zero in any row or column, or more than four zeros in total in its correlation table. Note, however, that correlation tables having more zeros may not even be realizable quantum mechanically. For example, a correlation table having more than twelve zeros is nonphysical whereas that having exactly eleven zeros must violate the NS conditions of \cref{Eq:NS}.

Hereafter, we focus on characterizing the boundary of $\Q$ where the corresponding correlation table contains four or less zeros. In particular, we are interested to know when each of these cases contain also nonlocal quantum correlations and whether they lead to self-testing statements. Again, given Lemma~\ref{lem1}, we do not consider Classes having more than one zero in any row or column of the correlation table.

\subsection{Four-zero Classes}

Using the freedom of relabeling, it is easy to see that there remain only two other kinds of four-zero correlation tables that are of interest:
\begin{itemize}[leftmargin=50pt]
    \item[Class 4a:] All blocks have exactly one zero each (\cref{tab:class4a})
    \item[Class 4b:] Two non-neighboring blocks each contains two zeros that lie on the opposite corner (\cref{tab:class4b})
\end{itemize}
    
Interestingly, Fritz already showed in Ref.~\cite{Fritz_2011} that correlations from Class 4a necessarily violate the NS conditions of \cref{Eq:NS}, and are thus not quantum realizable. For Class 4b, we prove in~\cref{app1} that all such $\vecP\in\Q$ must be local. Using Fact~\ref{CHSH:ExtremeQ2}, we find the following one-parameter family of strategies realizing such correlations:
\begin{subequations}\label{Eq:4b}
\begin{gather}
    \ket{\psi} = \cos{\theta}\ket{0}\!\ket{0}+\sin{\theta}\ket{1}\!\ket{1},\\
\begin{split}
    &A_0=\sigma_z,  \quad
    A_1=\cos{2\alpha}\,\sigma_z-\sin{2\alpha}\,\sigma_x, \\
    &B_0=\sigma_z, \quad
    B_1=\cos{2\beta}\,\sigma_z-\sin{2\beta}\,\sigma_x,
\end{split}
\end{gather}
\end{subequations}
where $\sigma_z$ and $\sigma_x$ are Pauli matrices, while $\theta = \pm\frac{\pi}{4}$ in accordance with $\beta=\pm\alpha$. Notice, however, that none of the resulting correlations is extremal in $\Q$ or $\L$.

\subsection{Three-zero Classes}

The remaining possibilities of correlation tables with three zeros fall under---after an appropriate relabeling---two distinct Classes:
\begin{itemize}[leftmargin=50pt]
    \item[Class 3a:]\label{3a} Three blocks having one zero (\cref{tab:class3a})
    \item[Class 3b:]\label{3b} One block has two DZs (ADZs) and a non-neighboring block has another zero  (\cref{tab:class3b})
\end{itemize}

We define the sets of quantum correlations corresponding to these boundary Classes as
\begin{equation}\label{Dfn:Class3}
    \begin{split}
        \Q_{3a}\coloneqq\{\vecP\,|&\,\vecP \in \Q ,\\
        &P(0,0|0,0)=P(1,1|1,0)=P(1,1|0,1)=0 \},\\
        \Q_{3b}\coloneqq\{\vecP\,|&\,\vecP \in \Q ,\\
        &P(0,0|0,0)=P(1,1|0,0)=P(1,0|1,1)=0 \}.
    \end{split}
\end{equation}
Strictly, the definition given above for $\Q_{3b}$ also contains Class $4b$ and other more-zero $\vecP\in\L$.  In this regard, note that it is more convenient to parameterize quantum strategies that include those more-zero Classes as special cases. However, in the following discussions for each Class, it should be understood that our primary interest lies in those correlations where the zeros probabilities are only those explicitly specified. The same remarks also hold for our definition of the two-zero and one-zero Classes given in~\cref{Dfn:Class2} below.

\subsubsection{Class 3a}    

For Class 3a, a  two-parameter family of quantum strategies compatible with the positions of the zeros consists of 
\begin{subequations}\label{Eq:Class2b}
\begin{gather}\label{Eq:HardyState}
    \ket{\psi} = \sin\theta\left(\cos{\alpha}\ket{0}-\sin\alpha\ket{1}\right)\!\ket{1}+\cos\theta\ket{1}\!\ket{0},\\
\begin{split}\label{Eq:2b:Observables}	
    A_0=\sigma_z,  \quad
    A_1=\cos{2\alpha}\,\sigma_z-\sin{2\alpha}\,\sigma_x, \\
    B_0=\sigma_z, \quad
    B_1=\cos{2\beta}\,\sigma_z-\sin{2\beta}\,\sigma_x,
\end{split}
\end{gather}
\end{subequations}
where 
\begin{equation}\label{Eq:3aParameters}	
	\theta = \tan^{-1}{\frac{\tan{\beta}}{\sin{\alpha}}}.
\end{equation} 
If we further set
\begin{equation}\label{Eq:HardyParameter}
	\beta=\alpha=\frac{1}{2} \tan ^{-1}\left(-2 \sqrt{\sqrt{5}+2}\right)
\end{equation}
one finds a {\em symmetric} quantum realization~\cite{Rabelo12} equivalent to the one originally proposed by Hardy~\cite{Hardy1993}.
The resulting correlation $\vecPH\in \Q_{3a}$ reads as:
\begin{equation}\label{Eq:vecPH}
    \vecPH = 
    \scalebox{0.95}{\begin{tabular}{cc|cc|cc|}
         &  & \multicolumn{2}{c|}{$x=0$} & \multicolumn{2}{c|}{$x=1$} \\
         &  & $0$ & $1$~ & $0$ & $1$ \\
        \hline
        \multirow{2}{*}{$y=0$} & $0$ & $0$ & $\frac{1-\nu}{2}$ & $\frac{1-3\nu}{2}$ & $\nu$\\
                               & $1$ & $\frac{1-\nu}{2}$~  & $\nu$  & $\frac{1+\nu}{2}$ &$0$ \\ [1.02ex]
        \hline
        \multirow{2}{*}{$y=1$} & $0$ & $\frac{1-3\nu}{2}$  & $\frac{1+\nu}{2}$ & $\frac{1+\nu}{2}$ & $\frac{1-3\nu}{2}$ \\
                               & $1$ & $\nu$ &$0$  & $\frac{1-3\nu}{2}$~ & $\frac{5\nu-1}{2}$\\ [1.02ex] \hline
        \end{tabular}}~,
\end{equation}
where $\nu = \sqrt{5}-2$.
This particular correlation is known to be non-exposed~\cite{Goh2018}. Moreover, from the result of~\cite{Xiang:2011wb}, we know that for quantum correlations having exactly the zeros shown in~\cref{tab:class3a}, its maximum achievable CHSH violation is $\SCHSH=10(\sqrt{5}-2)\approx2.36068$. 

Furthermore, $\vecPH$ has been shown~\cite{Rabelo12} to provide a self-test for the partially entangled two-qubit pure state defined in \cref{Eq:Class2b} -- \cref{Eq:HardyParameter}. Using the SWAP method of~\cite{Yang14,Bancal15}, we find that $\vecPH$ {\em also} self-tests the measurements given in the same equations. In fact, this {\em numerical} method allows us to understand how robust these self-testing results are with respect to deviations from the ideal situation. In~\cref{fig:selftest:class3a}, we show the fidelity lower bound with respect to the reference state of~\cref{Eq:HardyState} when the observed CHSH Bell-inequality violation is suboptimal and/or the zero probability constraints are not strictly enforced. Analogous robustness results for self-testing the target observables are also shown in~\cref{fig:selftest:class3a}.

More recently, a different two-parameter family of correlations from this Class have been provided in~\cite{Rai:PRA:2022} (see also~\cite{Jordan:PRA:1994}). Intriguingly, each correlation from this family again self-tests a partially entangled two-qubit state. Upon a closer inspection, one finds that their family of correlations are in fact equivalent to the ones that result from \cref{Eq:Class2b} and \cref{Eq:3aParameters}. To obtain our family of correlations $\vecP$ from theirs, it suffices to identify their parameters $r$ and $s$ with ours via:
\begin{equation*}
	r= \frac{3-\cos{2\alpha}-2\cos^2{\alpha}\cos{2\beta}}{4},~
	s=\frac{\csc^2{\alpha}\tan^2{\beta}}{1+\csc^2{\alpha}\tan^2{\beta}}
\end{equation*}
followed by the three steps of relabeling: 
(1) $A_1\mapsto -A_1$ and $B_1\mapsto -B_1$ (2) $A_0\leftrightarrow A_1$ and $B_0\leftrightarrow B_1$
(3) swapping Alice and Bob.

\begin{figure}[h!tbp]
\centering
    \captionsetup{justification=RaggedRight,singlelinecheck=off}
  \includegraphics[width=0.9\linewidth]{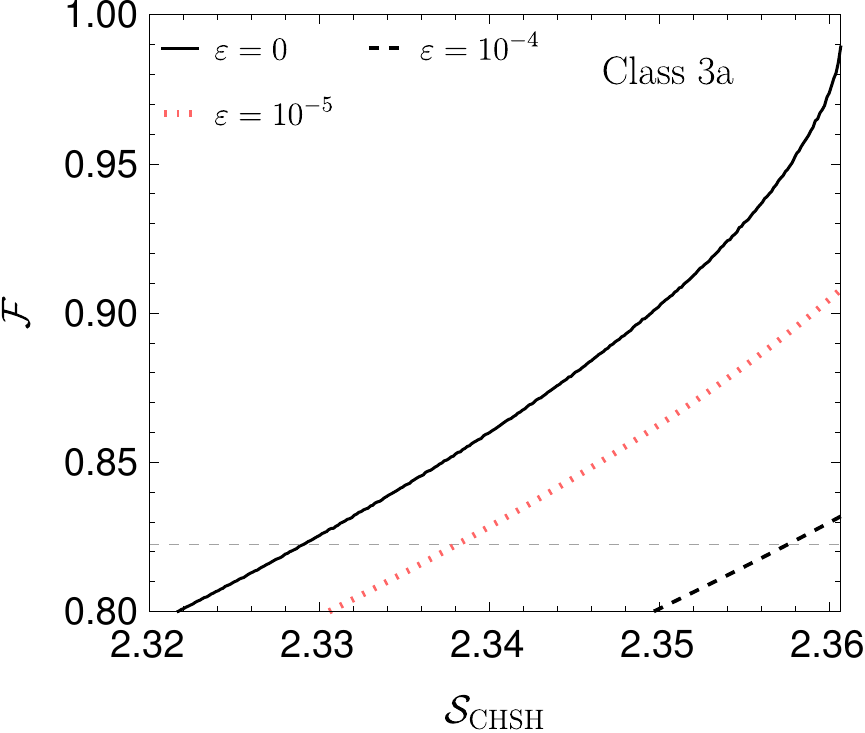}\vspace{0.5cm}
  \includegraphics[width=0.9\linewidth]{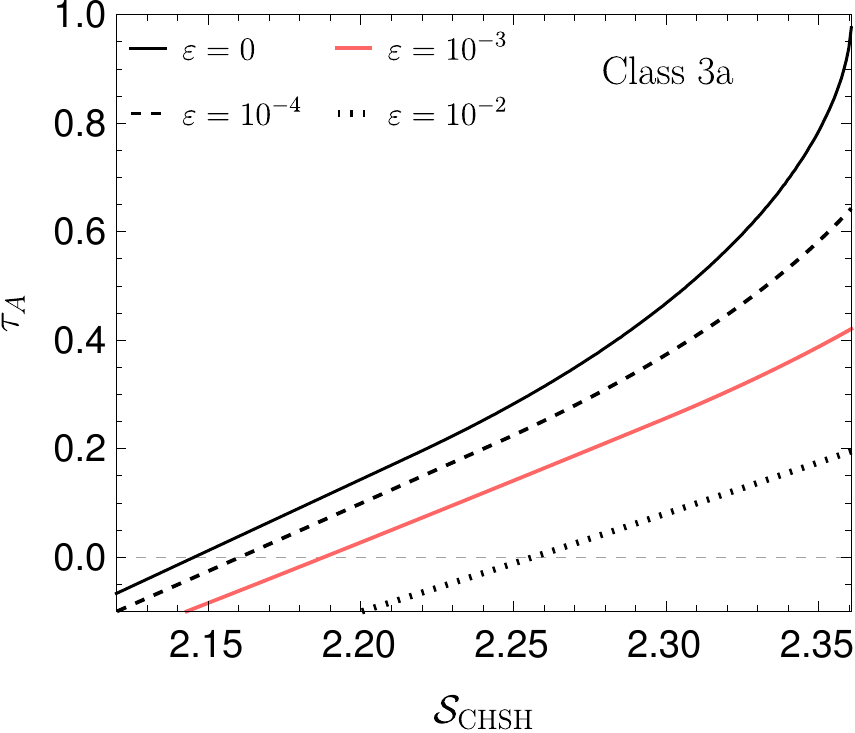}  
  \caption{Plots illustrating the robustness of the self-testing result pertinent to the correlation $\vecPH$ of \cref{Eq:vecPH} from Class 3a. The top and the lower plot show as a function of the Bell value $\SCHSH$, respectively, a lower bound on the relevant figure of merit for the self-testing of the quantum strategy given in \cref{Eq:Class2b} -- \cref{Eq:HardyParameter}. For the self-testing of state, we use the fidelity with respect to the reference state, \cref{Eq:Fidelity}, as the figure of merit. Meanwhile for the self-testing of the measurements, we use the figure of merit defined in \cref{Robust_ST_Meas}. Notice that due to the symmetry in Alice's and Bob's reference measurements, it suffices to show the plot for Alice. Here and below, $\varepsilon$ represents the allowed deviation  from the required zero probability ($\NS$ boundary). Throughout, we use the level-3 outer approximation of $\Q$ due to~\cite{Moroder13} in the fidelity (top) plot but the level-4 outer approximation for the others. Each dashed line marks the respective value for the figure of merit achievable by a classical strategy, i.e., only the region above the horizontal line could give nontrivial self-testing. For details about these figures of merit, see~\cref{App:robust_self-test}.}
  \label{fig:selftest:class3a}
\end{figure}

\subsubsection{Class 3b}    

For Class 3b, a two-parameter family of quantum examples is obtained by considering the two-qubit strategies 
\begin{subequations}\label{Eq:Class3b}
\begin{gather}
\label{Eq:3b:State}
    \ket{\psi} = \cos{\theta}\ket{0}\!\ket{1}+\sin{\theta}\ket{1}\!\ket{0},\\
\label{Eq:3b:Observables}
\begin{split}
    A_0=\sigma_z,\quad
    A_1=\cos{2\alpha}\,\sigma_z-\sin{2\alpha}\,\sigma_x, \\
    B_0=\sigma_z,\quad
    B_1=-\cos{2\beta}\,\sigma_z-\sin{2\beta}\,\sigma_x,
\end{split}
\end{gather}
\end{subequations}
where 
\begin{equation}\label{Eq:3b:Constraint}
	\theta= \tan^{-1}\frac{\tan\alpha}{\tan\beta}.
\end{equation}

In \cref{App:MaxCHSH-Q3b}, we prove that among all quantum correlations from this Class (see \cref{tab:class3b}), the quantum strategy of \cref{Eq:Class3b} with
\begin{subequations}\label{Eq:3b:Parameters}
\begin{equation}\label{Eq:Parameters}
        \alpha = \tan^{-1}\sqrt{\tan{\theta}},\quad \beta = \frac{\pi}{2}-\alpha,
\end{equation}
and
\begin{equation}\label{Eq:theta-tau}
        \theta = \tan^{-1}\tfrac{1}{3}\left( -1-\tfrac{2}{\tau}+\tau \right),\quad \tau  = \sqrt[3]{17+3\sqrt{33}}\\
\end{equation}
\end{subequations}

gives the maximal CHSH  violation 
\begin{equation}\label{Eq:MaxCHSH:Class3b}
	\SCHSH=4 - 4 (2\kappa_1 + \kappa_2)\approx 2.26977. 
\end{equation}
Moreover, the corresponding quantum strategy can be self-tested using the correlation (see \cref{App:Self-test:Class3b} for a proof):

\begin{subequations}\label{Eq:vecPQ}
\begin{equation}
    \vecP_{\text{Q}} := 
    \begin{tabular}{cc|cc|cc|}
         &  & \multicolumn{2}{c|}{$x=0$} & \multicolumn{2}{c|}{$x=1$} \\
         &  & $0$ & $1$~ & $0$ & $1$ \\
        \hline
        \multirow{2}{*}{$y=0$} & $0$ & $0$ & $\frac{1}{2}-\kappa_2$ & $\kappa_1$ & $\kappa_3$\\
                               & $1$ & $\frac{1}{2}+\kappa_2$  & $0$  & $\frac{1}{2}$ &$\kappa_2$ \\ [1.02ex]
        \hline
        \multirow{2}{*}{$y=1$} & $0$ & $\kappa_2$  & $\kappa_3$ & $\frac{1}{2}-\kappa_1$ & $0$ \\
                               & $1$ & $\frac{1}{2}$ &$\kappa_1$  & $2\kappa_1$& $\frac{1}{2}-\kappa_1$\\ [1.02ex]
        \hline
        \end{tabular}
\end{equation}
where 
\begin{equation}\label{Eq_kappa}
\begin{split}
	\kappa_1 &= \frac{1}{2}\left(\frac{\tau^2-\tau-2}{2\tau}\right)^3\approx 0.0804,\\
	\kappa_2 &= \frac{1-\cos(2\tan^{-1}\sqrt{\tan\theta})}{2\sec^2\theta}\approx 0.2718,\\
	\kappa_3 &= \frac{1-2\kappa_1-2\kappa_2}{2}.
\end{split}
\end{equation}
\end{subequations}
The robustness of these self-testing results is illustrated in~\cref{fig:selftest:class3b}. The correlation of \cref{Eq:vecPQ} is also provably non-exposed (see~\cref{App:NonExposed:3b}). 
More generally, for the one-parameter strategies of \cref{Eq:Class3b} and \cref{Eq:3b:Constraint} with $\theta\in[0.0338\pi,0.2350\pi]$, our numerical results based on the SWAP method~\cite{Yang14,Bancal15} give a lower bound of at least 0.99 for both figures of merit, \cref{Eq:Fidelity} and \cref{Robust_ST_Meas}, associated with self-testing. 
 This observation strongly suggests that the resulting correlation also provides a self-test for the corresponding quantum strategies. Also worth noting is that $\vecP$ in this one-parameter family only violates inequality~\cref{Eq:CHSH} for $\theta\in(0,\frac{\pi}{4})$ but for $\theta\in(\frac{\pi}{4},\frac{\pi}{2})$, it violates the one corresponding to the winning condition $(x\oplus1)(y\oplus1)=a\oplus b$.

Before moving to other Classes with fewer zeros, we note that the non-exposed nature of $\vecP_\text{Hardy}$ and $\vecPQ$ can be illustrated using a projection plot (see also Figure 5 of~\cite{Goh2018}). 
\cref{fig:1sub1} and~\cref{fig:1sub2} show, respectively, a 2-dimensional projection plot of  $\NS$, $\L$, and an outer approximation of $\Q$ (due to~\cite{Moroder13}) on a plane with its axes given by $\SCHSH$ and a linear sum of probabilities that are expected to vanish for the two Classes.\footnote{To determine the said outer approximation of the upper (lower) boundary of $\Q$ on these plots, we compute, for each {\em fixed} value of the Bell function specified on the horizontal axis, the maximum (minimum) CHSH value over all $\vecP\in\tilde{\Q}_\ell$ where $\tilde{\Q}_\ell$ is the level-$\ell$ approximation of $\Q$ due to~\cite{Moroder13}.} 
For comparisons, we also mark on these plots the position of $\vecP_\text{PR}$ [\cref{Eq:vecPR}], the relabeled ($a=0\leftrightarrow a=1$) version of $\vecP_\text{PR}$ denoted by $\vecP_\text{PR,2}$, and their noisy versions giving the maximal CHSH violation in quantum theory:
\begin{equation}\label{Eq:3vecP}
\begin{split}
    \vecP_\text{PR,2}&\overset{a=0\leftrightarrow a=1}{=} \vecP_\text{PR},\\
    \vecP_{\text{CHSH}}&=\frac{1}{\sqrt{2}}\vecP_\text{PR}+(1-\frac{1}{\sqrt{2}})\vecP_{\mathbb{I}},\\
    \vecP_{\text{CHSH,2}}&=\frac{1}{\sqrt{2}}\vecP_\text{PR,2}+(1-\frac{1}{\sqrt{2}})\vecP_{\mathbb{I}},
\end{split}
\end{equation}
where $\vecP_{\mathbb{I}}$ is the uniform distribution, i.e., $P_{\mathbb{I}}(a,b|x,y)=\frac{1}{4}$ for all $a,b,x,y$.

\begin{figure}[h!tbp]
\captionsetup{justification=RaggedRight,singlelinecheck=off}
\centering
\begin{subfigure}{0.45\textwidth}
\centering
  \includegraphics[width=0.9\linewidth]{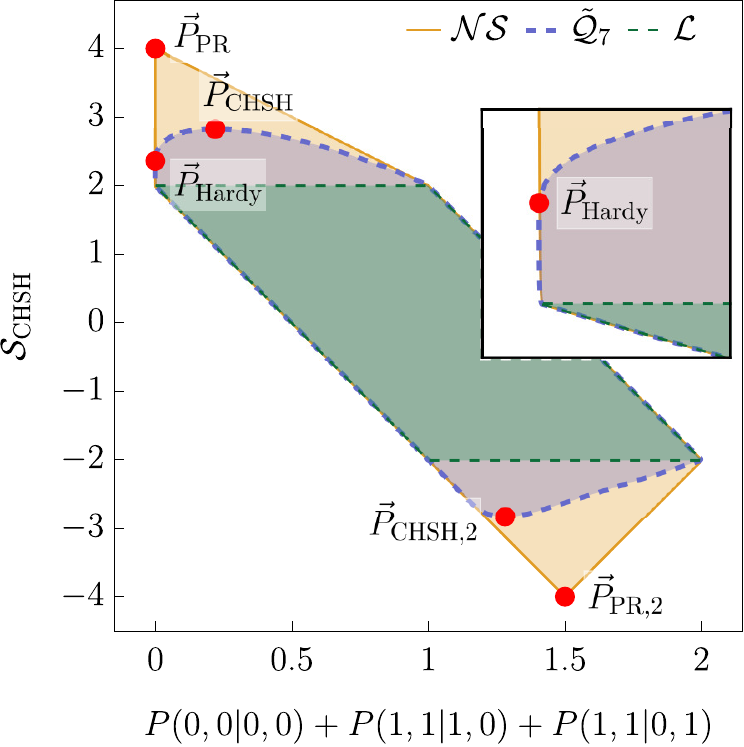}
  \caption{A projection plot illustrating $\Q_{3a}$, which consists of all quantum correlations satisfying $P(0,0|0,0)+P(1,1|1,0)+P(1,1|0,1)=0$. These correlations lie on the vertical line where the value of the horizontal axis is zero. Among them, $\vecPH$ specified in \cref{Eq:vecPH} gives the maximal Bell violation. 
  }
  \label{fig:1sub1}
\end{subfigure}
\begin{subfigure}{0.45\textwidth}
\centering
  \includegraphics[width=0.9\linewidth]{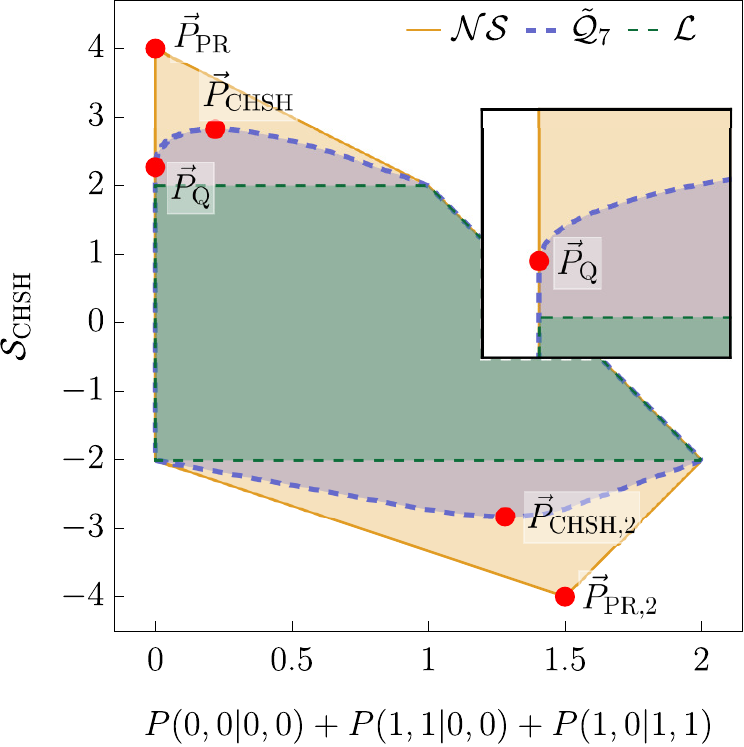}
  \caption{A projection plot illustrating $\Q_{3b}$, which consists of all quantum correlations satisfying $P(0,0|0,0)+P(1,1|0,0)+P(1,0|1,1)=0$. These correlations lie on the vertical line where the value of the horizontal axis is zero. Among them, $\vecPQ$ specified in \cref{Eq:vecPQ} gives the maximal Bell violation and is a self-test for the quantum strategy of \cref{Eq:Class3b} with \cref{Eq:3b:Parameters}. }
  \label{fig:1sub2}
\end{subfigure}
\caption{Two-dimensional projections of $\NS$, $\L$, and an outer approximation~\cite{Moroder13} of $\Q$ onto the plane labeled by the CHSH value $\SCHSH$, cf.~\cref{Eq:CHSH}, and some linear sum of probabilities. When the horizontal parameter is equal to zero, we recover the configurations specified, respectively, in the correlation tables of~\cref{tab:class3a} and ~\cref{tab:class3b}.  The insets illustrate, correspondingly, the non-exposed nature of both $\vecPH$ and $\vecPQ$. The explicit form of the other $\vecP$'s shown are given in \cref{Eq:vecPR} and \cref{Eq:3vecP}.
}\label{fig:1}
\end{figure}

\subsection{Classes having less then three zeros} \label{sec3c}

By exploiting the freedom of relabeling, it is easy to see that there remain three two-zeros Classes and a one-zero Class that can be described as follows:
\begin{itemize}[leftmargin=55pt]
    \item[Class 2a:] There are two DZs (or ADZs) in the same block (\cref{tab:class2a})
    \item[Class 2b:] The two zeros appear in two neighboring blocks having the same value of $x$ (or $y$) (\cref{tab:class2b})
    \item[Class 2c:] The two zeros appear in two non-neighboring blocks (\cref{tab:class2c})
    \item[Class 1~~:] There is only one zero somewhere in the table (\cref{tab:class1})
\end{itemize}
Accordingly, the sets of quantum correlations exhibiting these structures of zeros are defined\footnote{See the remarks given in the paragraph immediately after \cref{Dfn:Class3}.} as
\begin{align}\label{Dfn:Class2}
        \Q_{2a}\coloneqq\{\vecP\,|&\,\vecP \in \Q,  P(0,0|0,0)=P(1,1|0,0)=0 \},\nonumber\\
        \Q_{2b}\coloneqq\{\vecP\,|&\,\vecP \in \Q,  P(0,0|0,0)=P(1,1|1,0)=0 \},\nonumber\\
        \Q_{2c}\coloneqq\{\vecP\,|&\,\vecP \in \Q, P(0,0|0,0)=P(1,0|1,1)=0 \},\nonumber\\
        \Q_{1}\coloneqq\{\vecP\,|&\,\vecP \in \Q,  P(0,0|0,0)=0 \}.
\end{align}

\subsubsection{Class 2a}
\label{Sec:Class2a}

Interestingly, with \cref{Eq:Class3b}
, we immediately get a three-parameter family of strategies for realizing quantum correlations in Class 2a. With the change of basis via  $\sigma_z\otimes\sigma_x$, we arrive at
\begin{subequations}\label{Eq:Class2a}
\begin{gather}
\label{Eq:2a:State}
    \ket{\psi} = \cos{\theta}\ket{0}\!\ket{0}-\sin{\theta}\ket{1}\!\ket{1},\\
    \label{Eq:2a:Observables}
 \begin{split}
    A_0=\sigma_z,\quad
    A_1=\cos{2\alpha}\,\sigma_z+\sin{2\alpha}\,\sigma_x, \\
    B_0=-\sigma_z,\quad
    B_1=\cos{2\beta}\,\sigma_z -\sin{2\beta}\,\sigma_x.
 \end{split}
\end{gather}
Together with the choice of 
\begin{equation}
	\theta=-\frac{\pi}{4},\quad \alpha = \frac{5\pi}{6},\quad \text{and}\quad \beta = -\frac{2\pi}{3},
\end{equation}
\end{subequations}
which corresponds to the Bell state $\ket{\Phi_2^+}$, we get
\begin{equation}\label{Eq:PQ2}
    \vecPQii := ~
    \begin{tabular}{cc|cc|cc|}
         &  & \multicolumn{2}{c|}{$x=0$} & \multicolumn{2}{c|}{$x=1$} \\
         &  & $0$ & $1$~ & $0$ & $1$ \\
        \hline
        \multirow{2}{*}{$y=0$} & $0$ & $0$ & $\frac{1}{2}$ & $\frac{1}{8}$ & $\frac{3}{8}$\\
                               & $1$ & $\frac{1}{2}$  & $0$  & $\frac{3}{8}$ &$\frac{1}{8}$ \\ [1.02ex]
        \hline
        \multirow{2}{*}{$y=1$}& $0$ & $\frac{1}{8}$ & $\frac{3}{8}$  & $\frac{3}{8}$ & $\frac{1}{8}$ \\
                               & $1$&$\frac{3}{8}$& $\frac{1}{8}$  & $\frac{1}{8}$& $\frac{3}{8}$\\ [1.02ex]
        \hline
        \end{tabular}~,
\end{equation}
which gives the maximal CHSH value  $\SCHSH=2.5$ among all  $\vecP\in\Q_{2a}$. The proof is completely analogous to that given in \cref{App:MaxCHSH-Q3b} for $\vecP\in\Q_{3b}$ and is thus omitted.\footnote{In this case, the fact that $\SCHSH\le2.5$ also follows directly from Theorem B.2 of~\cite{Quintino:JPA:2012}.}

Note that the correlation $\vecPQii$ satisfies $\Exp{A_0B_0} = -1$, $\max\{|\Exp{A_1B_1},|\Exp{A_0B_1}|,|\Exp{A_1B_0}|\} < 1$, and saturates the inequality:
\begin{equation}
\begin{split}
	&2\Exp{A_1B_1}\Exp{A_1B_0}\Exp{A_0B_1}+\Exp{A_1B_1}^2\\
	+&\Exp{A_1B_0}^2 + \Exp{A_0B_1}^2 \le 1.
\end{split}
\end{equation}
Hence, from~\cite{Wang2016} (see also the paragraphs below Eq. (23) of~\cite{Goh2018}), the nonlocal correlation $\vecPQii$ self-tests~\cite{Supic19} the Bell state $\ket{\Phi_2^+}$ and  is an extreme point. The robustness of this self-testing result and that for the quantum observables of \cref{Eq:Class2a} are illustrated in \cref{fig:selftest:class2a}. Indeed, by following an analysis similar to those presented in~\cref{app-st2}, it can be shown that $\vecPQii$ also self-tests \cref{Eq:Class2a}. More generally, for the one-parameter family of quantum strategies defined by \cref{Eq:2a:State} -- \cref{Eq:2a:Observables}, $\theta=-\frac{\pi}{4}$, $\beta=-\alpha+\frac{\pi}{6}$, and $\alpha\in(0,\frac{\pi}{6})\cup(\frac{\pi}{6}, \frac{\pi}{2})\cup(\frac{\pi}{2},\frac{2\pi}{3})\cup(\frac{2\pi}{3}, \pi)$, the resulting correlations are easily verified to be nonlocal and self-test a Bell pair since they satisfy the same set of sufficient criteria given above. Also worth noting is that for $\alpha\in[0.0081\pi,0.1513\pi]\cup[0.1822\pi,0.4895\pi]\cup[0.5113\pi,0.6598\pi]\cup[0.6752\pi, 0.9908\pi]$, our numerical results based on the SWAP method~\cite{Yang14,Bancal15} give a lower bound of 0.99 for the figure of merit associated with measurement-self-testing, \cref{Robust_ST_Meas}.
In~\cref{Non-Exposed_Class.2a}, we further show that $\vecPQii$ is non-exposed.

\subsubsection{Class 2b}

For Class 2b, a three-parameter family of quantum realization can be obtained using the strategy specified in \cref{Eq:Class2b}.

For $\theta=\frac{\pi}{4}$, the resulting correlations with $\beta=\frac{\alpha}{2},\pi-\frac{\alpha}{2}$ are local. However, if we set 
\begin{equation}\label{Eq:2b:Parameters}
	\alpha = \frac{\pi}{6}\quad \text{and}\quad \beta =\theta=\frac{\pi}{4},
\end{equation} 
we get
\begin{equation}\label{Eq:PQ3}
    \vecPQiii := ~
    \begin{tabular}{cc|cc|cc|}
         &  & \multicolumn{2}{c|}{$x=0$} & \multicolumn{2}{c|}{$x=1$} \\
         &  & $0$ & $1$~ & $0$ & $1$ \\
        \hline
        \multirow{2}{*}{$y=0$} & $0$ & $0$ & $\frac{1}{2}$ & $\frac{1}{8}$ & $\frac{3}{8}$\\
                               & $1$ & $\frac{3}{8}$  & $\frac{1}{8}$  & $\frac{1}{2}$ &$0$ \\ [1.02ex]
        \hline
        \multirow{2}{*}{$y=1$}& $0$ & $\frac{3}{16}$ & $\frac{9}{16}$  & $\frac{9}{16}$ & $\frac{3}{16}$ \\
                               & $1$&$\frac{3}{16}$& \hl{$\frac{1}{16}$}  & $\frac{1}{16}$& \hl{$\frac{3}{16}$}\\ [1.02ex]
        \hline
        \end{tabular}~,
\end{equation}
which we show in~\cref{App:MaxCHSH-Q2b} to give the maximal CHSH value of $\SCHSH=2.5$ among all $\vecP\in\Q_{2b}$.

Notice that a generalization of Hardy's argument~\cite{Hardy1993} can be given~\cite{NewArgument} for nonlocal correlations in Class 2b using
\begin{equation}
    \begin{split}
        &P(0,0|0,0)=0,\quad
        ~P(1,1|1,0)=0,\\
        &P(1,1|0,1)=q_1,\quad
        P(1,1|1,1)=q_2,
    \end{split}
\end{equation}
and a success probability defined by $D=q_2-q_1$, i.e., the difference between the two highlighted entries in \cref{Eq:PQ3}. This argument has a different logical structure from the so-called Cabello's argument formulated in~\cite{Liang:2005vl} (inspired by~\cite{Cabello02}, see also~\cite{KCA+06}). In particular, it results in an even higher maximal success probability of $D=0.125$. To this end, notice from \cref{normalization} and \cref{Eq:NS} that we can rewrite the left-hand-side of \cref{Eq:CHSH} as $\SCHSH=4D+2-P(0,0|0,0)-P(1,1|1,0)$. Since the last two terms vanish for $\vecP \in \Q_{2b}$, we thus see that for such correlations, the CHSH violation and the success probability are related linearly  by $\SCHSH=4D+2$. This implies that the maximal value of $D$ can be achieved only by $\vecPQiii$.

In ~\cref{App:Self-test:Class2b}, we show further that $\vecPQiii$ self-tests  the quantum strategy of \cref{Eq:Class2b} and \cref{Eq:2b:Parameters}. 
The robustness of this self-testing result is illustrated in \cref{fig:selftest:class2b}. 
More generally, for the one-parameter family strategies of \cref{Eq:Class2b} with $\alpha = \frac{\pi}{6}$ and $\beta =\theta\in[0.0510,0.4614\pi]$, our numerical results obtained from the SWAP method~\cite{Yang14,Bancal15} give a lower bound of at least 0.99 for both figures of merit, \cref{Eq:Fidelity} and \cref{Robust_ST_Meas}, associated with self-testing. Again, this observation suggests that the resulting correlation also provides a self-test for the respective quantum strategies.

In~\cref{Non-Exposed_Class.2b}, we also show that $\vecPQiii$ is non-exposed. For the advantage of this nonlocality-without-inequality argument over the others, see \cite{NewArgument} by some of the present authors.

\subsubsection{Class 2c}

For Class 2c, a three-parameter family of quantum strategies is given by:
\begin{subequations}\label{Eq:Class1-0}
\begin{gather}
    \ket{\psi} = \cos{\phi}(\cos{\theta}\ket{0}\!\ket{1}\!+\!\sin{\theta}\ket{1}\!\ket{0})\!+\!\sin{\phi}\ket{1}\!\ket{1},\\
\begin{split}
    A_0=\sigma_z,  \quad
    A_1=\cos{2\alpha}\,\sigma_z-\sin{2\alpha}\,\sigma_x, \\
    B_0=\sigma_z, \quad
    B_1=\cos{2\beta}\,\sigma_z-\sin{2\beta}\,\sigma_x,
\end{split}
\end{gather}
\end{subequations}
where 
\begin{equation}\label{Eq:phi}
	\phi = \tan^{-1}\left(\frac{\sin\theta}{\tan\beta}-\tan\alpha\cos\theta\right).
\end{equation}
{\noindent}Interestingly, the distribution of zeros in this Class coincides with that required for the so-called Cabello's argument. Moreover, the following quantum correlation~\cite{KCA+06} in $\Q_{2c}$ 
\begin{equation}\label{Eq:PCabello}
    \vecPC \approx
    \scalebox{0.9}{\begin{tabular}{cc|cc|cc|}
         &  & \multicolumn{2}{c|}{$x=0$} & \multicolumn{2}{c|}{$x=1$} \\
         &  & $0$ & $1$~ & $0$ & $1$ \\
        \hline
        \multirow{2}{*}{$y=0$} & $0$ & $0$ & $0.2770$ & $0.1444$ & $0.1326$\\
                               & $1$ & $0.6410$  & $0.0820$  & $0.5786$ &$0.1444$ \\
        \hline
        \multirow{2}{*}{$y=1$}& $0$ & $0.3068$ & $0.3342$  & $0.6410$ & $0$ \\
                               & $1$ & $0.3342$& $0.0247$  & $0.0820$& $0.2770$\\
        \hline
        \end{tabular}}
\end{equation}
gives the maximal probability of success for this argument. Again, from~\cite{Xiang:2011wb}, we know that this also gives the maximal CHSH violation of $\SCHSH\approx 2.4312$ among all $\vecP\in\Q_{2c}$. 

The correlation $\vecPC$ can be obtained using \cref{Eq:Class1-0} with  
\begin{subequations}\label{Eq:2c:AllParameters}
\begin{equation}\label{Eq:2c:Parameters}
	\beta=\frac{\pi}{2}-\alpha \quad\text{and}\quad \theta = \tan^{-1}{\left ( \frac{k_2^2-k_1k_3}{k_2^2+k_3^2}\right )},
\end{equation}
where 
\begin{equation}\label{Eq:k}
    \begin{split}
        k_1 &= \tfrac{1}{6}\left(4-\tfrac{\mu_1^2+1}{\mu_1}\right),\quad k_2  = \sqrt{1-\tfrac{31 \sqrt[3]{36}}{12\mu_2}+\tfrac{\sqrt[3]{6}\mu_2}{12}},\\
        k_3 &= \tfrac{1}{6}\left(\tfrac{\mu_3^2-29}{\mu_3}-2\right),\quad
        \alpha =\tan^{-1}\sqrt{\tfrac{\mu_+ + \mu_--1}{12}},
    \end{split}
\end{equation}
\end{subequations}
$\mu_1=\sqrt[3]{53-6\sqrt{78}}$, $\mu_2=\sqrt[3]{67\sqrt{78}-414}$, $\mu_3=\sqrt[3]{307+39\sqrt{78}}$, and $\mu_\pm = \sqrt[3]{359\pm12\sqrt{78}}$.

Recently, $\vecPC$ was shown~\cite{Rai:PRA:2021} to self-test the partially entangled state:
\begin{subequations}\label{Eq:Class2c}
\begin{equation}\label{Eq:CabelloState}
    \ket{\psi} = k_{1}\ket{0}\!\ket{0}+k_{2}\ket{0}\!\ket{1}+k_{2}\ket{1}\!\ket{0})+k_{3}\ket{1}\!\ket{1},
\end{equation}
which may be obtained from \cref{Eq:Class1-0} via the local unitary $i \sigma_y\otimes(\sin{\beta}\,\mathbb{1}+i\cos{\beta}\,\sigma_y)$. For completeness, note that the same change of basis results in the following observables:
\begin{equation}\label{Eq:2c:Observables}
\begin{split}
    A_0=-\sigma_z,\quad
    A_1=-\cos{2\alpha}\,\sigma_z+\sin{2\alpha}\,\sigma_x, \\
    B_0=\cos{2\alpha}\,\sigma_z-\sin{2\alpha}\,\sigma_x,\quad
    B_1=-\sigma_z,
\end{split}
\end{equation}
\end{subequations}
Using the SWAP method of~\cite{Yang14,Bancal15}, we can further show that $\vecPC$  self-tests the measurements of~\cref{Eq:2c:Observables} with the parameters specified in \cref{Eq:k}.
The robustness of this self-testing result is illustrated in \cref{fig:selftest:class2c}.
More generally, for the one-parameter family of strategies defined by \cref{Eq:Class1-0}, with $\theta$ given in \cref{Eq:2c:Parameters}, $\beta = \frac{\sin\theta-\cos\theta}{\tan\phi}$, $\alpha=\frac{\pi}{2}-\beta$, and  $\phi \in[0.0034\pi,0.1208\pi]$, our numerical results based on the SWAP method give a lower bound of at least 0.99 for both figures of merit, \cref{Eq:Fidelity} and \cref{Robust_ST_Meas}, associated with self-testing, strongly suggesting that these correlations are self-tests.
In~\cref{Non-Exposed_Class.2c} we further show that $\vecPC$ is non-exposed. 

\subsubsection{Class 1}

Finally, for Class 1, a four-parameter family of quantum realization is specified by \cref{Eq:Class1-0}. Among them, one obtains the following permutationaly-invariant strategy by setting $\theta\to\frac{\pi}{4}$, $\phi\to\frac{\pi}{2}-\theta$, $\alpha\to - \frac{\alpha}{2}$, $\beta\to - \frac{\alpha}{2}$:
\begin{subequations}\label{Eq:Class1}
\begin{gather}\label{Eq:StateClass1}
    \ket{\psi} = \sin{\theta}\left(\frac{\ket{0}\!\ket{1}+\ket{1}\!\ket{0}}{\sqrt{2}}\right)+\cos{\theta}\ket{1}\!\ket{1},\\
\label{Eq:1:Observables}
        A_1 = B_1 = \sigma_z, \quad
        A_2 = B_2 = \cos{\alpha} \,\sigma_z + \sin{\alpha}\,\sigma_x.
\end{gather}
Of particular interest is the choice of
\begin{equation}
	\theta =  \cos^{-1}(\sqrt{\xi_{2}}) \approx 0.4267\pi,\quad \alpha = \cos^{-1}(\xi_1)\approx 0.4076\pi,
\end{equation} 
\end{subequations}
with $\xi_1 \approx 0.2862, \xi_2 \approx 0.0521$, and $\xi_3 \approx 2.6353$ being, respectively, the smallest positive roots of the cubic polynomials:
\begin{align}
    p_1(x) &= x^3-9 x^2-x+1, \nonumber \\
    p_2(x) &= 7 x^3-35 x^2+21 x-1, \label{Eq:Poly}\\ 
    p_3(x) &= x^3+26 x^2-36 x-104. \nonumber
\end{align}
 
This quantum strategy gives the correlation
\begin{equation}\label{Eq:PQ4}
    \vecPQiv \approx ~
    \scalebox{0.9}{\begin{tabular}{cc|cc|cc|}
         &  & \multicolumn{2}{c|}{$x=0$} & \multicolumn{2}{c|}{$x=1$} \\
         &  & $0$ & $1$~ & $0$ & $1$ \\
        \hline
        \multirow{2}{*}{$y=0$} & $0$ & $0$ & $0.4740$ & $0.1692$ & $0.3048$\\
                               & $1$ & $0.4740$  & $0.0521$  & $0.4740$ &$0.0521$ \\
        \hline
        \multirow{2}{*}{$y=1$}& $0$ & $0.1692$ & $0.4740$  & $0.5492$ & $0.0939$ \\
                               & $1$ & $0.3048$& $0.0521$  & $0.0939$& $0.2630$\\
        \hline
        \end{tabular}}~,
\end{equation}
which provably gives the maximal CHSH violation of $\SCHSH = \xi_3$ among all $\vecP\in\Q_1$. In~\cref{App:Self-test:Class1}, we prove that $\vecPQiv$ can be used to self-test the strategy of \cref{Eq:Class1}. The robustness of this  self-testing result is shown in \cref{fig:selftest:class1}. In~\cref{Non-Exposed_Class.1}, we show further that $\vecPQiv$ is non-exposed.

In \cref{tab:SummaryReduction}, we summarize how our examples of quantum strategies leading to fewer zeros reduce to those with more zeros when additional constraints are imposed.

\begin{table*}
\small\addtolength{\tabcolsep}{-5pt}
          \captionsetup{justification=RaggedRight}
    \scalebox{0.92}{\begin{tabular}{|c|c|c||c|c||c|}
    \hline
         Classes & Examples & Parameters & Classes   & Examples & Additional constraint(s)/ transformation \\ \hline\hline
        1 (\cref{tab:class1}) & \cref{Eq:Class1-0} & 4 & 2a (\cref{tab:class2a}) &  \cref{Eq:Class3b} & $\phi\to0$ and $2\beta\mapsto \pi-2\beta$    \\ \hline 
        1 (\cref{tab:class1}) & \cref{Eq:Class1-0} & 4 &2b (\cref{tab:class2b}) & \cref{Eq:Class2b} &  $\theta \mapsto \tan^{-1}\frac{\cot{\theta}}{\cos{\alpha}} $ and $\phi \mapsto \sin^{-1}{\left( -\sin{\theta}\sin{\alpha} \right)}$
        \\ \hline 
        1 (\cref{tab:class1}) & \cref{Eq:Class1-0} & 4 &2c (\cref{tab:class2c}) & \cref{Eq:Class1-0} + \cref{Eq:phi} &  \cref{Eq:phi}\\ \hline 
        2a (\cref{tab:class2a}) & \cref{Eq:Class3b} & 3 & 3b (\cref{tab:class3b}) & \cref{Eq:Class3b} + \cref{Eq:3b:Constraint} & \cref{Eq:3b:Constraint}\\ \hline 
        2b (\cref{tab:class2b}) & \cref{Eq:Class2b} & 3 & 3a (\cref{tab:class3a}) &  \cref{Eq:Class2b} + \cref{Eq:3aParameters} & \cref{Eq:3aParameters}\\ \hline 
        2c (\cref{tab:class2a}) & \cref{Eq:Class1-0} + \cref{Eq:phi} & 3 & 3b (\cref{tab:class3b}) & \cref{Eq:Class3b} + \cref{Eq:3b:Constraint} & $\phi\mapsto 0$ and $2\beta\mapsto \pi-2\beta$\\ \hline 
        3b (\cref{tab:class3b}) & \cref{Eq:Class3b} + \cref{Eq:3b:Constraint} & 2 & 4b (\cref{tab:class4b}) & \cref{Eq:4b} & $B_0\mapsto -B_0$ followed by $\mathbb{1}_2\otimes\sigma_x$ and $\beta\mapsto\pm\alpha$\\ \hline 
\end{tabular}}
    \caption{\label{tab:SummaryReduction} Summary of how our multi-parameter family of examples for various Classes reduces to those with more zeros. From left to right, we list, respectively, the Classes with fewer zeros, the equation number(s) corresponding to the multi-parameter examples, the number of free parameters involved, the Classes with more zeros to which they reduce, the corresponding equation number(s), and how this reduction is achieved by imposing additional constraint(s) alongside other auxiliary transformation(s).}
\end{table*}

\section{Correlations from finite-dimensional maximally entangled states}~\label{sec4}

Consider now $\M$, the subset of $\Q$ due to finite-dimensional maximally entangled states $\{\ket{\Phi^+_d}\}_{d=2,3,\ldots}$, see~\cref{Sec:Sets}. To understand when the boundary of $\M$ meets the boundary of $\NS$, we first derive the analog of Fact~\ref{CHSH:ExtremeQ} for $\M$.

\begin{lemma}\label{CHSH:Extreme:MES}
In the CHSH Bell scenario, all extreme points of $\M$ are achievable by measuring a two-qubit maximally entangled state with projective measurements.
\end{lemma} 

\begin{proof}
Let us begin by simplifying the notation and denote the POVM elements of Alice and Bob, respectively, by $\{E_{a|x}\}_{a,x=0,1}$ and $\{F_{b|y}\}_{b,y=0,1}$. Then as mentioned in \cref{Eq:MESCorrelation}, the correlation due to the maximally entangled two-qudit state $\ket{\MESd}$ can be written as
\begin{equation}\label{Eq:MESCorrelation2}
    P(a,b|x,y)=\bra{\MESd}\!E_{a|x}\otimes F_{b|y}\! \ket{\MESd}=\frac{\tr(E_{a|x}\tp F_{b|y})}{d}.
\end{equation}
Since we are concerned with a two-outcome Bell scenario, it follows from Theorem 5.4 of~\cite{Cleve:IEEE:2004} (see also~\cite{Liang:PRA:2007}) that it suffices to consider projective measurements for correlations $\vecP$ lying on the boundary of $\M$. Moreover, by the Lemma of~\cite{Masanes06}, there exists a local orthonormal basis such that all four projectors $\{E_{a|x}\}_{a,x=0,1}$ are simultaneously block diagonal, with blocks of size $2\times 2$ or $1\times 1$, i.e.,\footnote{\label{fn:DirectSum}Here and below, we use the direct sum symbol $\oplus$ to emphasize the block diagonal structure of the corresponding operators. Strictly, $\oplus$ should be replaced by a regular sum $\sum$ and each $\Pi_i$, $E\tp_{a|x}$, $F_{b|y}$, $E^{(i)\,\text{\tiny T}}_{a|x}$, $F^{(i)}_{b|y}$ should be thought of as a full $d$-dimensional matrix.}
\begin{equation}\label{Eq:BlockDecomposition}
   E\tp_{a|x}=\bigoplus_{i=1}^{\frac{d+k}{2}}\Pi_iE\tp_{a|x}\Pi_i \quad\forall\,\, a,x=0,1
\end{equation}
where $\Pi_i$ is the projection operator onto the $i$-th block and $k$ ($\frac{d-k}{2}$) is the number of $1\times 1$ ($2\times 2$) blocks.

Using \cref{Eq:BlockDecomposition} in \cref{Eq:MESCorrelation2} and the fact that the block projection operator satisfies $(\Pi_i)^2=\Pi_i$, we get
\begin{equation}\label{Eq:MESCorrelation3}
    \begin{split}
        P(a,b|x,y)&=\frac{1}{d}\tr\left(\oplus_{i=1}^{\frac{d+k}{2}}\Pi_iE\tp_{a|x}\Pi_iF_{b|y}\right)\\
        &=\frac{1}{d}\sum_{i=1}^{\frac{d+k}{2}}\tr(\Pi_iE\tp_{a|x}\Pi_iF_{b|y})\\
        &=\sum_{i=1}^{\frac{d+k}{2}}\frac{2}{d}\frac{1}{2}\tr(E^{(i)\,\text{\tiny T}}_{a|x}F^{(i)}_{b|y})
    \end{split}
\end{equation}
where $E^{(i)\,\text{\tiny T}}_{a|x}:=\Pi_iE\tp_{a|x}\Pi_i$ and $F^{(i)}_{b|y}:=\Pi_iF_{b|y}\Pi_i$ are easily seen to define $d$-dimensional POVMs,
with $d=1$ or $2$.

Whenever $\Pi_i$ projects onto a two-dimensional subspace, 
\begin{equation}
	 \frac{1}{2}\tr(E^{(i)\,\text{\tiny T}}_{a|x}F^{(i)}_{b|y})=\bra{\MESqb}E^{(i)}_{a|x}\otimes F^{(i)}_{b|y} \ket{\MESqb}
\end{equation}
is a correlation in $\M_2$ attainable by a two-qubit maximally entangled $\ket{\MESqb}$ using local POVM  $\{E^{(i)}_{a|x}\}$ and $\{F^{(i)}_{b|y}\}$. If instead, $\Pi_i$ is a one-dimensional projector, the corresponding $E^{(i)}_{a|x}$ and $F^{(i)}_{b|y}$ are real numbers lying in $[0,1]$, so does the expression $\tr(E^{(i)\,\text{\tiny T}}_{a|x}F^{(i)}_{b|y})$. The latter can again be reproduced using $\ket{\MESqb}$ and POVM elements that are proportional to the identity operator. To complete the proof, we recall from~\cite{DAriano:2005tg} that for qubit measurements, all extremal POVMs are projectors. All in all, we thus see that {\em any} correlation in $\M$ can always be written as a convex mixture of correlations attainable by performing projective measurements on $\ket{\MESqb}$. In other words, extreme points of $\M$ in this Bell scenario originate from performing projective measurements on $\ket{\MESqb}$.
\end{proof}

\begin{corollary}\label{CHSH:Decompose:MES}
In the CHSH Bell scenario, any $\vecP\in\M_d$ obtained by performing projective measurements on  $\ket{\MESd}$ has a convex decomposition using at most $\frac{d-1}{2}$ nonlocal $\vecP'\in\M_2$ when $d$ is odd, but $\frac{d}{2}$ nonlocal $\vecP'\in\M_2$ when $d$ is even.
\end{corollary} 

\begin{proof}
From the proof of Lemma~\ref{CHSH:Extreme:MES}, we see that any $\vecP\in\M_d$ obtained by performing projective measurements on  $\ket{\MESd}$ admits the decomposition:
\begin{equation}\label{Eq:MES:vecP:Decomposition}
	\vecP=\frac{2}{d}\sum_{i=1}^{\frac{d-k}{2}}\vecP'_i+\frac{1}{d}\sum_{j = 1}^{k}\vecP''_j
\end{equation}
where $k$ is the number of $1\times 1$ block in the decomposition of $\{E\tp_{a|x}\}_{a,x=0,1}$, $\vecP'_i\in\M_2$ for all $i$, and $\vecP''_j\in\L$ for all $j$. Clearly, only the $\vecP'_i$ can be nonlocal. To complete the proof, note that for $d$ odd, $k\ge 1$ whereas for $d$ even, $k\ge0$.
\end{proof}

Using Lemma~\ref{CHSH:Extreme:MES}, one can show also the following corollary concerning the maximal CHSH violation by $\ket{\MESd}$ (see \cref{App:MaxCHSH:MES} for a proof).
\begin{corollary}\label{Prop:MaxCHSH:MES}
The maximal CHSH violation by the maximally entangled two-qudit state $\ket{\Phi^+_d}$ is:
\begin{equation}
    \SCHSH^\text{max}(\ket{\Phi^+_d}) =\left\{ 
        \begin{aligned}
            ~&2\sqrt{2}\left(\frac{d-1}{d}\right)+\frac{2}{d},& &d = \text{odd}\\
            ~&2\sqrt{2},&&d = \text{even}.
        \end{aligned}
    \right .
\end{equation}    
\end{corollary} 

We now return to the problem of characterizing when the boundaries of $\M$ meet  the boundaries of $\NS$. For general quantum correlations, Classes shown in \cref{tab:Classes} with three or fewer zeros are the only ones lying on the $\NS$ boundary yet being outside the Bell polytope $\L$. Since $\M \subsetneq \Q$, these are the only boundary Classes that we need to consider for the present purpose. 
In~\cref{Sec:Class2a}, we give an example of a nonlocal correlation in $\Q_{2a}$ due to $\ket{\MESqb}$, thus showing that the overlap of $\M$ with $\NS$ in Class 2a is not empty. The following theorem, which we show below, takes care of most of the other Classes.
\begin{theorem}\label{Thm:MES:Local}
Nonlocal correlation in $\Q_{2b}$ and $\Q_{2c}$ (and hence, respectively, $\Q_{3a}$ and $\Q_{3b}$) cannot arise from any finite-dimensional maximally entangled state. 
\end{theorem}

\begin{figure}[h!tbp]
\captionsetup{justification=RaggedRight,singlelinecheck=off}
\centering
\begin{subfigure}{0.45\textwidth}
\centering
  \includegraphics[width=\linewidth]{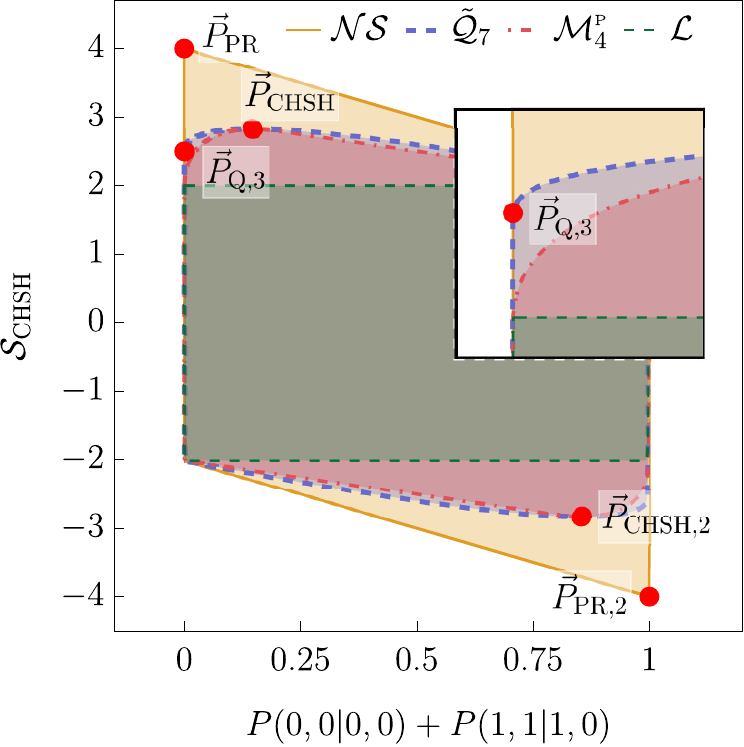}
  \caption{
  A projection revealing the gap between $\Q$ and $\M$ within Class 2b, i.e., all correlations satisfying $P(0,0|0,0)+P(1,1|1,0)=0$ (shown on the vertical line where the horizontal value is zero). Among them, $\vecPQiii$ of \cref{Eq:PQ3} gives the maximal $\SCHSH$ value and self-tests the quantum strategy of \cref{Eq:Class2b} with \cref{Eq:2b:Parameters}. As shown in~\cref{Thm:MES:Local}, $\M \cap \Q_{2b} \subsetneq \L$.}  
  \label{fig:2sub1}
\end{subfigure}
\begin{subfigure}{0.45\textwidth}
\centering
  \includegraphics[width=\linewidth]{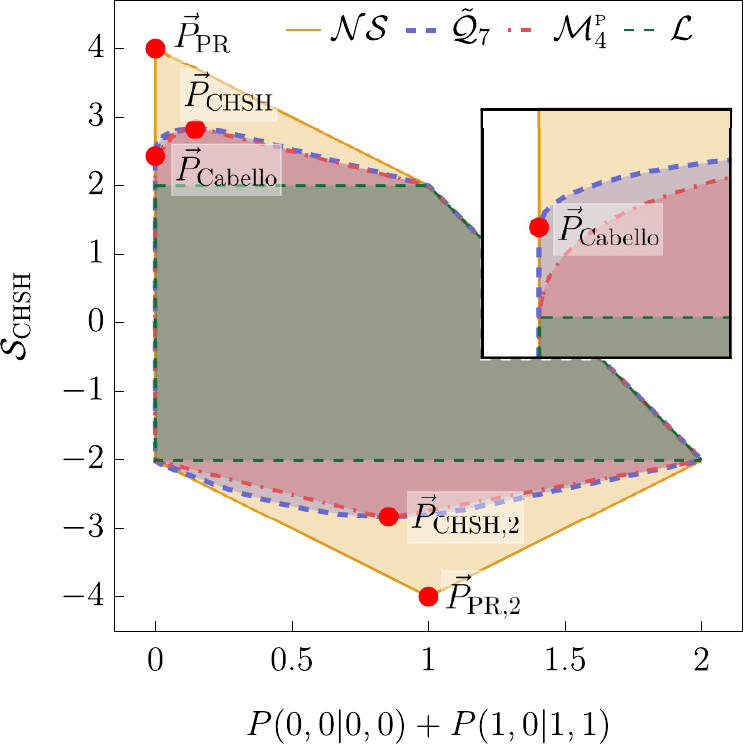}
  \caption{A projection revealing the gap between $\Q$ and $\M$ within Class 2c, i.e., all correlations satisfying $P(0,0|0,0)+P(1,0|1,1)=0$ (shown on the vertical line where the horizontal value is zero). Among them, $\vecPC$ of \cref{Eq:PCabello} gives the maximal $\SCHSH$ value and self-tests  the quantum strategy of \cref{Eq:Class2c} with \cref{Eq:2c:AllParameters}. As shown in~\cref{Thm:MES:Local}, $\M \cap \Q_{2c} \subsetneq \L$.}
  \label{fig:2sub2}
\end{subfigure}
\caption{Two-dimensional projection of $\NS$, $\L$, $\Q_7$, and $\Mp_4$ onto the plane labeled by the CHSH Bell value $\SCHSH$ and some linear sum of probabilities (see text for details). When the horizontal parameter is equal to zero, we recover the configuration specified in the correlation tables of~\cref{tab:class2b} and ~\cref{tab:class2b}. The insets clearly illustrate, correspondingly, the non-exposed nature of both $\vecPQiii$ and $\vecPC$ as well as the gap between $\Q$ and $\M$ within each Class. The explicit form of the other $\vecP$'s shown are given in \cref{Eq:vecPR} and \cref{Eq:3vecP}.
}
\label{fig:2}
\end{figure}

\begin{proof}
To prove the Theorem, we make use of the following Lemma, whose proof is given in~\cref{App:MESlocal}.
\begin{lemma}\label{lem2}
For the maximally entangled state $\ket{\MESd}$, if the probability $P(a,b|x,y)=\bra{\MESd} E_{a|x}\otimes F_{b|y}\ket{\MESd}=0$ for any $x,y,a,b\in\{0,1\}$, the corresponding operators $E_{a|x}\tp$ and $F_{b|y}$ can be diagonalized in the same basis.
\end{lemma}

Since extremal correlations of $\M$ have qubit realization, cf. Lemma~\ref{CHSH:Extreme:MES}, it suffices to show that $\ket{\MESqb}$ cannot give nonlocal correlation in these Classes. Moreover, one should bear in mind that (i) correlations having fewer zeros may be obtained as mixtures of correlations having more zeros, but not the other way around (ii) $\vecP\in\Q$ (and hence $\M$) having more than three zeros must be local (iii) it suffices to consider rank-one projective measurements if we want the resulting extremal correlation to be nonlocal.

Now, note from Lemma~\ref{lem2} that for rank-one projective measurements acting on $\ket{\MESqb}$, 
\begin{subequations}
\begin{align}
	P(0,0|0,0)=0\implies P(1,1|0,0)=0,\label{Eq:Implication1}\\
	P(1,0|1,1)=0\implies P(0,1|1,1)=0 \label{Eq:Implication2}
\end{align}
\end{subequations}
Applying \cref{Eq:Implication1} to correlations from Class 3a (\cref{tab:class3a}), we end up having two zeros on the same row of the correlation table. By Lemma~\ref{lem1}, we know that such correlations must be local. Moreover, by point (i) and (ii) above, there cannot be mixtures of nonlocal correlations with more zeros that give a nonlocal $\vecP\in\M\cap\Q_{3a}$. Hence, all $\vecP\in\M\cap\Q_{3a}$ are local.  Similarly, applying \cref{Eq:Implication1} to correlations in Class 2b (\cref{tab:class2b}) leads to correlations that are local since they have two zeros in the same row. Even though correlations from the Class 3a can be mixed to give correlations from the Class 2b, using the just-established result for $\vecP\in\M\cap\Q_{3a}$, we know that such mixtures must again be local.

Next, by applying \cref{Eq:Implication2} to correlations in the Class 3b (\cref{tab:class3b}), we end up with correlations (up to relabeling) from the Class 4b (\cref{tab:class4b}). However, we already show in \cref{app1} that all $\vecP\in\Q_{4b}$ are local. Together with point (i) and (ii) given above, we thus conclude that all $\vecP\in\M\cap\Q_{3b}$ must be local. Finally, the application of \cref{Eq:Implication1} to correlations in the Class 2c (\cref{tab:class2c}) give correlations in the Class 3b. Moreover, an inspection of \cref{tab:Classes} shows that there are no other Classes of correlations that can be mixed to give correlations from the Class 2c. Hence, all $\vecP\in\M\cap\Q_{2c}$ are local, which completes the proof of the Theorem.

\end{proof}

\begin{figure}[t!]
\captionsetup{justification=RaggedRight,singlelinecheck=off}
\centering
  \includegraphics[width=0.93\linewidth]{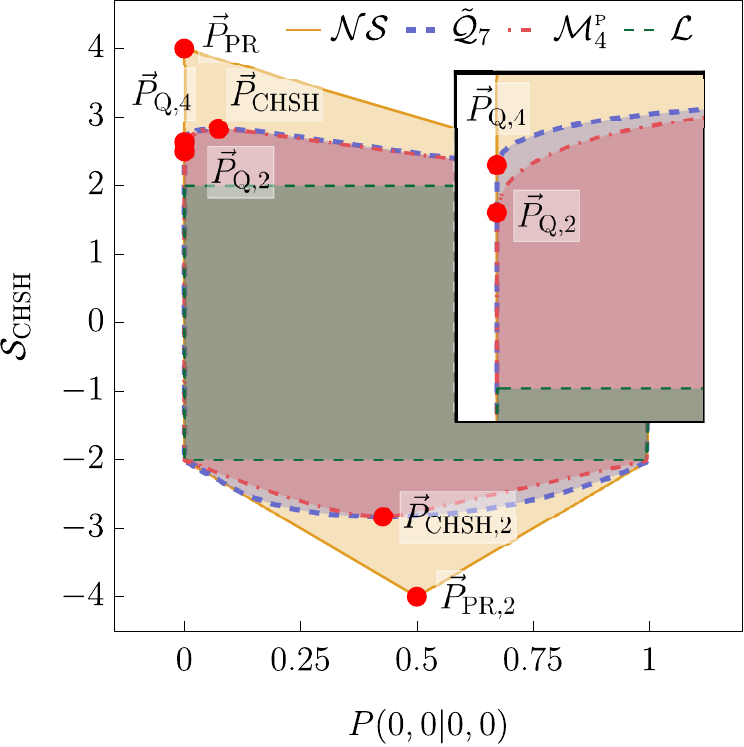}
\caption{Two-dimensional projection of $\NS$, $\L$, $\Q_7$, and $\Mp_4$ onto the plane labeled by the CHSH value $\SCHSH$ and $P(0,0|0,0)$. When the horizontal parameter is equal to zero, we recover the configuration specified in~\cref{tab:class1}. The inset clearly illustrates the gap between $\Q$ and $\M$ within Class 1 (i.e., when the horizontal value is zero) even though there exists nonlocal point $\vecP \in \{\Q_1\cap \M\}$ that is arbitrarily close to $\vecPQii$. The explicit form of the other $\vecP$'s shown are given in \cref{Eq:vecPR} and \cref{Eq:3vecP}. 
}\label{fig:3}
\end{figure}

\cref{Thm:MES:Local} can be seen as a no-go theorem for self-testing maximally entangled states on the common boundaries of $\Q$ and $\NS$ belonging to the aforementioned Classes. In a way, the theorem is anticipated given Lemma~\ref{CHSH:Extreme:MES} and the absence~\cite{Hardy1993,KCA+06} of a Hardy-type argument for the Bell state $\ket{\MESqb}$.
 To illustrate the fact that $\Q_{2b}\cap\M\subsetneq \L$ and $\Q_{2c}\cap\M\subsetneq \L$, we plot in~\cref{fig:2sub1} and~\ref{fig:2sub2}, correspondingly, the 2-dimensional projections of $\NS$, $\L$, $\tilde{\Q}_7$ (the level-$7$ outer approximation of $\Q$ due to~\cite{Moroder13}) and $\M_4$ (the level-$4$ outer approximation of $\M$ due to~\cite{Lin:Quantum:2022}\footnote{\label{footnote:Mp=M}Strictly, the technique of~\cite{Lin:Quantum:2022} only outer approximates the  set of $\M$ assuming projective measurements. However, as we remark in the proof of~\cref{Thm:MES:Local}, $\Mp$ and $\M$ coincide in this simplest Bell scenario.}). Indeed, in both cases, as the boundary of $\M_4$ approaches that of $\NS$, it also approaches that of $\L$. 

\begin{table*}
\small\addtolength{\tabcolsep}{-4pt}
          \captionsetup{justification=RaggedRight}
    \begin{tabular}{|c|c|c||c|c|c|c|c|c|c|}
    \hline
         Class & Necessarily $\in\L$? & $\L^\text{c} \cap \M \neq \emptyset$?  & Maximizing $\vecP$ & Extremal in $\Q$? & Self-test? & Non-Exposed? & $\SCHSH^{\max}$ & $E(\ket{\psi})$ \\ \hline\hline
        4b (\cref{tab:class4b}) & Yes & N/A & N/A & No & No & No &$2$ & 1\\ \hline
         3a (\cref{tab:class3a}) & No & No & Eq.~\eqref{Eq:vecPH}~\cite{Hardy1993} & Yes & Yes~\cite{Rabelo12} & Yes~\cite{Goh2018} & $2.3607$ & 0.6742 \\
        3b (\cref{tab:class3b}) & No & No & \cref{Eq:vecPQ} & Yes & Yes & Yes &$2.2698$ & 0.7748\\
        2a (\cref{tab:class2a}) & No & Yes & \cref{Eq:PQ2} & Yes & Yes & Yes &  $2.5$ & 1\\
        2b (\cref{tab:class2b}) & No & No & \cref{Eq:PQ3}  & Yes & Yes & Yes & $2.5$ & 0.8113 \\
        2c (\cref{tab:class2c}) & No & No & \cref{Eq:PCabello}~\cite{KCA+06} & Yes & Yes~\cite{Rai:PRA:2021} & Yes & $2.4312$ & 0.7794\\ 
        1 (\cref{tab:class1}) & No & Yes & \cref{Eq:PQ4} & Yes & Yes & Yes &  $2.6353$ & 0.9255\\ \hline 
\end{tabular}
    \caption{\label{tab:Summary} Summary of our findings regarding the common boundaries of $\Q$ and $\NS$. These boundaries are classified according to the number of zeros present in the respective correlation tables, see \cref{tab:Classes}. For each of these Classes given in the leftmost column, we list in the second and third column, respectively, whether a $\vecP\in\Q$ from the Class is necessarily local and whether a nonlocal $\vecP\in\Q$ from the Class can be realized using a finite-dimensional maximally entangled state (here $\L^\text{c}$ denotes the complement of $\L$ in $\NS$, i.e., the set of nonlocal correlations). Further to the right, we give the equation number for the CHSH-maximizing quantum correlation from this Class, an indication on whether the example is extremal in $\Q$, is a self-test,\footnote{Although it was only shown in~\cite{Rabelo12} and~\cite{Rai:PRA:2021} that the quantum correlation self-tests the corresponding state, one can easily check there exists diagonal local unitaries [$u_A=u_B=\text{diag}(e^{i\theta}, 1)$ for the state and the measurements in Lemma 2 in~\cite{Rabelo12} and $u_A=\text{diag}(e^{i\phi}, 1)$, $u_B=\text{diag}(e^{i\xi}, 1)$ for Eq.~(5) and Eq.~(7) in~\cite{Rai:PRA:2021}] such that Definition~\ref{Dfn:Equivalent} holds, and thus~\cref{Thm:self-test:upgrade} can be applied to also self-test the underlying measurements.} and is non-exposed. Equation numbers for the more general examples are given in \cref{tab:SummaryReduction}. In the last two columns, we give, respectively, the maximal CHSH  value attainable in each Class and the entanglement of formation~\cite{Wootters:PRL:1998} $E(\ket{\psi})$ of the state $\ket{\psi}$ giving this maximum. Notice that even among those Classes having the same number of zeros, the state $\ket{\psi}$ giving a  larger $\SCHSH^\text{max}$ may not have a larger amount of entanglement.}
\end{table*}

As for Class 1, \cref{Eq:Implication1} implies that with rank-one projective measurements acting on $\ket{\MESqb}$, we cannot obtain $\vecP\in\M$ with only one zero. Together with Lemma~\ref{CHSH:Extreme:MES}, we come to the intriguing observation that $\vecP\in\M$ having exactly one zero are non-extremal in $\M$. However, if we introduce some noise to Alice's $0$-th measurements in~\cref{Eq:Class2a}, namely, by adopting for some positive $\epsilon\approx 0$,
\begin{equation}
         M^A_{0|0}=(1-\epsilon)\frac{\mathbb{1}_2+A_0}{2}\,\,\text{ and }\,\,  M^A_{1|0} = \mathbb{1}_2-M^A_{0|0}
\end{equation}
instead of \cref{Eq:Obs->POVM}, we can clearly generate nonlocal $\vecP\in\M$ within Class 1. Moreover, it is easy to see that for such strategies, $\sup_{0 < \epsilon < 1}\SCHSH=2.5$, i.e., 
we can find $\vecP\in\M$ having only one zero that is arbitrarily close to $\vecPQii$. In fact, using the technique described in~\cite{Lin:Quantum:2022}, we find that $2.5$ matches the best upper bound on $\SCHSH$ for $\vecP\in\M\cap\Q_{1}$ to within a numerical precision better than $10^{-5}$. Still, it is worth noting that this largest value $\SCHSH$ that $\vecP\in\M$ can approach within Class 1 lies strictly below the value of $\SCHSH \approx 2.63533$ attainable by $\vecPQiv\in\Q_{1}$, as illustrated in the inset of~\cref{fig:3}.

As a side remark, note that the mixing strategy given above allows one to transform quantum correlations belonging to one Class to another Class having fewer zeros.

\section{Conclusion}~\label{sec5}

Knowing when the boundary of the quantum set $\Q$ meets the boundary of the no-signaling set $\mathcal{NS}$ helps us better understand the limits of quantum resources when it comes to information processing tasks. In this work, independently of~\cite{Rai:PRA:2019}, we completely characterize the common boundary of these two sets in the simplest Bell scenario. Inspired by the work of ~\cite{Fritz_2011}, we classify the distinct cases according to the numbers of zero present in the correlation table (see~\cref{tab:Classes}). Our main findings are summarized in~\cref{tab:Summary}.

In addition to showing that only six of these Classes contain nonlocal correlations (also shown independently in~\cite{Rai:PRA:2019}), we prove that only two contain nonlocal correlations that are attainable using finite-dimensional maximally entangled states. For each of these Classes with $k$ zeros where $k\in\{1,2,3,4\}$, we provide a $(5-k)$-parameter family of (extremal) quantum strategies realizing these correlations. Furthermore, we  determine for each Class the correlation leading to the maximal CHSH-Bell-inequality violation. As we show in this work, all these maximizing quantum strategies can be self-tested by the resulting correlation with some robustness. Moreover, all these self-testing correlations are provably non-exposed~\cite{Goh2018}.

In arriving at the above results, we obtain several other results that may be of independent interest. For example, we establish in~\cref{Thm:self-test:upgrade} that in the CHSH Bell scenario, if all qubit strategies giving the same extremal quantum correlation are local-unitarily equivalent, then this correlation is a self-test. We also show in Corollary~\ref{CHSH:Decompose:MES} that in this simplest Bell scenario, the set of correlations attainable by finite-dimensional maximally entangled states can be realized using mixtures of correlations from a Bell pair with projective measurements. In turn, this allows us to derive in Corollary~\ref{Prop:MaxCHSH:MES} the maximal CHSH Bell inequality violation by an arbitrary finite-dimensional maximally entangled state. This question was left open since the work of~\cite{Liang2006}.

An observant reader has probably noticed that the self-testing results presented here for the CHSH-maximizing correlation of Class 3a, 3b, 2b, and 2c do not seem particularly robust (see \cref{fig:selftest:class3a} and \cref{fig:selftest:class3b} -- \cref{fig:selftest:class1} in~\cref{App:Robust}). Thus, an obvious avenue of research is to adapt the techniques of~\cite{Kaniewski:PRL:2015} here to improve the robustness of these results. In fact, these and other related work~\cite{LMS+2023,Jed:private} would suggest that the weak form~\cite{Jeba:PRR:2019,Kaniewski:PRR:2020} of self-testing discovered in other Bell scenarios may well be absent in this Bell scenario. Another closely related question concerns the possibility to self-test other quantum correlations from the same Class. The recent work by Rai {\em et al.}~\cite{Rai:PRA:2022} shows that this is possible for a two-parameter family of correlations from Class 3a. Our results show the same for a one-parameter family of correlations from Class 2a. 
Moreover, we have identified a one-parameter family from Class 3b, 2b, and 2c where our numerical results suggest that self-testing is plausible. A complete analysis for these and the other correlations from each Class would be desirable.

Also worth noting is that among the correlations given in~\cite{Rai:PRA:2022}, our calculation based on $21^4=194,481$ uniformly chosen samples shows that all these self-testing Hardy-type correlations are non-exposed,\footnote{These observations hold to a numerical precision better than $10^{-13}$.} likewise, for $10^5$ self-testing correlations from the one-parameter family given in~\cref{Sec:Class2a}. Based on these observations, it seems conceivable that (1) all extremal quantum correlations on the boundary of $\NS$ may be self-tested, and (2) all these self-testing correlations on the boundary of $\NS$ are non-exposed. If so, it would be impossible to self-test such extremal nonlocal quantum correlations using {\em solely} the maximal quantum violation of any Bell inequality. However, as we show in this work, we can circumvent this impossibility by supplementing the Bell value with additional (zero) constraint(s). We are optimistic that this approach remains viable for non-exposed points in even more complicated Bell scenarios such as that considered in~\cite{Chen:PRA:2013}. Determining the validity of these conjectures clearly requires more work and shall thus be left for future work.

\section*{Acknowledgments}

We are grateful to R. Augusiak, J. Kaniewski, M. T. Quintino, V. Scarani for helpful discussions, to A. Rai and anonymous reviewers of Quantum for helpful comments on an earlier version of this manuscript, to P.-S. Lin for sharing his code and other contributions to this research project. This work is supported by the National Science and Technology Council (formerly Ministry of Science and Technology), Taiwan (Grants No. 107-2112-M-006-005-MY2,  109-2112-M-006-010-MY3, 110-2112-M-032-005-MY3, and 111-2119-M-008-002),  a grant provided by Academia Sinica through the Academia Sinica Research Award for Junior Research Investigators, and the Foundation for Polish Science through the First Team project (First TEAM/2017-4/31) co-financed by the European Union under the European Regional Development Fund.

\appendix

\section{Proof of \cref{Thm:self-test:upgrade}}
\label{App:Proof:SelfTestingTheorem}

The following proof mirrors that given in~\cite{Rabelo12} for showing that the quantum correlation giving the maximal success probability in the original Hardy paradox~\cite{Hardy1993} is a self-test.
\begin{proof}
Let us denote by $\vecP$ the extremal nonlocal correlation of interest. By Neumark's theorem~\cite{Neumark}, all possible quantum strategies leading to $\vecP$ can be realized using a pure entangled state (possibly of a higher Hilbert space dimension) and projective measurements. Then, by the Lemma shown in~\cite{Masanes06}, there exists a local basis such that all the projective POVM elements on Alice's (Bob's) side  are simultaneously block diagonal, with blocks of size $2\times 2$ or $1\times 1$, i.e.,\footnote{As with \cref{fn:DirectSum}, to be mathematically precise, all the direct sums should be replaced by a regular sum in the following equations. Accordingly,  each $\Pi^A_i$ (and $M^{A_i}_{a|x}$) is an operator acting on the full space, but only has support on the subspace of the $i$-th block. Likewise for $\Pi^B_j$ and $M^{B_j}_{b|y}$.} 
\begin{equation}~\label{gen_meas}
\begin{split}
        M^A_{a|x} = \bigoplus_i M^{A_i}_{a|x} = \bigoplus_i \Pi^A_i M^{A_i}_{a|x} \Pi^A_i,\\
        M^B_{b|y} = \bigoplus_j M^{B_j}_{b|y} = \bigoplus_j \Pi^B_j M^{B_j}_{b|y} \Pi^B_j,
\end{split}
\end{equation}
where $i$, $j$ are labels for the block and $\Pi^A_i$ ($\Pi^B_j$) is the projector onto the $i$-th ($j$-th) block of Alice's (Bob's) Hilbert space, in the sense that $\bigoplus_i \Pi^A_i=\mathbb{1}_A$ and $\bigoplus_j \Pi^B_j=\mathbb{1}_B$.

Using this in Born's rule, \cref{Eq:Born}, we get
\begin{align}
        P(a,b|x,y)&=\tr[\rho(M^A_{a|x}\otimes M^B_{b|y})]\nonumber\\
        &=\tr\big[\rho \bigoplus_i M^{A_i}_{a|x}\otimes \bigoplus_j M^{B_j}_{b|y}\big]\nonumber\\
        &=\sum_{i,j}\nu_{ij} \tr[\rho_{ij} (M^{A_i}_{a|x}\otimes  M^{B_j}_{b|y})],\nonumber\\
        &=\sum_{i,j}\nu_{ij}P_{ij}(a,b|x,y),\label{Eq:ExtremalMixture}
\end{align}
where $\nu_{ij}:=\tr[\Pi^A_i \otimes \Pi^B_j) \rho (\Pi^A_i \otimes \Pi^B_j)]$ and when $\nu_{ij}\neq 0$, 
\begin{gather}\label{Eq_rho_P}
\begin{split}
        \rho_{ij}:= \frac{(\Pi^A_i \otimes \Pi^B_j) \rho (\Pi^A_i \otimes \Pi^B_j) }{\nu_{ij}},\\
        P_{ij}(a,b|x,y):=\tr[\rho_{ij} (M^{A_i}_{a|x}\otimes  M^{B_j}_{b|y})].
\end{split}        
\end{gather}
Note that the probability of a successful projection onto the $i$-th block of Alice's Hilbert space and the $j$-th block of Bob's Hilbert space can be factorized into the product of a successful projection onto the individual block of the respective Hilbert space, i.e., $\nu_{ij}=r_is_j$, where $r_i=\tr[\Pi^A_i\,\tr_\text{B}(\rho)\,\Pi^A_i], s_j=\tr[\Pi^B_j\,\tr_\text{A}(\rho)\,\Pi^B_j]\ge0$, and $\sum_i r_i = \sum_j s_j=1$.

By assumption, $\vecP$ is extremal. However, we also see in \cref{Eq:ExtremalMixture} that $\vecP$ is a convex mixture of $\vecP_{ij}$. This can only hold if $\vecP_{ij}$ is in fact $\vecP$ for all $i,j$. This implies that $\nu_{ij}$ involving one-dimensional projection must vanish as the otherwise surviving $\rho_{ij}$ cannot lead to a nonlocal $\vecP_{ij}$. Moreover, since the surviving local blocks are all of size $2\times 2$, we know by assumption that for all $i,j$, the quantum strategy of $\left\{\rho_{ij}=\proj{\psi_{ij}}, \{M^{A_i}_{a|x}\}, \{M^{B_j}_{b|y}\}\right\}$ is $2$-equivalent to the reference strategy. In other words, there exist local unitaries $u_{A,i}$ and $u_{B,j}$, cf. \cref{Eq:Dfn:d-equivalent},  such that
\begin{gather}\label{Eq_blockselftest}
	u_{A,i}\otimes u_{B,j} \ket{\psi_{ij}}=\ket{\widetilde{\psi}_{ij}},\\
	u_{A,i}\otimes u_{B,j} (M^{A_i}_{a|x}\otimes  M^{B_j}_{b|y})u_{A,i}^\dagger\otimes u_{B,j}^\dagger=\widetilde{M^{A_i}_{a|x}}\otimes \widetilde{M^{B_j}_{b|y}}\nonumber
\end{gather}
where $\ket{\widetilde{\psi}_{ij}}$ and $\{\widetilde{M^{A_i}_{a|x}}\}$,  $\{\widetilde{M^{B_j}_{b|y}}\}$ are the reference two-qubit strategy having support in the $i$-th ($j$-th) qubit block of Alice's (Bob's) Hilbert space.

Defining $\Lambda_A:=\bigoplus_{l}u_{A,l}^\dagger$ and $\Lambda_B:=\bigoplus_{m}u_{B,m}^\dagger$, then the global state $\rho=\proj{\psi}$ is easily seen to be:
\begin{equation}~\label{gen_state}
    \begin{split}
        \ket{\psi}=&\bigoplus_{i,j}\sqrt{\nu_{ij}}\ket{\psi_{ij}}\\
        =&\bigoplus_{ij}u_{A,i}^\dagger\otimes u_{B,j}^\dagger\sqrt{\nu_{ij}}\ket{\widetilde{\psi}_{ij}}\\
        = & \Lambda_A\otimes \Lambda_B\bigoplus_{i,j}\sqrt{\nu_{ij}}\ket{\widetilde{\psi}_{ij}}.
    \end{split}
\end{equation}
whereas the corresponding POVM elements are
\begin{align}~\label{localisometiesmeasure}
        M^A_{a|x}\otimes M^B_{b|y} &= 
        \bigoplus_i M^{A_i}_{a|x}\otimes \bigoplus_j M^{B_j}_{b|y} \\
        &= \Big(\bigoplus_i u_{A,i}^\dagger \widetilde{M^{A_i}_{a|x}} u_{A,i})\otimes (\bigoplus_j u_{B,j}^\dagger \widetilde{M^{B_j}_{b|y}} u_{B,j}\Big)\nonumber\\
        &= \Lambda_A\otimes \Lambda_B \bigoplus_{i,j}(\widetilde{M^{A_i}_{a|x}}\otimes \widetilde{M^{B_j}_{b|y}})\Lambda_A^\dagger\otimes \Lambda_B^\dagger.\nonumber
\end{align}

Without loss of generality, we may take the reference POVM elements for the first input ($x=y=0$) as a measurement in the computational basis\footnote{If the intended reference measurements are not already in this form, one can make use of the promised $2$-equivalence to choose, instead, a reference measurement that has this feature.}
\begin{equation}
	\widetilde{M^{A_i}_{a|0}}=\proj{2i+a},\quad \widetilde{M^{B_j}_{b|0}}=\proj{2j+b}
\end{equation}
and the reference state $\ket{\widetilde{\psi}_{ij}}$ in the same computational basis.

What remains is to apply the local isometries 
\begin{equation}
	\Phi_A=\Phi'\circ\Lambda_A^\dagger,\quad \Phi_B=\Phi'\circ\Lambda_B^\dagger 
\end{equation}
defined via
\begin{gather}
    \begin{split}
        \Phi'\ket{2s}_{K}\mapsto \ket{2s}_{K'}\ket{0}_{K''},\\
        \Phi' \ket{2s+1}_{K} \mapsto \ket{2s}_{K'}\ket{1}_{K''},
    \end{split}
\end{gather}
for $K\in\{A,B\}$, likewise for $K'$ and $K''$. It is then straightforward to see that these isometries indeed fulfill the self-testing requirement given in \cref{Eq:Selftest}, thus completing the proof that the given extremal nonlocal correlation $\vecP$ self-tests the reference state and measurement.
\end{proof}

\section{All quantum correlations in Class 4b are local}\label{app1}

As with the proof of Lemma~\ref{lem1}, we denote by $\{\ketA{a|x}\}$  ($\{\ketB{b|y}\}$)  the orthonormal bases defining Alice's $x$-th (Bob's $y$-th) measurement.
With the help of relabeling, quantum correlations in Class 4b can always be cast in the form of \cref{tab:class4b}.
Then, the two ADZs in the upper-left block imply that the shared  two-qubit pure state can be expressed as
\begin{equation}\label{eqa1}
    \ket{\Psi}=\sum_{i=0}^1 c_i\ketA{i|0}\ketB{i|0},\quad \sum_{i=0}^1 |c_i|^2=1.
\end{equation}
Similarly, the ADZs in the lower-right block mean that the shared state can also be written as
\begin{equation}\label{eqa2}
    \ket{\Psi}=\sum_{i=0}^1 d_i\ketA{i|1}\ketB{i|1},\quad \sum_{i=0}^1 |d_i|^2=1.
\end{equation}

Now let $U_A$ and $U_B$, respectively, be the unitary connecting the bases of different local measurement settings, i.e., 
\begin{gather}\label{Eq:Transformation}
	 \ketA{i|1} = \sum_j (U_A)_{ji}\ketA{j|0},\,\, \ketB{i|1} = \sum_j (U_B)_{ji}\ketB{j|0},\\
\label{Eq.unitarites.lemme3.1}
	 U_A=\begin{bmatrix}
    \alpha & -\beta^*  \\
    \beta & \alpha^*
    \end{bmatrix},\,\, U_B=\begin{bmatrix}
    \alpha' & -\beta'^* \\
    \beta' & \alpha'^*
    \end{bmatrix},
\end{gather}
where $^*$ represents complex conjugation and $|\alpha|^2+|\beta|^2=|\alpha'|^2+|\beta'|^2=1$. Here and below, without loss of generality, we consider only unitary matrices of unit determinant.

Substituting the transformations of \cref{Eq:Transformation} and \cref{Eq.unitarites.lemme3.1} into equation~\eqref{eqa2}, we obtain an expansion of $\ket{\Psi}$ in the basis of $\{\ketA{i|0}\ketB{j|0}\}$. Specifically, the expansion coefficients for the term $\ketA{0|0}\ketB{1|0}$ and $\ketA{1|0}\ketB{0|0}$, which ought to be zero according to \cref{eqa1}, are, respectively,
\begin{equation}
\begin{split}
	\ketA{0|0}\ketB{1|0}:\quad d_0\alpha \beta '-d_1\beta ^* \alpha '^*=0,\\
	\ketA{1|0}\ketB{0|0}:\quad d_0\beta \alpha '-d_1\alpha^* \beta'^*=0.
\end{split}
\end{equation}
Hence, for $d_0,d_1\neq0$, their ratio is just a phase factor:
\begin{equation}
    \frac{d_0}{d_1}=\sqrt{\frac{\alpha^* \alpha '^* \beta^* \beta'^*}{\alpha \alpha ' \beta \beta '}}=e^{i\phi_d}\,\,\because\,\,\alpha \alpha ' \beta \beta '\equiv|\alpha \alpha ' \beta \beta '|e^{-i\phi_d}.
\end{equation}
Similarly, using the inverse unitaries $U_A^\dag$ and $U_B^\dag$, we can write the basis vectors $\{\ketA{i|0}\ketB{j|0}\}$ in the basis of $\{\ketA{i|1}\ketB{j|1}\}$ in~\cref{eqa1} and arrive at the conclusion that $c_0/c_1$ is also a phase factor, which we denote by $e^{i\phi_c}$.

This means that if $\ket{\Psi}$ exhibits exactly the zeros distribution shown in~\cref{tab:class4b}, it must be maximally entangled, i.e.,
\begin{equation}\label{eqa5}
    \begin{split}
        \ket{\Psi}=&\frac{e^{i\phi_0}}{\sqrt{2}}[e^{i\phi_c}\ketA{0|0}\ketB{0|0}+\ketA{1|0}\ketB{1|0}]\\
        =&\frac{e^{i\phi_1}}{\sqrt{2}}[e^{i\phi_d}\ketA{0|1}\ketB{0|1}+\ketA{1|1}\ketB{1|1}]\\
    \end{split}
\end{equation}
where $e^{i\phi_0}$ and $e^{i\phi_1}$ are global phase factors. Using \cref{Eq_Observables}, we then get
\begin{equation}\label{eqa7}
    \begin{split}
        &\langle A_0 B_0 \rangle =\bra{\Psi}A_0 B_0\ket{\Psi}=1,\\
        &\langle A_0 B_1 \rangle =\bra{\Psi}A_0 B_1\ket{\Psi}=|\alpha|^2-|\beta|^2,\\
        &\langle A_1 B_0 \rangle =\bra{\Psi}A_1 B_0\ket{\Psi}=|\alpha|^2-|\beta|^2,\\
        &\langle A_1 B_1 \rangle =\bra{\Psi}A_1 B_1\ket{\Psi}=1,
    \end{split}
\end{equation}
which satisfy the CHSH inequality of \cref{Eq:CHSH} and all its relabelings. Hence, all $\vecP\in\Q$ from Class 4b must be local.

\section{Robust self-testing of quantum strategies}\label{App:robust_self-test}

In any real experiment, the observed measurement statistics are likely to differ from the exact self-testing quantum correlation. As such, for practical purposes, it is important to understand the robustness of any given self-testing statement against imperfections. To this end, we shall recall from~\cite{Yang14,Bancal15} the SWAP method and explain how it can be applied to numerically determine the robustness of the self-testing results presented in~\cref{Sec:Q-NS}. 

For the robust self-testing of a quantum state, we follow~\cite{Yang14,Bancal15} by considering local operators $\Phi_{AA'}$ ($\Phi_{BB'}$) that act jointly on the black-box system $A$ ($B$) and the trusted auxiliary system $A'$ ($B'$), i.e., $\Phi \rho_{AB}\otimes\left(\ket{00}\bra{00}\right)_{A'B'}\Phi^\dagger$, where 
\begin{equation}\label{Eq:GlobalOp}
	\Phi = \Phi_{AA'}\otimes\Phi_{BB'}.
\end{equation} 
As we see shortly, in the ideal case where Alice's (Bob's) actual observables $A_i$ ($B_i$) coincide with the reference observables $\widetilde{A}_i$ ($\widetilde{B}_i$), the chosen local operator $\Phi_{AA'}$ ($\Phi_{BB'}$) becomes the operator that swaps the Hilbert spaces $\H_A$ and $\H_{A'}$ ($\H_B$ and $\H_{B'}$). Then the fidelity 
 \begin{equation}\label{Eq:Fidelity}
    F = \bra{\widetilde{\psi}}\rho_{\text{\tiny{SWAP}}}\ket{\widetilde{\psi}},
 \end{equation}
between the reference state $\ket{\widetilde{\psi}}$ and the ``swapped" state 
  \begin{equation}~\label{def:rho_swap}
     \rho_\text{\tiny{SWAP}} = \tr_{AB}\left[\Phi~\rho_{AB}\otimes (\proj{00})_{A'B'}~\Phi^\dagger\right],
 \end{equation}
naturally quantifies the closeness of relevant shared state in the black boxes $\rho_{AB}$ to the reference state $\ket{\widetilde{\psi}}$.

To this end, let us first note from \cref{Sec:Q-NS} that all reference observables considered are of the form:
\begin{equation}\label{Eq:RefObs}
    \widetilde{A_i}=\cos{(\theta_{i})}\sigma_x+\sin{(\theta_{i})\sigma_z}. 
\end{equation}
Next, we define, for Alice's subsystem,  the local operator: 
\begin{subequations}\label{Eq:LocalSwap}
\begin{align}
    \Phi&_{AA'}(A_0,A_1):=\nonumber\\
    &U_{AA'}(A_0,A_1)V_{AA'}(A_0,A_1) U_{AA'}(A_0,A_1),
\end{align}
where the ``controlled-NOT'' gates are:
\begin{align}
    U_{AA'}(A_0,A_1) &:= \mathbb{1}_A\otimes \proj{0}+\widetilde{\sigma_{x}}_{,A}(A_0,A_1)\otimes\proj{1}, \nonumber\\
    V_{AA'}(A_0,A_1) &:= \frac{\mathbb{1}_A+\widetilde{\sigma_{z}}_{,A}(A_0,A_1)}{2} \otimes \mathbb{1} \\ 
    				&\qquad+\frac{\mathbb{1}_A-\widetilde{\sigma_z}_{,A}(A_0,A_1)}{2}\otimes\sigma_x,\nonumber
\end{align}
and following the form given in \cref{Eq:RefObs}, we define
\begin{equation}
\begin{split}
    \widetilde{\sigma_x}_{,A}(A_0,A_1):= \frac{-\sin\theta_{1}\,A_0+\sin\theta_{0}\,A_1}{\sin\theta_{0}\,\cos\theta_{1}\,-\cos\theta_{0}\,\sin\theta_{1}\,},\\
    \widetilde{\sigma_z}_{,A}(A_0,A_1):= \frac{\cos\theta_{1}\,A_0-\cos\theta_{0}\,A_1}{\sin\theta_{0}\,\cos\theta_{1}\,-\cos\theta_{0}\,\sin\theta_{1}\,}
\end{split}
\end{equation}
\end{subequations}
 Using~\cref{Eq:Obs->POVM}, we can rewrite the local operator $\Phi_{AA'}(A_0,A_1)$ in terms of the actual POVM elements, i.e., $\Phi_{AA'}(A_0,A_1)\to \Phi_{AA'}(\mathbb{1}_2,M^A_{0|0},M^A_{0|1})$. Completely analogous definitions can be given for Bob's local operator $\Phi_{BB'}$ in terms of his actual POVM elements $M^B_{0|0}$ and $M^B_{0|1}$.  Importantly, when the black boxes indeed implement the reference measurements,  $\widetilde{\sigma_x}_{,A}(\widetilde{A_0},\widetilde{A_1}) =\widetilde{\sigma_x}_{,B}(\widetilde{B_0},\widetilde{B_1})=\sigma_x$ and $\widetilde{\sigma_z}_{,A}(\widetilde{A_0},\widetilde{A_1})=\widetilde{\sigma_z}_{,B}(\widetilde{B_0},\widetilde{B_1}) =\sigma_z$, the local operators $\Phi_{AA'}$, $\Phi_{BB'}$ become the ideal operator that swap the black-box subsystem and the respective auxiliary subsystem, and hence the fidelity of \cref{Eq:Fidelity} is unity.

More generally, by substituting \cref{Eq:GlobalOp}, \cref{def:rho_swap}, \cref{Eq:LocalSwap} and the analogous expressions for Bob's operator into \cref{Eq:Fidelity}, we see that the fidelity of interest is a linear function of a subset of the moments $\mu = \{\tr(\rho_{AB}\mathbb{1}),\tr(\rho_{AB}M^A_{0|0}),\dots\}$. To determine the robustness of the established self-testing statements (with respect to the reference state), we shall compute the worst-case  fidelity by optimizing over all possible quantum realizations compatible with the observed value of $\SCHSH$ and some bounded deviation from the $\NS$ boundary. 

In practice, since we only have access to a converging hierarchy~\cite{NPA2008,Doherty08,Moroder13} of outer approximations to $\Q$ (see also~\cite{Lin:Quantum:2022}), we solve, instead, the following semidefinite program (SDP) to obtain a lower bound on this worst-case fidelity:
 \begin{equation}\label{Robust_ST_State}
 \begin{split}
     &\qquad \mathcal{F} = \min_{\mu\in\tilde{\Q}_\ell} \bra{\widetilde{\psi}}\rho_{\text{\tiny{SWAP}}}\ket{\widetilde{\psi}}\\
     \text{such that} &\sum_{a,b,x,y=0}^1 (-1)^{xy+1+b+1} P(a,b|x,y)=\SCHSH,\\
     &~P(a,b|x,y) \le \varepsilon\,\,\forall\,\, P(a,b|x,y) \in \mathbb{P},
 \end{split}
 \end{equation}
where $\mathbb{P}$ is the set of conditional probabilities that are required to vanish in each Class, $\varepsilon$ represents the maximum allowed deviation from this requirement, and $\tilde{\Q}_\ell$ is, for example, the level-$\ell$ outer approximation of $\Q$ described in~\cite{Moroder13} Not further that in this work, all reference states are some two-qubit state, which may be written in the Schdmit form: $\ket{\widetilde{\psi}}=c_0\ket{00}+c_1\ket{11}$ where $c_0\ge c_1$. So, even if Alice and Bob share only the product state $\ket{00}$ and hence do not violate any Bell inequality, they can already achieve a fidelity of $c_0^2$. In other words, for a nontrivial self-testing of the reference state $\ket{\widetilde{\psi}}$, the fidelity has to be larger than $c_0^2$.

For the self-testing of measurements, we now consider 
 the figure of merit:
 \begin{align}
     T_A &\equiv \frac{1}{2}\{P_A(0|0,\ketA{0|0})+P_A(1|0,\ketA{1|0})\nonumber\\
     &\qquad +P_A(0|1,\ketA{0|1})+P_A(1|1,\ketA{1|1})\}-1 \label{FoM_A}
 \end{align}
 for Alice where
  \begin{align}\label{Eq:P_A}
     P_A(&a|x,\ket{\phi})=\nonumber\\
     &\tr\left\{M^A_{a|x}  \otimes\mathbb{1}_{A'} [\Phi_{AA'}(\rho_{AB}\otimes\proj{\phi})\Phi_{AA'}^\dagger]\right\}.
 \end{align} 
For Bob's measurements, we use a similar figure of merit $T_B$ with $P_A(a|x,\ketA{a|x})$ in \cref{FoM_A} replaced by $P_B(b|y,\ketB{b|y})$ where the latter is defined in essentially the same way as  $P_A(a|x,\ket{\phi})$ in \cref{Eq:P_A}.
Again, notices that when the black boxes indeed implement the reference measurements, each of these figure of merits becomes unity.

To determine the robustness of the established self-testing statements with respect to the reference measurements, we are interested in the smallest possible value of these figures of merit given an observation of some correlation that deviates from the ideal self-testing correlation. For that matter, we again outer-approximate $\Q$ by considering some superset $\tilde{\Q}_\ell$ of $\Q$ and  solve the following SDPs to obtain a lower bound on the respective figures of merit:
  \begin{equation}\label{Robust_ST_Meas}
 \begin{split}
      &\qquad  \tau_{i} = \min_{\mu\in\tilde{\Q}_\ell}~T_i  \\
     \text{such that} &\sum_{a,b,x,y=0}^1 (-1)^{xy+1+b+1} P(a,b|x,y)=\SCHSH,\\
     &~P(a,b|x,y) \le \varepsilon\,\,\forall\,\, P(a,b|x,y) \in \mathbb{P},
 \end{split}
 \end{equation}
 for $i\in\{A,B\}$. Note also that if both parties' measurements are trivial, i.e., $M^A_{a|x} = M^A_{b|y} = \frac{\mathbb{1}_2}{2}$ for all $a,b,x,y$, and hence being jointly measurable, one can already achieve a score of $T_A=T_B=0$. Thus, for a nontrivial self-testing of measurements, we must have $\tau_A, \tau_B>0$.
 
\section{Quantum correlation giving maximal CHSH value in $\Q_{3b}$ is a self-test}\label{app-st3}

Here, we prove that $\vecPQ$ of \cref{Eq:vecPQ} gives the maximal CHSH violation in Class 3b. Moreover, $\vecPQ$ is not only an extreme point of $\Q$ but also self-tests the state and measurements of \cref{Eq:Class3b} with \cref{Eq:3b:Parameters}.

\subsection{Maximal CHSH value for $\vecP\in\Q_{3b}$}
\label{App:MaxCHSH-Q3b}

With the measurement bases defined in \cref{Eq_MeasBases} and \cref{Eq_Observables}, a general pure two-qubit state $\ket{\psi}$ giving $P(0,0|0,0)=P(1,1|0,0)=0$ (cf. \cref{tab:class3b}) is:
\begin{equation}~\label{state_two_zeros-in_same_block}
    \ket{\psi}=\cos{\theta}\ketA{0|0}\ketB{1|0}+e^{i\phi}\sin{\theta}\ketA{1|0}\ketB{0|0}.
\end{equation}
Suppose the  measurement bases are related via \cref{Eq:Transformation} by 
\begin{equation}\label{uab3}
\begin{split}
	 U_A&=\begin{bmatrix}
    e^{i\gamma_A}\cos{\alpha} & e^{i\omega_A}\sin{\alpha} \\
    -e^{-i\omega_A}\sin{\alpha} & e^{-i\gamma_A}\cos{\alpha}
    \end{bmatrix},\\
	U_B&=\begin{bmatrix}
    -e^{-i\omega_B}\sin{\beta} & e^{-i\gamma_B}\cos{\beta} \\
     e^{i\gamma_B}\cos{\beta} & e^{i\omega_B}\sin{\beta}
    \end{bmatrix}.
\end{split}
\end{equation}
To arrive at nonlocal correlations, the local measurements should not commute, while the state must be entangled, thus:
\begin{equation}\label{Eq_constraints}
	\alpha, \beta, \theta \notin \left\{\tfrac{n}{2}\pi\right\}_{n\in \mathbb{Z}}.
\end{equation}

Using \cref{state_two_zeros-in_same_block} and \cref{uab3} in \cref{Eq:CHSH} and \cref{Eq:Born}, we arrive, respectively, at the CHSH Bell value:
\begin{align}~\label{S_two_zeros_in_same_block}
        \SCHSH = \zeta\prod_{j=\alpha,\beta,\theta}\sin{2j} +2\cos 2\alpha\cos^2\beta+2\sin^2\beta
\end{align}
and the joint conditional probability
\begin{align}\label{Eq:Constraint:P1011}
        P(1,0|1,1)=&-\tfrac{\zeta}{4}\prod_{j=\alpha,\beta,\theta}\sin{2j}+\sin^2{\alpha}\cos^2{\beta}\cos^2{\theta}\nonumber\\
                  &+\cos^2{\alpha}\sin^2{\beta}\sin^2{\theta}=0
\end{align}
where 
\begin{equation}\label{Eq:zeta}
	\zeta:= \cos{(\gamma_A+\gamma_B+\omega_A+\omega_B+\phi)}
\end{equation} 
and the requirement that the latter probability vanishes follows from the structure of zeros shown in \cref{tab:class3b}.

Substituting the expression of $\zeta$ obtained from \cref{Eq:Constraint:P1011} into \cref{S_two_zeros_in_same_block} and simplifying give
\begin{equation}\label{3zerosS1}
    \SCHSH =2 [\sin^2{\theta}(\cos{2\alpha}-\cos{2\beta})+1].
\end{equation}
Note further that in the expression of $P(1,0|1,1)$ in \cref{Eq:Constraint:P1011}, the variables $\gamma_A, \gamma_B, \omega_A, \omega_B, \phi$ appear only in $\zeta$. If the constraint were to hold for some non-extremal value of the cosine function in $\zeta$, a slight variation of some of these parameters must lead to a negative value for the probability $P(1,0|1,1)$. Hence, the constraint of \cref{Eq:Constraint:P1011} must also imply an extremal value of $\zeta$, i.e., $\zeta=\pm1$. Using this in \cref{Eq:Constraint:P1011} itself and simplifying lead to
\begin{equation}\label{tantheta1}
    \tan{\alpha}\pm \tan{\beta}\tan{\theta}=0\quad\text{ for }\quad \zeta=\mp 1.
\end{equation}

Consider now the maximization of \cref{3zerosS1} subjected to the constraint of \cref{tantheta1} via the Lagrangian functions $\mathcal{L}_\pm$:
\begin{align}
    \mathcal{L}_\pm = &2 [\sin^2{\theta}(\cos{2\alpha}-\cos{2\beta}) +1]\nonumber\\
    &+ \lambda (\tan{\alpha}\pm \tan{\beta}\tan{\theta}),
\end{align}
where $\lambda$ is a Lagrange multiplier. Requiring that their vanishing partial derivatives vanish give
\begin{gather}
    \frac{\partial\L_\pm}{\partial \alpha}  =  -4\sin^2{\theta}\sin{2\alpha}+\lambda \sec^2{\alpha}=0,\label{3dsa}\\
    \frac{\partial\L_\pm}{\partial \beta}  = 4\sin^2{\theta}\sin{(2\beta)} \pm \lambda \sec^2{\beta}\tan{\theta}=0,\label{3dsb}
\end{gather}
and, upon the elimination of $\lambda$, 
\begin{equation}\label{tantheta2}
    \tan{\theta}=\mp \frac{\sin{\beta}\cos^3{\beta}}{\sin{\alpha}\cos^3{\alpha}}.
\end{equation}

Together, ~\cref{tantheta2} and~\cref{tantheta1} imply that extreme values of $\SCHSH$, under the constraint of \cref{tantheta1}, occur when $\sin^2{2\alpha}=\sin^2{2\beta}$, i.e., $\beta \in \left\{\tfrac{n\pi}{2} \pm \alpha\right\}_{n \in \mathbb{Z}}$. Using this in~\cref{3zerosS1}, one sees that $\SCHSH$ can be larger than $2$ only if 
\begin{equation}\label{betaconst2}
    \beta \in \left\{(n+\tfrac{1}{2})\pi \pm \alpha\right\}_{n \in \mathbb{Z}}.
\end{equation}

Substituting~\cref{betaconst2} into~\cref{tantheta1} and \cref{3zerosS1}, we see that the optimization problem reduces to
\begin{equation}\label{opt3zeros}
    \begin{split}
        \max_{\theta,\alpha}\quad& 4\cos{2\alpha}\sin^2{\theta}+2\\
        \rm{s.t.}\quad&  \pm\tan^2{\alpha}=\tan{\theta},
    \end{split}
\end{equation}
where the sign in the constraint is positive if and only if the sign of $\zeta$ is opposite to that of $\alpha$ in \cref{betaconst2}.
By eliminating $\alpha$ using the constraint or otherwise, we get
\begin{equation}\label{eq_of_theta1}
    \pm\tan^3{\theta}+\tan^2{\theta}\pm\tan{\theta}-1 = 0,
    \end{equation}
whose solutions in the interval $[-\pi,\pi]$ are\footnote{The other solutions of $\theta\in[-\pi,\pi)$ are equivalent to the two presented in \cref{Eq_theta} up to a global phase factor, see \cref{state_two_zeros-in_same_block}.}
\begin{equation}\label{Eq_theta}
    \theta= \pm\theta_0, \quad\,\theta_0:=\tan^{-1}\tfrac{1}{3}\left( -1-\tfrac{2}{\tau}+\tau \right)\approx0.15851\pi
\end{equation}
where $\tau$ is defined in \cref{Eq:theta-tau} and the sign associated with $\tan\theta$ in \cref{eq_of_theta1} [$\theta_0$ in \cref{Eq_theta}] follows that of the constraint in \cref{opt3zeros}.
Accordingly, one finds from \cref{opt3zeros} that admissible solutions of $\alpha$ in the interval of $[0,2\pi)$ are
\begin{subequations}\label{Eq:alpha}
\begin{equation}
	\alpha= \alpha_0, \,\,\alpha_0+\pi \quad\text{and}\quad \alpha = \pi -\alpha_0, \,\,2\pi-\alpha_0
\end{equation}
where 
\begin{equation}
	\alpha_0:=\tan^{-1}\sqrt{|\tan{\theta_0}|}\approx0.20224\pi
\end{equation}
\end{subequations}
Finally, using \cref{betaconst2}, the extreme value of $\SCHSH$ are due to the following choice of $\beta$ in the interval $[0,2\pi)$:
\begin{equation}\label{Eq:beta}
	\beta= \frac{\pi}{2}\pm\alpha_0, \frac{3\pi}{2}\pm\alpha_0.
\end{equation}

It can be verified that all 32 legitimate combinations of $\theta, \alpha, \beta$, and $\zeta$  [see, respectively, \cref{Eq_theta}, \cref{Eq:alpha}, \cref{Eq:beta}, \cref{Eq:zeta} and the paragraph just below \cref{opt3zeros}] give rise to the same correlation $\vecPQ$ and hence the same maximal Bell value $\SCHSH=4 - 4 (2\kappa_1 + \kappa_2)\approx 2.26977$, where $\vecPQ$, $\kappa_1$ and $\kappa_2$ are defined in \cref{Eq_kappa}.

\subsection{Self-testing based on the maximally-CHSH-violating correlation in $\Q_{3b}$}
\label{App:Self-test:Class3b}

\begin{theorem}~\label{self-test_state&meas3}
If the CHSH value given in \cref{Eq:MaxCHSH:Class3b} is observed alongside $P(0,0|0,0)=P(1,1|0,0)=P(1,0|1,1)=0$, then the underlying state and measurements are equivalent to those given in \cref{Eq:Class3b} with \cref{Eq:3b:Parameters}
up to local isometries, i.e., $\vecPQ$ self-tests this quantum realization.
\end{theorem}

\begin{proof}
In \cref{App:MaxCHSH-Q3b}, we have characterized all the qubit strategies leading to the maximal CHSH  value attainable by $\vecP\in\Q_{3b}$. Next, we show that all these strategies are $2$-equivalent to the one that performs on the reference state
\begin{subequations}\label{Eq:Reference3b}
\begin{equation}\label{MDS3state}
    \ket{\widetilde{\psi}}=\cos{\theta_0}\ketA{0|0}\ketB{1|0}+\sin{\theta_0}\ketA{1|0}\ketB{0|0}
\end{equation}
the measurement [cf.~\cref{Eq:Obs->POVM}]
\begin{equation}\label{MDS3meas}
    \begin{split}
            \widetilde{A}_0=\sigma_z,  \quad
	    \widetilde{A}_1=\cos{2\alpha_0}\,\sigma_z-\sin{2\alpha_0}\,\sigma_x, \\
	    \widetilde{B}_0=\sigma_z, \quad
	    \widetilde{B}_1=-\cos{2\beta_0}\,\sigma_z-\sin{2\beta_0}\,\sigma_x,
    \end{split}
\end{equation}
\end{subequations}
where $\theta_0, \alpha_0$ are defined, respectively, in \cref{Eq_theta} and \cref{Eq:alpha}, $\beta_0=\frac{\pi}{2}-\alpha_0$, while $\ketA{0,1|0}$ and $\ketB{0,1|0}$ are taken as the $\pm1$ eigenstate of the Pauli matrix $\sigma_z$ in \cref{MDS3meas}.

To prove the lemma, we first note that the quantum strategy specified in \cref{Eq:Reference3b} is indeed an example of those characterized in \cref{App:MaxCHSH-Q3b} with $\theta=\theta_0$, $\alpha=\alpha_0$, $\beta=\beta_0$, and $\phi=\gamma_A=\gamma_B=\omega_A=\omega_B=0$. Next, note that for a general choice of these parameters, we have
\begin{equation}\label{Eq:3b:genericobservables}
    \begin{split}
    	&\qquad\quad A_0 = B_0 = \sigma_z=\begin{bmatrix}
            1 & 0 \\
            0 & -1
            \end{bmatrix},\\
	    {A}_1&=\begin{bmatrix}
            \cos2\alpha & -e^{i(\gamma_A+\omega_A)}\sin2\alpha \\
            -e^{-i(\gamma_A+\omega_A)}\sin2\alpha & -\cos2\alpha
            \end{bmatrix},\\
	    {B}_1&=\begin{bmatrix}
            -\cos2\beta & -e^{-i(\gamma_B+\omega_B)}\sin2\beta \\
            -e^{i(\gamma_B+\omega_B)}\sin2\beta & \cos2\beta
            \end{bmatrix}.
    \end{split}
\end{equation}
Moreover, for any fixed value of $\gamma_A, \gamma_B, \omega_A, \omega_B,$ and $\phi$, all the 32 strategies characterized in \cref{App:MaxCHSH-Q3b} related to different choices of $\theta$, $\alpha$, and $\beta$ only lead to different combination of signs in $\tan\theta, \sin2\alpha, \cos2\alpha, \sin2\beta$, and $\cos2\beta$. The strategy of \cref{Eq:Reference3b}, in particular, corresponds to a choice of positive sign for all these parameters.

With these observations in mind, it is not difficult to verify that the unitaries 
\begin{equation}~\label{uab3_self-test}
    \begin{split}
        u_A&=\begin{bmatrix}
    e^{-i\gamma_A}\frac{\cos{\alpha}}{|\cos{\alpha}|} & 0 \\
    0 & e^{i\omega_A}\frac{\sin{\alpha}}{|\sin{\alpha}|}
    \end{bmatrix},\\
    u_B&=\begin{bmatrix}
     e^{i\gamma_B}\frac{\cos{\beta}}{|\cos{\beta}|} & 0 \\
    0 & e^{-i\omega_B}\frac{\sin{\beta}}{|\sin{\beta}|}
    \end{bmatrix},
    \end{split}
\end{equation}
transform the state of \cref{state_two_zeros-in_same_block} to the reference state of \cref{MDS3state} via \cref{Eq:Equivalent:State}, while ensuring that the observables of \cref{Eq:3b:genericobservables} transform to those of \cref{MDS3meas} via \cref{Eq:Equivalent:Meas} for all $a,x,b,y$. To arrive \cref{Eq:Equivalent:State}, we have also made use of the relation of the sign of $\zeta$ and $\tan\theta$, as mentioned above \cref{eq_of_theta1}. This completes the proof that all qubit strategies giving the maximal CHSH  value in $\Q_{3b}$ are $2$-equivalent.

The last ingredient that we need is the observation that $\vecPQ$ is extremal in $\Q$, which follows from Fact~\ref{CHSH:ExtremeQ} and the fact in our characterization (see~\cref{App:MaxCHSH-Q3b}) of {\em all} qubit strategies leading to the maximal CHSH violation in Class 3b, {\em only} $\vecPQ$ leads to this largest violation while complying with the constraint of Class 3b (cf. \cref{tab:class3b}). A direct application of \cref{Thm:self-test:upgrade} using the extremality of $\vecPQ$ and the just-established $2$-equivalence of all qubit strategies leading to this maximum completes the proof.
\end{proof}

\section{Quantum correlation giving maximal CHSH value in $\Q_{2b}$ is a self-test}
\label{app-st2}

\subsection{Maximal CHSH value for $\vecP\in\Q_{2b}$}
\label{App:MaxCHSH-Q2b}

With the measurement bases defined in \cref{Eq_MeasBases} and \cref{Eq_Observables}, we may write a general pure two-qubit state $\ket{\psi}$ as
\begin{equation}
    \ket{\psi}=\sum_{i=0}^{1}\left(c_{i0}\ketA{i|0}\ketB{0|0}+c_{i1}\ketA{i|1}\ketB{1|0}\right),
\end{equation}
where $\sum_{i,j=0}^{1}|c_{ij}|^2=1$. Since $\vecP\in\Q_{2b}$ satisfies $P(0,0|0,0)=P(1,1|1,0)=0$, these constraints imply, respectively, that  $c_{00}=c_{11}=0$. We may thus restrict our attention to states of the form:
\begin{equation}\label{state_two_zeros-state}
    \ket{\psi}=\cos{\theta}\ketA{1|0}\ketB{0|0}+e^{i\phi}\sin{\theta}\ketA{0|1}\ketB{1|0}.
\end{equation}

To proceed, let the local measurement bases be related via \cref{Eq:Transformation} by
\begin{equation}\label{uab2}
\begin{split}
	 U_A=\begin{bmatrix}
    e^{-i\gamma_A}\cos{\alpha} & e^{-i\omega_A}\sin{\alpha} \\
    -e^{i\omega_A}\sin{\alpha} & e^{i\gamma_A}\cos{\alpha}
    \end{bmatrix},\\
	 U_B=\begin{bmatrix}
    e^{-i\gamma_B}\cos{\beta} & e^{-i\omega_B}\sin{\beta} \\
     -e^{i\omega_B}\sin{\beta} & e^{i\gamma_B}\cos{\beta}
    \end{bmatrix}.
\end{split}
\end{equation}

Using \cref{state_two_zeros-state} and \cref{uab2} in \cref{Eq:CHSH} and \cref{Eq:Born}, we arrive at the  Bell value:
\begin{align}~\label{S_class2b}
        \SCHSH &= 2-4\sin^2{\alpha}(\cos^2{\theta}\sin^2{\beta}+\sin^2{\theta}\cos^2{\beta})\nonumber\\
        &+2\zeta\sin{2\theta}\sin{\alpha}\sin{2\beta}
\end{align}
where $\zeta:= \cos(\omega_B+\gamma_B-\phi-\omega_A)$.

Since the variables $\phi, \omega_A, \omega_B, \gamma_B$ appear in $\SCHSH$ only through $\zeta$, it is clear that for $\SCHSH$ to arrive at its largest value, we must have $\zeta$ taking its extremal value of $\pm1$, in accordance with the sign of $\sin{2\theta}\sin{\alpha}\sin{2\beta}$. We denote the resulting the Bell value, respectively, by
\begin{align}~\label{S_class2bpm}
        \S_\pm &= 2-4\sin^2{\alpha}(\cos^2{\theta}\sin^2{\beta}+\sin^2{\theta}\cos^2{\beta})\nonumber\\
        &\pm2\sin{2\theta}\sin{\alpha}\sin{2\beta}.
\end{align}
For a Bell violation to be possible (i.e., for either of $\S_\pm$ to be larger than the local bound of $2$), the local measurements should not commute and the state must be entangled. It can then be verified from \cref{state_two_zeros-state} and \cref{S_class2bpm} that the remaining variables must again satisfy \cref{Eq_constraints}.

To find the largest value of $\SCHSH$, we  proceed using standard variational techniques by first computing the partial derivatives of $\S_\pm$ with respect to $\alpha$, $\beta$, $\theta$ and setting each of them to zero to get, correspondingly,
\begin{gather}
    \sin{\alpha}=\pm \frac{\sin{2\theta}\sin{2\beta}}{4(\cos^2{\beta}\sin^2{\theta}+\sin^2{\beta}\cos^2{\theta})},\label{dsa}\\
    \pm \sin{2\theta}\cos{2 \beta}-\sin{\alpha}\sin{2\beta}\cos2\theta=0,\label{dsb}\\
    \pm\cos{2\theta}\sin{2 \beta}-\sin{\alpha}\sin{2\theta}\cos2\beta=0.\label{dst}
\end{gather}
By substituting \cref{dsa} into \cref{dsb} and \cref{dst}, and upon reduction, one eventually arrives at
\begin{equation}\label{const1}
    |\sin{\theta}|=|\cos{\theta}|=|\sin{\beta}|=|\cos{\beta}|=\frac{1}{\sqrt{2}},
\end{equation}
and hence $\sin{\alpha}=\pm \frac{1}{2}\sin{2\beta}\sin{2\theta}$. 
By considering all possible values of $\theta, \alpha, \beta\in [0,2\pi)$ meeting these constraints, one finds that the maximal value of \cref{S_class2b} is $\SCHSH=\frac{5}{2}$, which is attainable 
{\em iff} $\alpha\in \{\frac{\pi}{6},\frac{5\pi}{6},\frac{7\pi}{6},\frac{11\pi}{6}\}$, $\beta\in \{ \frac{\pi}{4},\frac{3\pi}{4},\frac{5\pi}{4},\frac{7\pi}{4}\}$, and $\theta\in \{ \frac{\pi}{4},\frac{3\pi}{4},\frac{5\pi}{4},\frac{7\pi}{4}\}$. By explicit calculations, it can  be verified that all these strategies give exactly $\vecPQiii$. Since this is the only correlation in $\Q_3$ giving this maximal value, $\vecPQiii$ is an extreme point of $\Q$.

\subsection{Self-testing based on the maximally-CHSH-violating correlation in $\Q_{2b}$}
\label{App:Self-test:Class2b}

\begin{theorem}~\label{self-test_state&meas2}
If the CHSH value of $\frac{5}{2}$ is observed alongside the zeros shown in~\cref{tab:class2b}, then the underlying state and measurements are equivalent to those given in \cref{Eq:Class2b} and \cref{Eq:2b:Parameters} up to local isometries, i.e., $\vecPQiii$ self-tests this particular quantum realization.
\end{theorem}

\begin{proof}
From the unitaries given in~\cref{uab2}, the state of~\cref{state_two_zeros-state} can be written explicitly in the bases of $\{\ketA{0|0},\ketA{1|0}\}$ and $\{\ketB{0|0},\ketB{1|0}\}$ as
\begin{align}~\label{generalstate}
    \ket{\psi}&=\cos{\theta}\ketA{1|0}\ketB{0|0} \\
    &+e^{i\phi}\sin{\theta}(e^{-i\gamma_A} \cos{\alpha}\ketA{0|0}-e^{i\omega_A}\sin{\alpha}\ketA{1|0})\ketB{1|0},\nonumber
\end{align}
whereas the corresponding observables read as:
\begin{equation}\label{Eq:2b:genericobservables}
    \begin{split}
    	&\qquad\quad A_0 = B_0 = \sigma_z=\begin{bmatrix}
            1 & 0 \\
            0 & -1
            \end{bmatrix},\\
	    {A}_1&=\begin{bmatrix}
            \cos2\alpha & -e^{-i(\gamma_A+\omega_A)}\sin2\alpha \\
            -e^{i(\gamma_A+\omega_A)}\sin2\alpha & -\cos2\alpha
            \end{bmatrix},\\
	    {B}_1&=\begin{bmatrix}
            \cos2\beta & -e^{-i(\gamma_B+\omega_B)}\sin2\beta \\
            -e^{i(\gamma_B+\omega_B)}\sin2\beta & -\cos2\beta
            \end{bmatrix}.
    \end{split}
\end{equation}
The quantum strategy of \cref{Eq:Class2b} with \cref{Eq:2b:Parameters} is simply a special case of the above corresponding to $\alpha=\frac{\pi}{6}$, $\beta=\theta=\frac{\pi}{4}$, $\gamma_A=\gamma_B=\omega_A=\omega_B=\phi=0$, giving
\begin{subequations}\label{Eq:Reference2b}
\begin{align}
    \ket{\widetilde{\psi}}&=\frac{1}{\sqrt{2}}\left(\frac{\sqrt{3}}{2}\ket{0}\!\ket{1}+\ket{1}\!\ket{0}-\frac{1}{2}\ket{1}\!\ket{1}\right),\label{MDS2state}\\
    \begin{split}\label{MDS2meas}
            &\widetilde{A}_0=\sigma_z,  \quad
	    \widetilde{A}_1=\frac{1}{2}\sigma_z-\frac{\sqrt{3}}{2}\,\sigma_x, \\
	    &\widetilde{B}_0=\sigma_z, \quad
	    \widetilde{B}_1=-\,\sigma_x.
    \end{split}
\end{align}
\end{subequations}

Now, notice that for any fixed value of $\gamma_A, \gamma_B, \omega_A, \omega_B,$ and $\phi$, all the 64 strategies characterized in \cref{App:MaxCHSH-Q2b} related to different choices of $\theta$, $\alpha$, and $\beta$ only lead to different combination of signs in $\tan\theta, \sin2\alpha, \cos2\alpha, \sin2\beta$, and $\cos2\beta$. The strategy of \cref{Eq:Reference2b}, in particular, corresponds to a choice of positive sign for all these parameters.

It is then easy to verify that the following unitaries
\begin{equation}~\label{uab2_self-test}
    \begin{split}
        u_A&=\begin{bmatrix}
    e^{i\gamma_A}\frac{\cos{\alpha}}{|\cos{\alpha}|} & 0 \\
    0 & e^{-i\omega_A}\frac{\sin{\alpha}}{|\sin{\alpha}|}
    \end{bmatrix},\\
    u_B&=\begin{bmatrix}
     e^{i\gamma_B}\frac{\cos{\beta}}{|\cos{\beta}|} & 0 \\
    0 & e^{-i\omega_B}\frac{\sin{\beta}}{|\sin{\beta}|}
    \end{bmatrix}
    \end{split}
\end{equation}
establish the $2$-equivalence of all these qubit strategies. In other words, they transform, via \cref{Eq:Dfn:d-equivalent}, the state of \cref{generalstate} and the observables of \cref{Eq:2b:genericobservables}, respectively, to the reference state of \cref{MDS2state} and reference observables of \cref{MDS2meas} for all $a,x,b,y$. Note that in arriving at \cref{MDS2state}, we have also made use of the fact that the sign of $\zeta= \cos(\omega_B+\gamma_B-\phi-\omega_A)$ follows that of of $\sin{2\theta}\sin{\alpha}\sin{2\beta}$.  

Finally, a straightforward application of \cref{Thm:self-test:upgrade} using the just-established $2$-equivalence and the fact that $\vecPQiii$ is extremal in $\Q$ (see the end of \cref{App:MaxCHSH-Q2b}) completes the proof of the theorem.
\end{proof}

\section{Quantum correlation giving maximal CHSH value in $\Q_1$ is a self-test}\label{App:Self-test:Class1}

To find the maximal CHSH violation for correlations in this Class, one may follow the procedure given in \cref{App:MaxCHSH-Q3b} or \cref{App:MaxCHSH-Q2b} by considering general qubit unitary matrices of unit determinant connecting the different measurement bases. For simplicity, however, we note that in the maximization of $\SCHSH$, it suffices to consider extreme points in $\Q$. By Fact~\ref{CHSH:ExtremeQ2}, it thus suffices to consider observables of the form of \cref{Eq:Meas:Class1}.
Consider now the Bell operator~\cite{Braunstein1992} corresponding to the Bell inequality of \cref{Eq:CHSH}
\begin{equation}
    \begin{split}
        \B &= -A_0\otimes B_0-A_1\otimes B_0-A_0\otimes B_1+ A_1\otimes B_1.
    \end{split}
\end{equation}
The maximal value of $\SCHSH$ is obtained by measuring the observables of \cref{Eq:Meas:Class1} on the eigenstate $\ket{\psi}$ of $\mathcal{B}$ giving the largest eigenvalue.

For correlations from Class 1, however, we also have the constraint $P(0,0|0,0)=|\braket{00|\psi}|^2=0$. Thus, to find the maximal CHSH value achievable by $\vecP\in\Q_1$, we may consider $\ket{\psi}$ to be an eigenstate of the principle $3\times 3$ sub-matrix $\B'$ of $\B$ related to the span of \{$\ket{10},\ket{01},\ket{11}$\}, i.e., 
\begin{equation}
    \B'=
\begin{bmatrix}
\tau  & S_\alpha S_\beta &  (1- C_\beta)S_\alpha \\
 S_\alpha s_\beta & \tau & (1-C_\alpha)S_\beta \\
 (1- C_\beta)S_\alpha & (1-C_\alpha)s_\beta & -\tau \\
\end{bmatrix}.
\end{equation}
where  $\tau:=1+C_\alpha+C_\beta-C_\alpha C_\beta$,  $C_\alpha=\cos\alpha$,  and $S_\beta=\sin\beta$, etc.

To compute the eigenvalue of $\B'$, note that its characteristic polynomial is 
\begin{equation}
   -\lambda ^3 +\tau\lambda ^2 +4 \lambda - 4 C_\alpha^2S_\beta^2+4 C_\alpha (C_\beta-1)-4 C_\beta (C_\beta+1),
\end{equation}
where
$\lambda$ is the eigenvalue. Using $\mathsf{Mathematica}$, we find that the largest value of $\SCHSH$, i.e., the largest real value of $\lambda$  is $\xi_3 \approx 2.6353$ with the unique parameters $C_\alpha = C_\beta=\xi_1 \approx 0.2862$, and the corresponding eigenvector is
\begin{equation}
    \ket{\psi} = \frac{\sin{\theta}}{\sqrt{2}}\left(\frac{\sin{\alpha}}{|{\sin{\alpha}}|}\ket{0}\!\ket{1}+\frac{\sin{\beta}}{|{\sin{\beta}}|}\ket{1}\!\ket{0}\right)+\cos{\theta}\ket{1}\!\ket{1},
\end{equation}
where $\xi_1$, $\xi_2$ and $\xi_3$ are, respectively, the smallest positive roots of the cubic polynomials given in \cref{Eq:Poly}.

\begin{theorem}~\label{self-test_state&meas1}
If the CHSH value of $\xi_3$ is observed alongside $P(0,0|0,0)=0$, then the underlying state and measurements are equivalent to those given in \cref{Eq:Class1}
up to local isometries, i.e., $\vecPQiv$ self-tests this particular quantum realization.
\end{theorem}
\begin{proof}
Since the only remaining freedom is the sign of $\sin{\alpha}$ and $\sin{\beta}$, it can be easily seen that the local unitary $\sigma_z$ can be used to generate the required sign. Therefore, the correlation $\vecPQiv$ which is the only one giving the maximal value of $\SCHSH$ in Class 1 is an extreme point of $\Q$ and any quantum strategy realizing this correlation is $2$-equivalent to the reference state and measurements given in \cref{Eq:Class1}.\footnote{Strictly, the current proof only applies to all quantum measurements of the form given in \cref{Eq:Meas:Class1}. However, even more general measurements involving complex state and observables, such as those considered in the proof in  \cref{App:MaxCHSH-Q3b} or \cref{App:MaxCHSH-Q2b} may be shown to be $2$-equivalent to the reference strategy in a similar manner. }
An application of~\cref{Thm:self-test:upgrade} then completes the proof.
\end{proof}

\onecolumngrid

\section{Robustness results}
\label{App:Robust}

Here, we provide the remaining plots that illustrate the level of robustness of the self-testing results established for $\vecP_\text{Q}$ of \cref{Eq:vecPQ} from Class 3b, $\vecP_\text{Q,2}$ of \cref{Eq:PQ2} from Class 2a, $\vecPQiii$ of \cref{Eq:PQ3} from Class 2b, $\vecPC$ of \cref{Eq:PCabello} from Class 2c, and $\vecPQiv$ of \cref{Eq:PQ4} from Class 1. In other words, for any given value $\SCHSH$ of the (suboptimal) CHSH Bell violation and the allowed tolerance $\varepsilon$ from the zero probabilities, we show the best numerically determined lower bound on (1) the fidelity with respect to the reference state, \cref{Eq:Fidelity}, and (2) an appropriate figure of merit, \cref{Robust_ST_Meas}, that quantifies the similarity with the target measurements. Further details about these figures of merit can be found in~\cref{App:robust_self-test}.

\begin{figure}[h!tbp]
\centering
    \captionsetup{justification=RaggedRight,singlelinecheck=off}
  \includegraphics[width=0.45\linewidth]{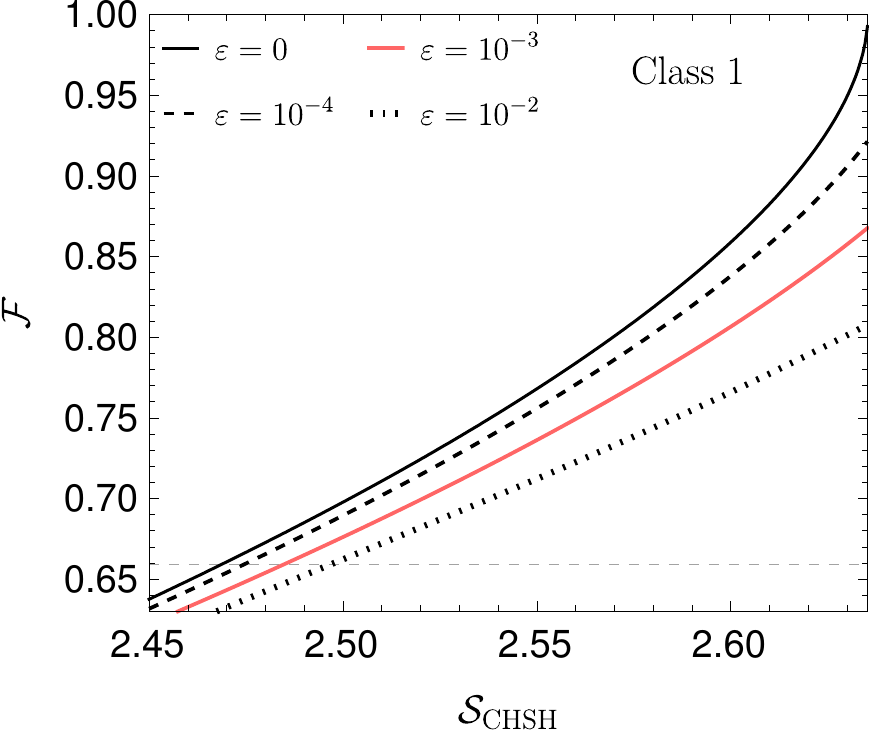}\hspace{1cm}
  \includegraphics[width=0.45\linewidth]{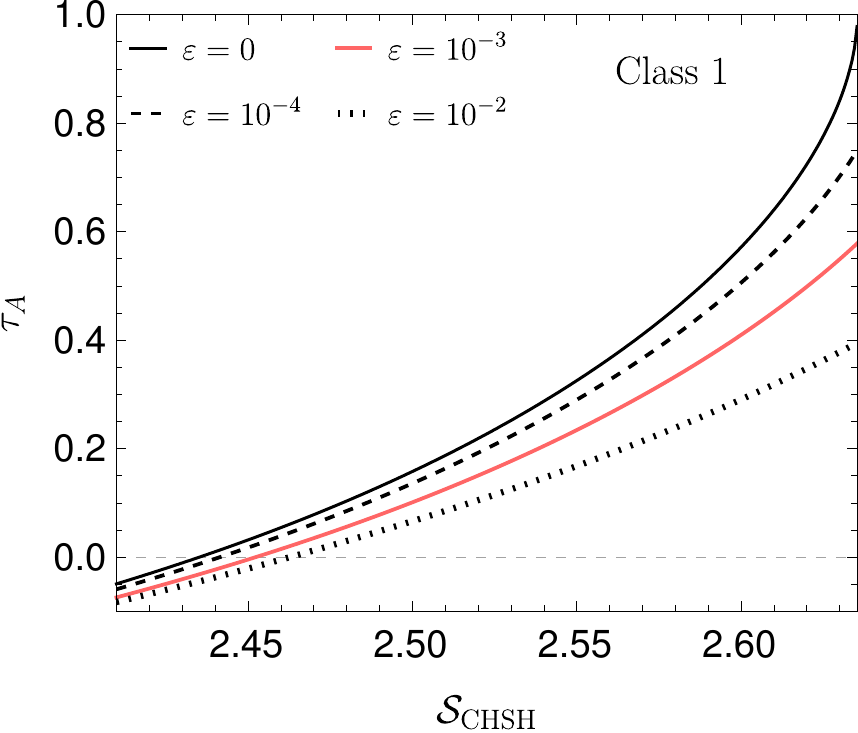}
  \caption{
  Plots illustrating the robustness of the self-testing of \cref{Eq:Class1} based on the correlation $\vecPQiv$ of \cref{Eq:PQ4} from Class 1. 
  From left to right, we show, respectively,  as a function of the Bell value $\SCHSH$, a lower bound on \cref{Eq:Fidelity} for self-testing the reference two-qubit state and a lower bound on \cref{Robust_ST_Meas} for self-testing both parties' observables.}
  \label{fig:selftest:class1}
\end{figure}

\twocolumngrid

\begin{figure}[h!tbp]
    \centering
    \captionsetup{justification=RaggedRight,singlelinecheck=off}
  \includegraphics[width=0.9\linewidth]{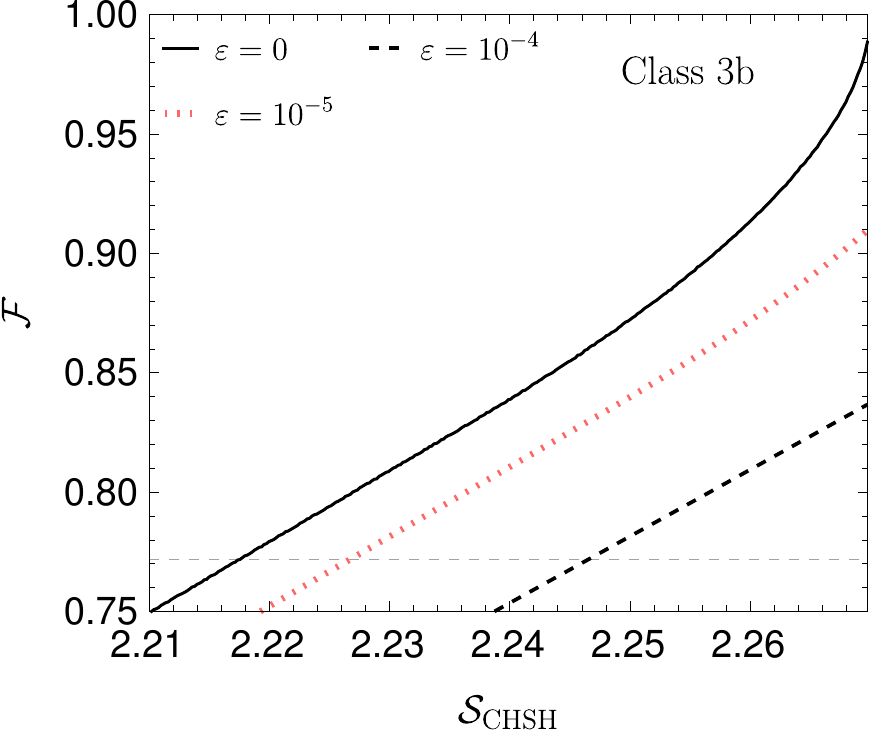}\vspace{0.5cm}
  \includegraphics[width=0.9\linewidth]{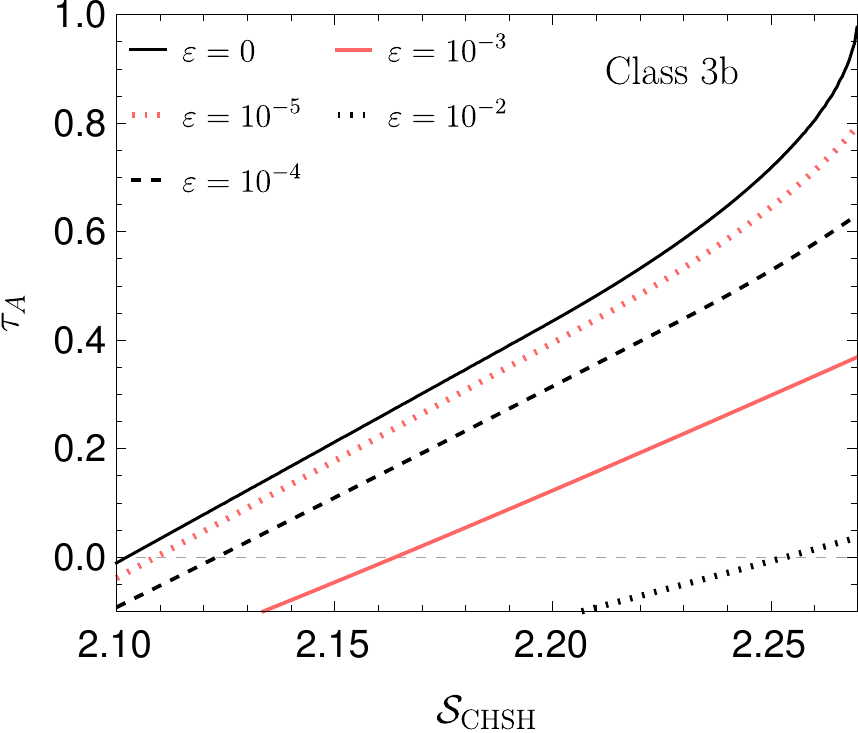}
  \includegraphics[width=0.9\linewidth]{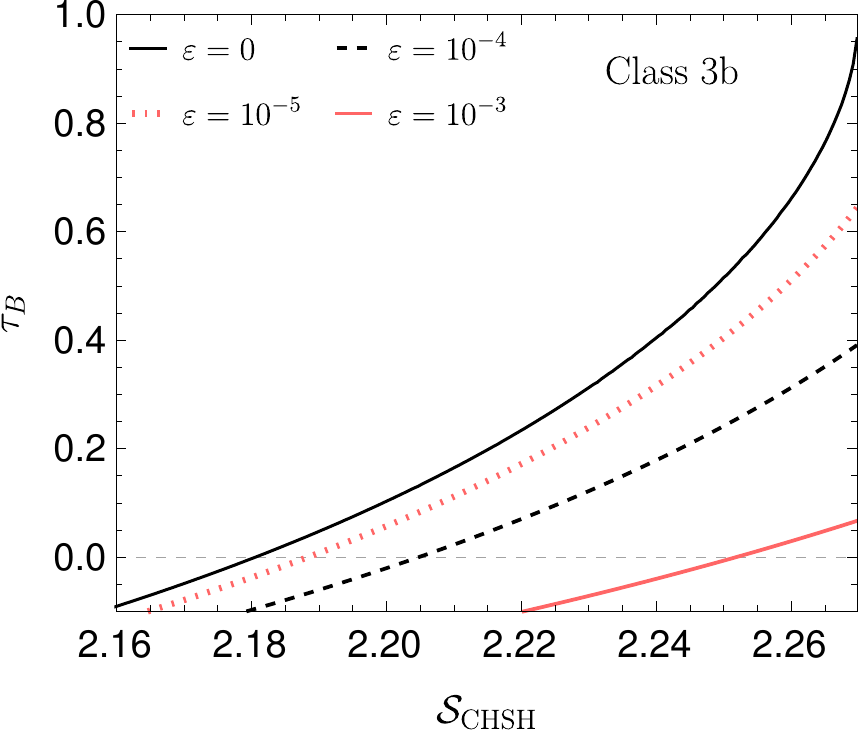}
  \caption{Plots illustrating the robustness of the self-testing of \cref{Eq:Class3b} and  \cref{Eq:3b:Parameters} based on the correlation $\vecP_\text{Q}$ of \cref{Eq:vecPQ} from Class 3b.  From top to bottom, the plots show, as a function of the Bell value $\SCHSH$, a lower bound on \cref{Eq:Fidelity} for self-testing the reference two-qubit state and a lower bound on \cref{Robust_ST_Meas} for self-testing Alice's and Bob's observables. For the significance of the dashed horizontal line and other details related to the plots, see the caption of~\cref{fig:selftest:class3a}.}
  \label{fig:selftest:class3b}
\end{figure}

\begin{figure}[h!tbp]
\centering
    \captionsetup{justification=RaggedRight,singlelinecheck=off}
  \includegraphics[width=0.9\linewidth]{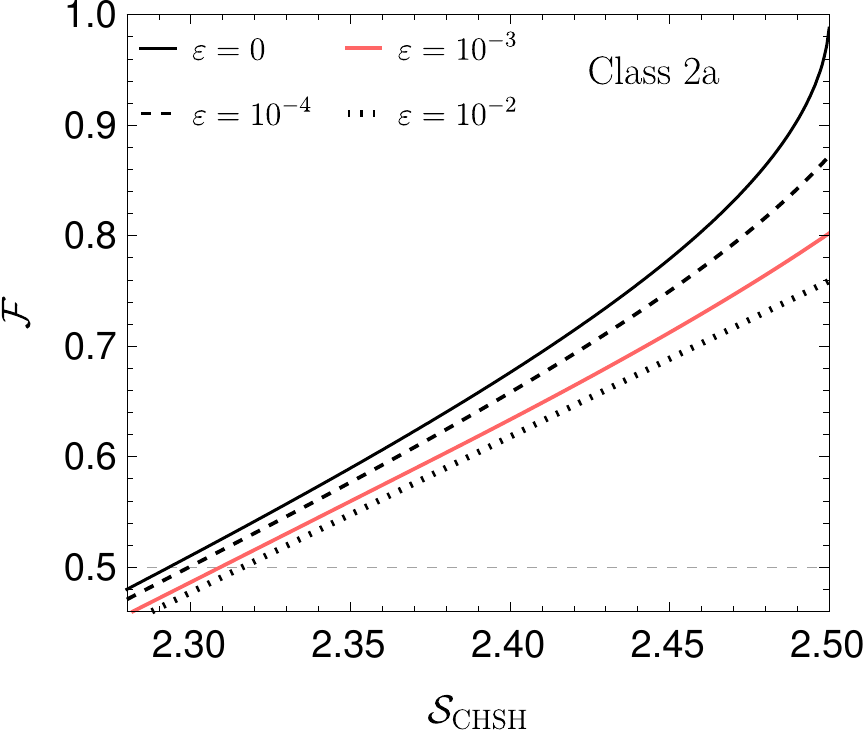}\vspace{0.5cm}
  \includegraphics[width=0.9\linewidth]{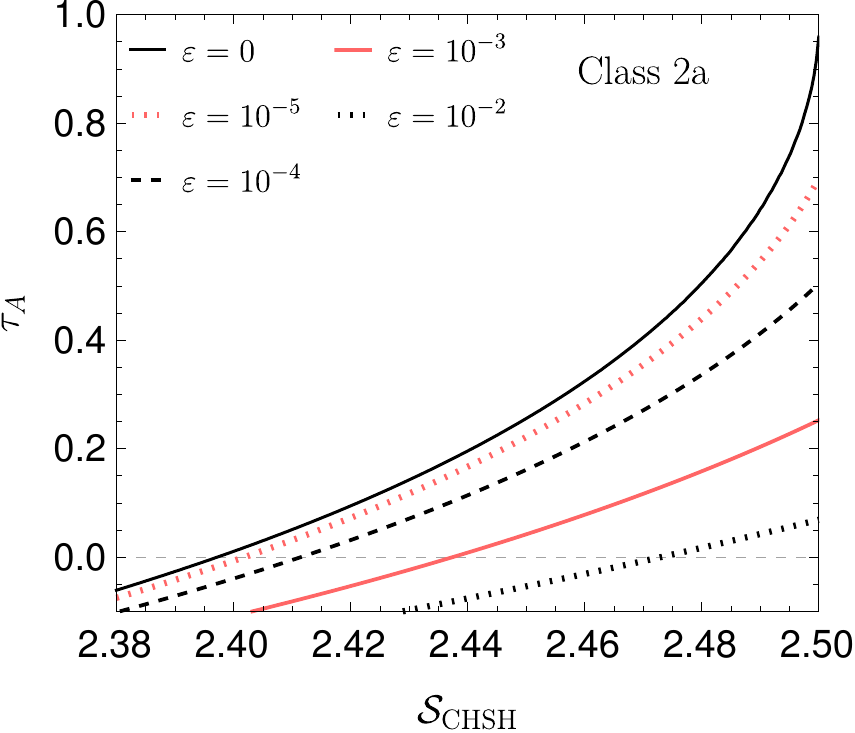}
  \includegraphics[width=0.9\linewidth]{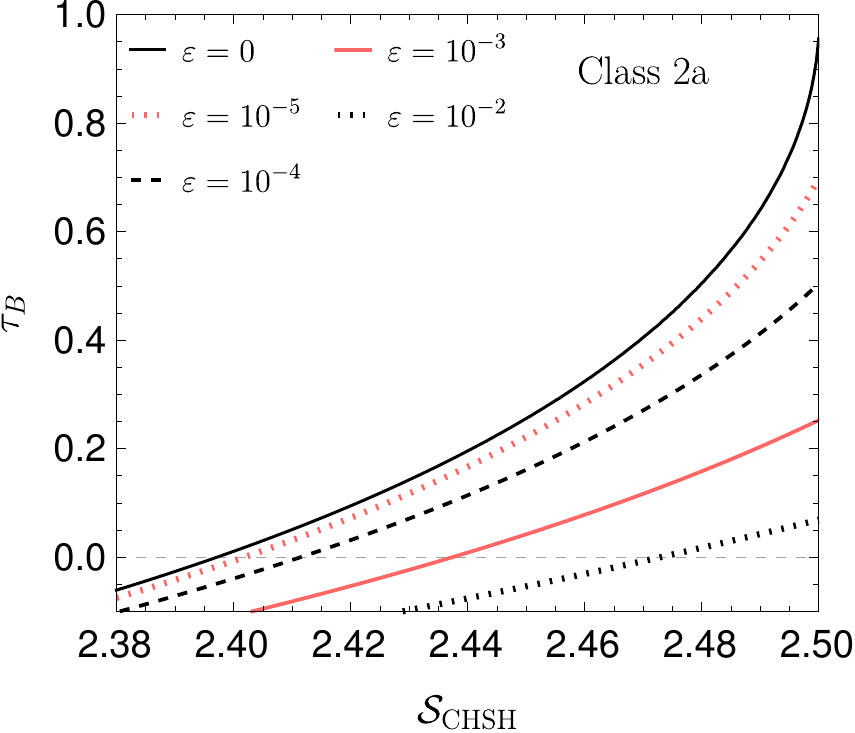}
  \caption{Plots illustrating the robustness of the self-testing of \cref{Eq:Class2a} based on the correlation $\vecP_\text{Q,2}$ of \cref{Eq:PQ2} from Class 2a. 
  From top to bottom, the plots show, as a function of the Bell value $\SCHSH$, a lower bound on \cref{Eq:Fidelity} for self-testing a Bell state and a lower bound on \cref{Robust_ST_Meas} for self-testing Alice's and Bob's observables. For the significance of the dashed horizontal line and other details related to the plots, see the caption of~\cref{fig:selftest:class3a}.} 
  \label{fig:selftest:class2a}
\end{figure}

\begin{figure}[h!tbp]
\centering
    \captionsetup{justification=RaggedRight,singlelinecheck=off}
  \includegraphics[width=0.9\linewidth]{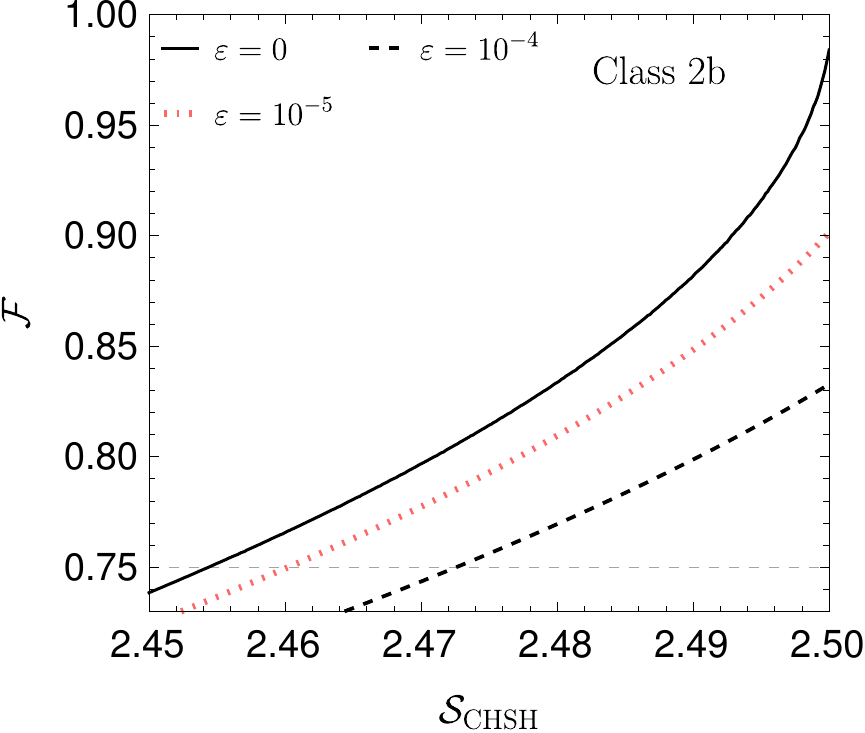}\vspace{0.5cm}
  \includegraphics[width=0.9\linewidth]{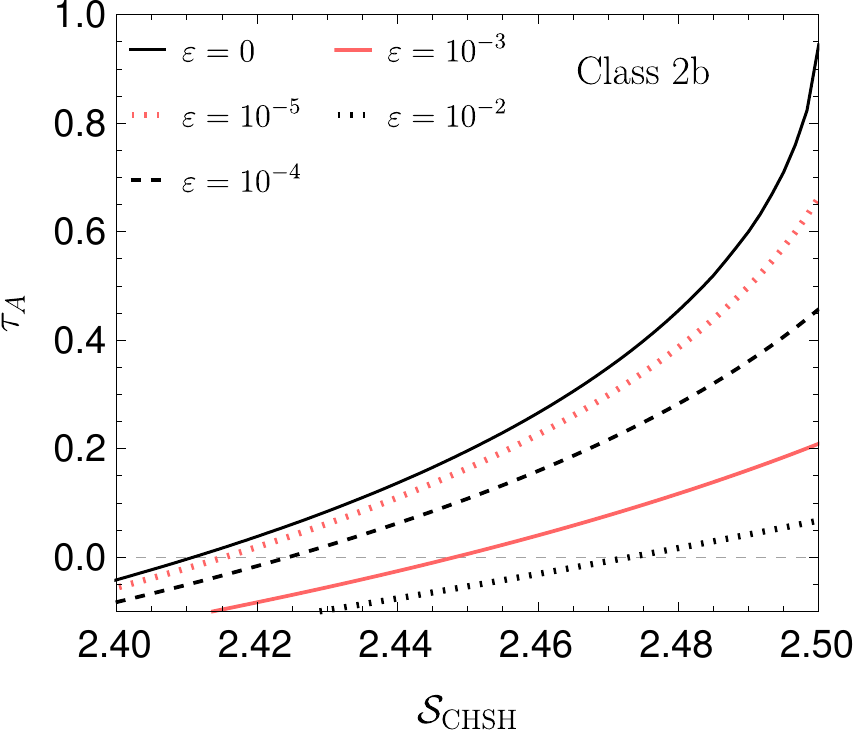}
  \includegraphics[width=0.9\linewidth]{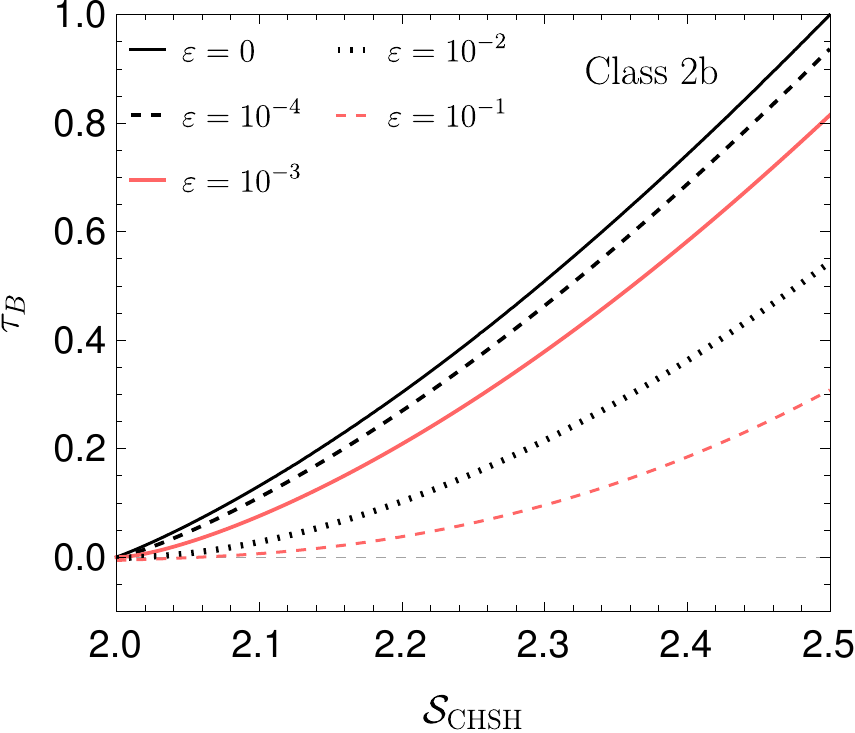}
  \caption{
  Plots illustrating the robustness of the self-testing of \cref{Eq:Class2b} and \cref{Eq:2b:Parameters} based on the correlation $\vecPQiii$ of \cref{Eq:PQ3} from Class 2b. 
  From top to bottom, the plots show as a function of the Bell value $\SCHSH$, a lower bound on \cref{Eq:Fidelity} for self-testing the reference two-qubit state and a lower bound on \cref{Robust_ST_Meas} for self-testing Alice's and Bob's observables. For the significance of the dashed horizontal line, and other details related to the plots, see the caption of~\cref{fig:selftest:class3a}.}
  \label{fig:selftest:class2b}
\end{figure}

\begin{figure}[h!tbp]
\centering
    \captionsetup{justification=RaggedRight,singlelinecheck=off}
  \includegraphics[width=0.9\linewidth]{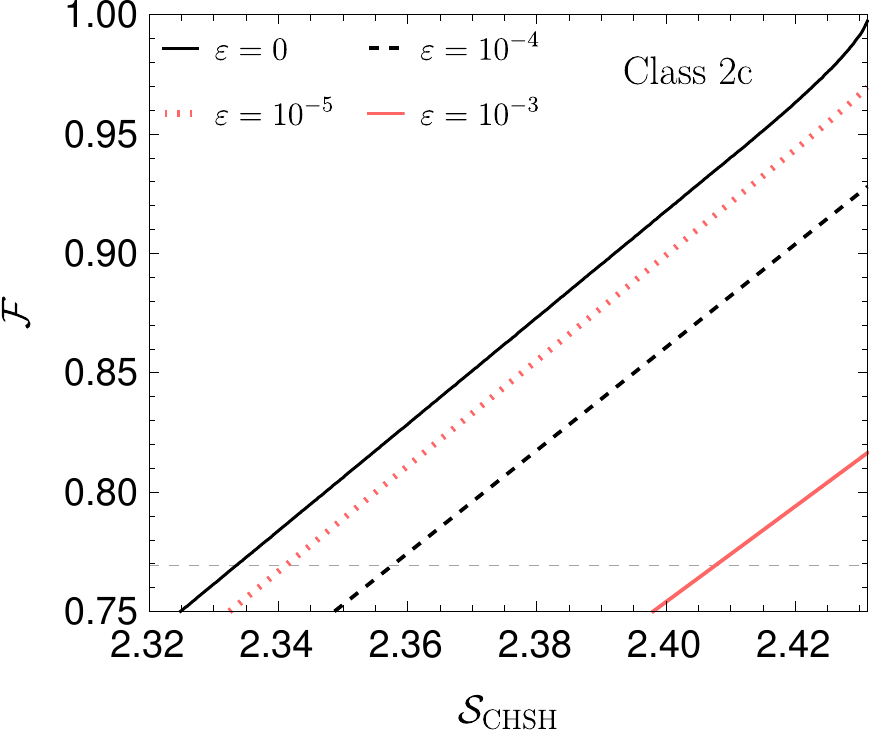}\vspace{0.5cm}
  \includegraphics[width=0.9\linewidth]{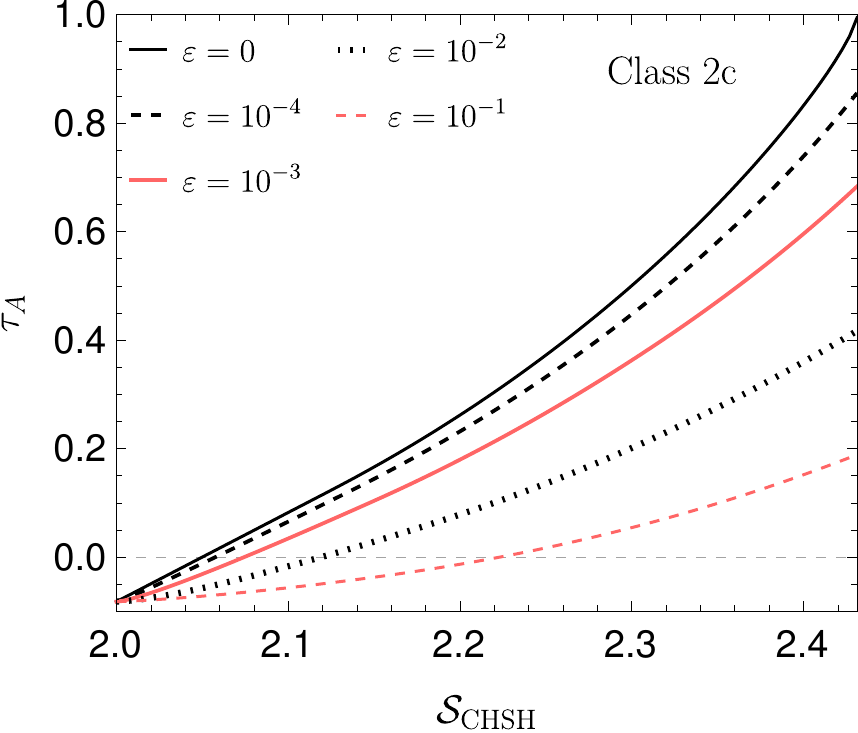}
  \includegraphics[width=0.9\linewidth]{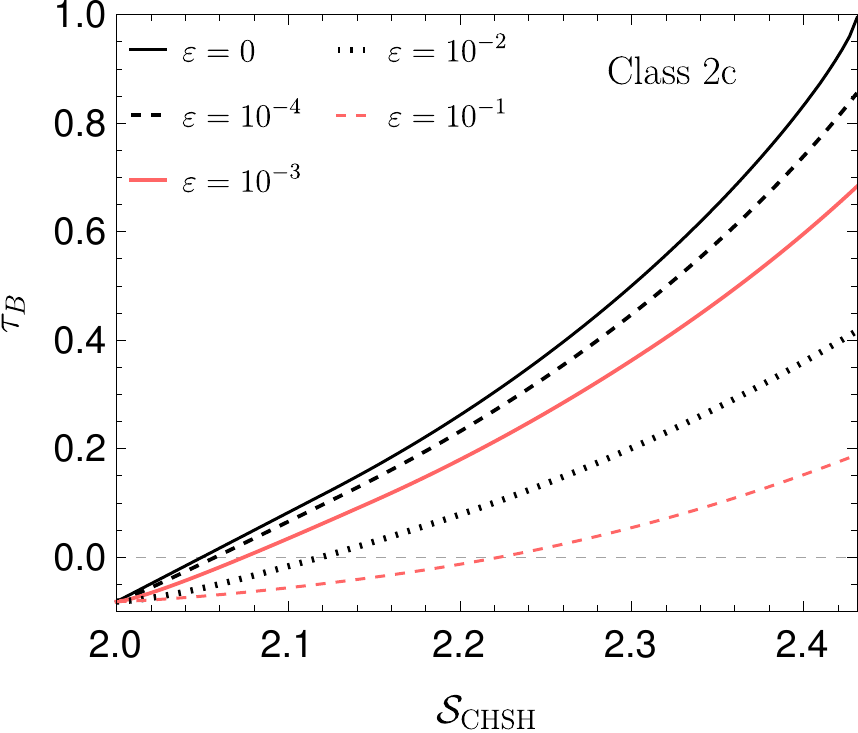}
  \caption{
  Plots illustrating the robustness of the self-testing of \cref{Eq:Class2c} with \cref{Eq:k} based on the correlation $\vecPC$ of \cref{Eq:PCabello} from Class 2c. 
  From top to bottom, the plots show as a function of the Bell value $\SCHSH$, a lower bound on \cref{Eq:Fidelity} for self-testing the reference two-qubit state and a lower bound on \cref{Robust_ST_Meas} for self-testing Alice's and Bob's observables. For the significance of the dashed horizontal line and other details related to the plots, see the caption of~\cref{fig:selftest:class3a}.}
  \label{fig:selftest:class2c}
\end{figure}

\section{Non-exposed quantum extreme points on the no-signaling boundary}

To arrive at the non-exposed property of the various example correlations given in \cref{Sec:Q-NS}, we apply the method described in Appendix H of Ref.~\cite{Goh2018}, i.e., we show that for each of the self-testing $\vecP$ described below, there is {\em no} Bell function  $\vec{B}$ such that it is the {\em unique} maximizer. For this matter, it suffices to show that for any Bell function where $\vecP$ is a maximizer, a local extreme point gives also the same Bell value. 

To this end, consider a vector $\vec{M}$ with components
\begin{equation}\label{OpeFunc}
	M(a,b|x,y)= M^A_{a|x}\otimes M^B_{b|y}.
\end{equation}
Then a Bell operator~\cite{Braunstein1992} can be written as
\begin{equation}\label{Eq:BellOperator}
    \B = \vec{B}\cdot \vec{M},
\end{equation}
where $\vec{B}$ is a vector of real numbers associated with an arbitrary Bell function.

Now, let $\vec{B}$ be a Bell function maximized by the self-testing correlation $\vec{P}$. Since the corresponding reference state $\ket{\widetilde{\psi}}$ is an eigenstate of the Bell operator $\B$, it follows from \cref{Eq:BellOperator} that any state vector $\ket{\phi}$ orthogonal to $\ket{\widetilde{\psi}}$, i.e., satisfying $\braket{\phi|\widetilde{\psi}}=0$, will result in a vector
\begin{equation}\label{Eq:vecT}
    \vec{T} \equiv \bra{\phi} \vec{M} \ket{\psi}
\end{equation}
orthogonal to $\vec{B}$, giving $\vec{B}\cdot\vec{T}=0$. Note also from \cref{OpeFunc}, \cref{Eq:vecT}, and the completeness relations of POVM that $\sum_{a,b} T(a,b|x,y)=0$ for all $x$ and $y$.

In our work, all $\ket{\widetilde{\psi}}$ are two-qubit states, so one can always find three $\ket{\phi}$ orthogonal to each reference state. For concreteness, we label all the sixteen local deterministic extreme points by $ j = 0,1,\ldots, 15$ such that
\begin{equation}
    P_{d,j}(a,b|x,y)=\delta_{a,(j\bmod{2})x\oplus \lfloor\frac{j\bmod{4}}{2}\rfloor} \delta_{b,\lfloor\frac{j\bmod{8}}{4}\rfloor y\oplus\lfloor\frac{j}{8}\rfloor}.
\end{equation}
Then, if the optimum value of the following linear program 
\begin{subequations}\label{Eq:Non-Exposed-LPs}
\begin{equation}\label{Non-Exposed_LinearProgram}
    \begin{split}
        \max_{\vec{B}\in \mathbb{R}^8}\quad &\vec{B}\cdot\vecP\\
        {\rm s.t.}\quad &\vec{B}\cdot \vec{T}_i = 0 \quad \text{for} \quad i \in\{1,2,3\}\\
        & \vec{B}\cdot \vecP_{d,j} \le 1 \quad \text{for}\quad j = 0,1,\ldots, 15,
    \end{split}
\end{equation}
is $1$ and this bound is saturated by at least one of the $\vecP_{d,j}$, we obtain a proof that the self-testing correlation $\vecP$ is non-exposed. To certify that the maximum value of \cref{Non-Exposed_LinearProgram} is upper bounded by 1,  we make use of {\em weak duality}~\cite{BoydBook}, namely the fact that {\em any} feasible solution of the linear program dual to \cref{Non-Exposed_LinearProgram}, i.e.,
\begin{equation}\label{Non-Exposed_DualLinearProgram}
    \begin{split}
        \min_{y_j \ge 0, z_i \in \mathbb{R}} \quad &\sum_{j=0}^{15} y_{j}\\
        {\rm s.t.}\quad &\sum_{j=0}^{15} y_{j}P_{d,j} + \sum_{i} z_{i}\vec{T}_{i} = \vecP \\ &\text{for}\quad i\in\{ 1,2,3\} \text{ and } j = 0,1,\ldots, 15.
    \end{split}
\end{equation}
\end{subequations}
always provides an upper bound to the optimum value of \cref{Non-Exposed_LinearProgram}. Note further that in \cref{Eq:Non-Exposed-LPs}, each $\vec{T}_i$ is obtained form a vector $\ket{\phi_i}$ orthogonal to $\ket{\widetilde{\psi}}$ via \cref{Eq:vecT}. To complete the proof, it suffices to consider any subset $\{\ket{\phi_i}\}$ such that the a dual feasible solution gives exactly $1$.

In the following, we provide the $\{\ket{\phi_i}\}$, $\{\vec{T}_i\}$ and the corresponding dual variables $\{y_j\}$, $\{z_i\}$ needed for the proof of the non-exposed nature of 
$\vecPQ$, $\vecPQii$, $\vecPQiii$ and $\vecP_{\text{Cabello}}$. 

\subsection{Class 3b}~\label{Non-Exposed_Class.3b}
\label{App:NonExposed:3b}

For the self-testing correlation $\vecPQ$ of \cref{Eq_kappa}, which gives the maximal CHSH Bell violation in $\Q_{3b}$, the corresponding reference state and measurements are provided in~\cref{Eq:Class3b} and \cref{Eq:3b:Parameters}. To certify the non-exposed nature of $\vecPQ$, it suffices to consider $\ket{\phi_1}= \ket{0}\!\ket{0}$ as the state orthogonal to the reference state of \cref{Eq:3b:State},
which gives
\begin{equation}
\begin{split}
    \vec{T}_{1} =\,\,
    \begin{tabular}{cc|cc|cc|}
         &  & \multicolumn{2}{c|}{$x=0$} & \multicolumn{2}{c|}{$x=1$} \\
         &  & $0$ & $1$~ & $0$ & $1$ \\
        \hline
        \multirow{2}{*}{$y=0$} & $0$ & $0$ & $0$ & $-t_{1}$ & $t_{1}$\\
                               & $1$ & $0$  & $0$  & $0$ &$0$ \\
        \hline
        \multirow{2}{*}{$y=1$} & $0$ & $-t_{2}$  & $0$ & $-t_{2}$ & $0$ \\
                               & $1$ & $t_{2}$ &$0$  & $t_2-t_1$& $t_{1}$\\
        \hline
        \end{tabular}~,
\end{split}\end{equation}
where
\begin{align}
    \begin{split}
        t_{1} &= \tfrac{1}{6} \left(\sqrt[3]{3 \sqrt{33}-17}-\tfrac{2}{\sqrt[3]{3 \sqrt{33}-17}}+4\right) \approx 0.2282\\
        t_{2} &= \tfrac{1}{6} \left(\sqrt[3]{19-3 \sqrt{33}}+\tfrac{4}{\sqrt[3]{19-3 \sqrt{33}}}-2\right) \approx 0.4196 
    \end{split}
\end{align}
The optimized $\vec{B}$ has only the non-zero components $B(a,b|1,0) =  1\quad\forall\,\,a, b$.

For the dual problem of~\cref{Non-Exposed_DualLinearProgram}, the assignment
\begin{align}
        y_{3} &= \tfrac{1}{6} \left(\sqrt[3]{2 \left(3 \sqrt{33}+13\right)}-\tfrac{4\ 2^{2/3}}{\sqrt[3]{3 \sqrt{33}+13}}-1\right) \approx 0.1478, \nonumber\\ 
        y_{8} &= \tfrac{1}{3} \left(\tfrac{\sqrt[3]{11 \left(3 \sqrt{33}-11\right)}}{2^{2/3}}-\tfrac{22^{2/3}}{\sqrt[3]{3 \sqrt{33}-11}}+2\right) \approx 0.1041 \nonumber\\
        y_{12} & = \tfrac{1}{2} - y_{8}, \quad y_6 = y_3-z_1 t_1,\nonumber\\
        y_9 &= 1- y_3 -y_6-y_{8}-y_{12}, \quad
        z_{1} = 2 y_{3},
\end{align}
and zero otherwise, is easily verified to be a feasible solution with a value of $1$, thus showing that $\vecPQ$ is not exposed. 

\subsection{Class 2a}~\label{Non-Exposed_Class.2a}

For the self-testing correlation $\vecPQii$ of \cref{Eq:PQ2}, which gives the maximal CHSH  violation in $\Q_{2a}$, the corresponding reference state and measurements are, respectively, $\ket{\Phi^+_2}$ and those provided in~\cref{Eq:Class2a}. To certify the non-exposed nature of $\vecPQii$, we consider the singlet state $\ket{\phi_1}= \frac{1}{\sqrt{2}}(\ket{0}\!\ket{1}-\ket{1}\!\ket{0})$ as the state orthogonal to $\ket{\Phi^+_2}$,
which gives
\begin{equation}
\begin{split}
       \tfrac{8}{\sqrt{3}}\,\vec{T}_{1} &=\,\,
    \begin{tabular}{cc|cc|cc|}
         &  & \multicolumn{2}{c|}{$x=0$} & \multicolumn{2}{c|}{$x=1$} \\
         &  & $0$ & $1$~ & $0$ & $1$ \\
        \hline
        \multirow{2}{*}{$y=0$} & $0$ & $0$ & $0$ & $-1$ & $1$ \\ 
	                                  & $1$ & $0$  & $0$  & $1$ &$-1$ \\ 
        \hline
        \multirow{2}{*}{$y=1$} & $0$ & $-1$  & $1$ & $-1$ &  $1$ \\
           	                          & $1$ & $1$ & $-1$  & $1$& $-1$\\ 
        \hline
        \end{tabular}~.
\end{split}
\end{equation}
The optimized $\vec{B}$ has only the non-zero components $B(1,b|0,0) = B(0,b|0,1) = 1\quad\forall\,\,b$.

For the dual problem of~\cref{Non-Exposed_DualLinearProgram}, the assignment
\begin{equation}
    y_{3} = y_{6} = y_{9} = y_{12} = \tfrac{1}{4}, \quad  z_{1} = \tfrac{1}{\sqrt{3}},
\end{equation}
and zero otherwise, is easily verified to be a feasible solution with a value of $1$, thus showing that $\vecPQii$ is not exposed. 

\subsection{Class 2b}~\label{Non-Exposed_Class.2b}

For the self-testing correlation $\vecPQiii$ of \cref{Eq:PQ3}, which gives the maximal CHSH violation in $\Q_{2b}$, the corresponding reference state and measurements are provided in~\cref{Eq:Class2b} and \cref{Eq:2b:Parameters}. To certify the non-exposed nature of $\vecPQiii$, it suffices to consider the following states orthogonal to the reference state of \cref{Eq:Class2b} with \cref{Eq:2b:Parameters},
\begin{equation}
    \begin{split}
        \ket{\phi_1} &= \ket{0}\!\ket{0},\quad
        \ket{\phi_2} = \sin{\alpha}\ket{0}\!\ket{1}+\cos{\alpha}\ket{1}\!\ket{1},
    \end{split}
\end{equation}
where $\alpha = \frac{\pi}{6}$. In this case,
\begin{equation}
\begin{split}
    8\sqrt{\tfrac{2}{3}}\,\vec{T}_{1} &=\,\,
    \begin{tabular}{cc|cc|cc|}
         &  & \multicolumn{2}{c|}{$x=0$} & \multicolumn{2}{c|}{$x=1$} \\
         &  & $0$ & $1$~ & $0$ & $1$ \\
        \hline
        \multirow{2}{*}{$y=0$} & $0$ & $0$ & $0$ & $-2$ & $2$ \\ 
	                                    & $1$ & $0$  & $0$  & $0$ &$0$ \\
        \hline
        \multirow{2}{*}{$y=1$} & $0$ & $-2$  & $0$ & $-3$ & $1$ \\
	                                   & $1$ & $2$ &$0$  & $1$& $1$\\ 
        \hline
        \end{tabular}~,\\
       8\sqrt{\tfrac{2}{3}}\, \vec{T}_{2} &=\,\,
    \begin{tabular}{cc|cc|cc|}
         &  & \multicolumn{2}{c|}{$x=0$} & \multicolumn{2}{c|}{$x=1$} \\
         &  & $0$ & $1$~ & $0$ & $1$ \\
        \hline
        \multirow{2}{*}{$y=0$} & $0$ & $0$ & $0$ & $0$ & $0$\\
                               & $2$ & $2$  & $-2$  & $0$ &$0$ \\  
        \hline
        \multirow{2}{*}{$y=1$} & $0$ & $1$  & $-3$ & $0$ & $-2$ \\
                               & $1$ & $1$ &$1$  & $0$& $2$\\  
        \hline
        \end{tabular}~,       
\end{split}
\end{equation}
and the optimized $\vec{B}$ takes the form of
\begin{equation}
\begin{split}
        \vec{B} &=
    \scalebox{0.92}{\begin{tabular}{cc|cc|cc|}
         &  & \multicolumn{2}{c|}{$x=0$} & \multicolumn{2}{c|}{$x=1$} \\
         &  & $0$ & $1$~ & $0$ & $1$ \\
        \hline
        \multirow{2}{*}{$y=0$} & $0$ & $-b_1$ & $b_2$ & $0$ & $1-b_1-b_2$\\
                               & $1$ & $0$  & $b_2$  & $0$ &$b_1$ \\ 
        \hline
        \multirow{2}{*}{$y=1$} & $0$ & $b_2$  & $0$ & $b_3$ & $b_1$ \\
                               & $1$ & $2b_1+b_2$ &$0$  & $0$& $0$\\
        \hline
        \end{tabular}}~,
\end{split}
\end{equation}
where $b_1\approx 0.2197$, $b_2\approx 0.3410$, and $b_3\approx0.6590$.

For the dual program in~\cref{Non-Exposed_DualLinearProgram}, the assignment
\begin{equation}
\begin{split}
     y_{2} &= y_{3} = y_{12} = y_{15} = \tfrac{1}{4}, \quad
     z_{1} = z_2 = \tfrac{1}{\sqrt{6}},
\end{split}
\end{equation}
and zero otherwise, is easily verified to be a feasible solution with a value of $1$, thus showing that $\vecPQiii$ is not exposed. 

\subsection{Class 2c}~\label{Non-Exposed_Class.2c}  

For the self-testing correlation $\vecPC$ of \cref{Eq:PCabello}, which gives the maximal CHSH violation in $\Q_{2c}$, the corresponding reference state and measurements are provided in~\cref{Eq:Class2c}. To certify the non-exposed nature of $\vecPC$, it suffices to consider the following states orthogonal to the reference state of \cref{Eq:CabelloState}
\begin{equation}\label{Eq:phi:2c}
    \begin{split}
        \ket{\phi_1} &= k_{2}\ket{0}\!\ket{0}-k_{1}\ket{0}\!\ket{1}-k_{3}\ket{1}\!\ket{0}+k_{2}\ket{1}\!\ket{1},\\
        \ket{\phi_2} &= k_{3}\ket{0}\!\ket{0}-k_{2}\ket{0}\!\ket{1}+k_{2}\ket{1}\!\ket{0}-k_{1}\ket{1}\!\ket{1}.
    \end{split}
\end{equation}
The exact analytic expression of $\vec{T}_1$ and $\vec{T}_2$ can be computed from \cref{Eq:vecT}, \cref{Eq:Class2c}, and \cref{Eq:phi:2c}. Below, we provide their approximate form for ease of reference:
\begin{equation}
\begin{split}
    \vec{T}_{1} &=
    \scalebox{0.92}{\begin{tabular}{cc|cc|cc|}
         &  & \multicolumn{2}{c|}{$x=0$} & \multicolumn{2}{c|}{$x=1$} \\
         &  & $0$ & $1$~ & $0$ & $1$ \\
        \hline
        \multirow{2}{*}{$y=0$} & $0$ & $0$ & $-t_{1}$ & $t_{2}$ & $t_{3}$\\
                               & $1$ & $0$  & $t_{1}$  & $t_{4}$ &$-t_2-t_3-t_4$ \\
        \hline
        \multirow{2}{*}{$y=1$} & $0$ & $t_{5}$  & $t_{6}$ & $t_{7}$ & $0$ \\
                               & $1$ & $-t_{5}$ &$-t_{6}$  & $t_{8}$& $-t_7-t_8$\\
        \hline
        \end{tabular}}~,\\
        \vec{T}_{2} &=
    \scalebox{0.92}{\begin{tabular}{cc|cc|cc|}
         &  & \multicolumn{2}{c|}{$x=0$} & \multicolumn{2}{c|}{$x=1$} \\
         &  & $0$ & $1$~ & $0$ & $1$ \\
        \hline
        \multirow{2}{*}{$y=0$} & $0$ & $0$ & $-u_{1}$ & $u_{2}$ & $u_{3}$\\
                               & $1$ & $u_{1}$  & $0$  & $u_{4}$ &$-u_2-u_3-u_4$ \\
        \hline
        \multirow{2}{*}{$y=1$} & $0$ & $u_{5}$  & $-u_{6}$ & $u_{7}$ & $0$ \\
                               & $1$ & $u_{6}$ &$-u_{5}$  & $u_{8}$& $-u_7-u_8$\\
        \hline
        \end{tabular}}~,\\
\end{split}
\end{equation}
where $t_1\approx 0.1507$, $t_2\approx 0.1319$, $t_3\approx -0.2826$, $t_4\approx0.2891$, $t_5=0.3202$, $t_6=0.0909$, $t_7\approx 0.4112$, and $t_8\approx 0.0098$ for $\vec{T}_1$, while $u_1\approx0.4214$, $u_2\approx-0.2950$, $u_3\approx-0.1264$, $u_4\approx0.2770$, $u_5\approx0.0871$, $u_6\approx0.3342$, $u_7\approx-0.2471$, $u_8\approx0.2291$ for $\vec{T}_2$.

The optimized $\vec{B}$ takes the form of
\begin{equation}
\begin{split}
        \vec{B} &=
    \begin{tabular}{cc|cc|cc|}
         &  & \multicolumn{2}{c|}{$x=0$} & \multicolumn{2}{c|}{$x=1$} \\
         &  & $0$ & $1$~ & $0$ & $1$ \\
        \hline
        \multirow{2}{*}{$y=0$} & $0$ & $b_1$ & $0$ & $0$ & $b_2$\\
                                            & $1$ & $1$  & $0$  & $0$ &$0$ \\ 
        \hline
        \multirow{2}{*}{$y=1$} & $0$ & $0$  & $1$ & $0$ & $-b_2$ \\
                                           & $1$ & $-b_3$ &$0$  & $b_4$& $b_5$\\
        \hline
        \end{tabular}~,
\end{split}
\end{equation}
where $b_1\approx 0.6918$, $b_2\approx 0.4170$, $b_3\approx0.2005$, $b_4\approx 0.1349$, and $b_5\approx 0.0917$.

For the dual program in~\cref{Non-Exposed_DualLinearProgram}, the assignment
\begin{align}
        y_{3} &= \tfrac{1}{6} \left(-\sqrt[3]{53-6 \sqrt{78}}-\tfrac{1}{\sqrt[3]{53-6 \sqrt{78}}}+7\right) \approx 0.3427, \nonumber\\
        y_{12} &= -\tfrac{\sqrt[3]{9-\sqrt{78}}}{3^{2/3}}-\tfrac{1}{\sqrt[3]{3 \left(9-\sqrt{78}\right)}}+2 \approx 0.4786, \nonumber\\
        y_{15} &= 1 - y_3 - y_{12} \approx 0.1787, \nonumber\\
        z_{1} &= -\sqrt{\tfrac{\zeta_1-\frac{133727}{\zeta_1}-145}{48}} \approx -0.6414, \\
        z_{2} &= \tfrac{1}{6} \left(\zeta_2-\tfrac{23}{\zeta_2}-11\right) \approx 0.3855, \nonumber\\
        \zeta_{1} &= \sqrt[3]{6827808 \sqrt{78}+35282447},\nonumber\\
        \zeta_{2} &=\sqrt[3]{186 \sqrt{78}+1639},\nonumber
\end{align}
and zero otherwise, is  easily verified to be a feasible solution with a value of $1$, thus showing that $\vecPC$ is not exposed.

\subsection{Class 1}~\label{Non-Exposed_Class.1}

For the self-testing correlation $\vecPQiv$ of \cref{Eq:PQ4}, which gives the maximal CHSH violation in $\Q_{1}$, the corresponding reference state and measurements are provided in~\cref{Eq:Class1}. To certify the non-exposed nature of $\vecPQiv$, it suffices to consider the following states orthogonal to the state of \cref{Eq:StateClass1}
\begin{equation}\label{Eq:phi:1}
    \begin{split}
        \ket{\phi_1} &= \ket{0}\!\ket{0},\\
        \ket{\phi_2} &= \cos\theta \left(\frac{\ket{0}\!\ket{1}+\ket{1}\!\ket{0}}{\sqrt{2}}\right)
        -\sin\theta\ket{1}\!\ket{1}.
    \end{split}
\end{equation}
The exact analytic expression of $\vec{T}_1$ and $\vec{T}_2$ can be computed from \cref{Eq:vecT}, \cref{Eq:Class2c}, and \cref{Eq:phi:2c}. Below, we provide their approximate form for ease of reference:
\begin{align}
    \vec{T}_{1} &=
    \begin{tabular}{cc|cc|cc|}
         &  & \multicolumn{2}{c|}{$x=0$} & \multicolumn{2}{c|}{$x=1$} \\
         &  & $0$ & $1$~ & $0$ & $1$ \\
        \hline
        \multirow{2}{*}{$y=0$} & $0$ & $0$ & $0$ & $t_{1}$ & $-t_{1}$\\
                               & $1$ & $0$  & $0$  & $0$ &$0$ \\
        \hline
        \multirow{2}{*}{$y=1$} & $0$ & $t_{1}$  & $0$ & $t_{2}$ & $t_{3}$ \\
                               & $1$ & $-t_{1}$ & $0$  & $t_{3}$& $-t_2-2t_3$\\
        \hline
        \end{tabular}~,\\
        \vec{T}_{2} &=
    \begin{tabular}{cc|cc|cc|}
         &  & \multicolumn{2}{c|}{$x=0$} & \multicolumn{2}{c|}{$x=1$} \\
         &  & $0$ & $1$~ & $0$ & $1$ \\
        \hline
        \multirow{2}{*}{$y=0$} & $0$ & $0$ & $u_{1}$ & $u_{2}$ & $u_{3}$\\
                               & $1$ & $u_{1}$  & $-2u_{1}$  & 
                               $u_{4}$ &$-u_{234}$ \\
        \hline
        \multirow{2}{*}{$y=1$} & $0$ & $u_{2}$  & $u_{4}$ & $u_{5}$ & $u_{6}$ \\
                               & $1$ & $u_{3}$ &$-u_{234}$  
                               & $u_{6}$& $-u_5-2u_6$\\
        \hline
        \end{tabular}~,\nonumber
\end{align}
where $t_1\approx 0.3298$, $t_2\approx 0.4766$, and $t_3\approx -0.1468$, for $\vec{T}_1$, while $u_1\approx0.1111$, $u_{234}\equiv u_2+u_3+u_4$, $u_2\approx 0.0397$, $u_3\approx 0.0715$,  $u_4\approx -0.3113$, $u_5\approx-0.1429$, and $u_6\approx -0.1288$ for $\vec{T}_2$.

The optimized $\vec{B}$ takes the form of
\begin{equation}
\begin{split}
        \vec{B} &=
    \begin{tabular}{cc|cc|cc|}
         &  & \multicolumn{2}{c|}{$x=0$} & \multicolumn{2}{c|}{$x=1$} \\
         &  & $0$ & $1$~ & $0$ & $1$ \\
        \hline
        \multirow{2}{*}{$y=0$} & $0$ & $1$ & $1+b_1$ & $-b_1$ & $0$\\
                                            & $1$ & $1$  & $1$  & $0$ &$b_2$ \\ 
        \hline
        \multirow{2}{*}{$y=1$} & $0$ & $0$  & $0$ & $0$ & $b_4$ \\
                                           & $1$ & $0$ &$b_3$  & $b_5$& $0$\\
        \hline
        \end{tabular}~,
\end{split}
\end{equation}
where $b_1\approx 0.4850$, $b_2\approx -0.1944$, $b_3\approx -0.6795$, $b_4\approx -0.7872$, and $b_5\approx -0.3025$.

For the dual program in~\cref{Non-Exposed_DualLinearProgram}, let $\zeta_1 \approx 0.4450, \zeta_2 \approx 0.7530$, and $\zeta_3  \approx 0.0677$ be the respective smallest positive roots of the following cubic polynomials:
$q_1(x) = x^3-x^2-2 x+1, q_2(x) = x^3-7 x^2+14 x-7,$ and $q_3(x) = x^3-32 x^2-116 x+8.$
The assignment
\begin{align}
        y_{3} &= y_{12} = \zeta_1, &
        y_{15} &= 1-2y_{3} \approx 0.1099, \nonumber\\
        z_{1} &= -\sqrt{\zeta_{2}} \approx -0.8678, &
        z_{2} &= \sqrt{\zeta_{3}} \approx 0.2602, 
\end{align}
and zero otherwise, is  easily verified to be a feasible solution with a value of $1$, thus showing that $\vecPQiv$ is not exposed.

\section{Proof of Corollary~\ref{Prop:MaxCHSH:MES}}
\label{App:MaxCHSH:MES}

From the proof of Corollary~\ref{CHSH:Decompose:MES}, we see that in maximizing the value of any Bell inequality in the CHSH Bell scenario by $\ket{\Phi^+_d}$, it suffices to consider correlations $\vecP\in\M_d$ taking the form of \cref{Eq:MES:vecP:Decomposition}.
Let $\vec{B}$ be the Bell coefficients associated with the CHSH inequality of \cref{Eq:CHSH}, then
\begin{align}~\label{maxP}
    &\max_{\vecP\in\M_d}\vec{B}\cdot\vecP \nonumber\\
    = &\max_{\vecP'_i\in \M_2,\vecP''_j\in \L,  k\ge 0, \frac{d-k}{2}\in \mathbb{Z}^+} \frac{2}{d}\sum_{i=1}^{\frac{d-k}{2}}\vec{B}\cdot\vecP'_i+\frac{1}{d}\sum_{j = 1}^{k}\vec{B}\cdot\vecP''_j \nonumber\\
     \le &\max_{k\ge 0, \frac{d-k}{2}\in \mathbb{Z}^+} \frac{2}{d}\sum_{i=1}^{\frac{d-k}{2}}\max_{\vecP'_i\in \M_2}\vec{B}\cdot\vecP'_i+\frac{1}{d}\sum_{j = 1}^{k}\max_{\vecP''_j\in \L}\vec{B}\cdot\vecP''_j \nonumber\\
     = &\max_{k\ge 0, \frac{d-k}{2}\in \mathbb{Z}^+} \frac{2}{d}\sum_{i=1}^{\frac{d-k}{2}}2\sqrt{2}+\frac{1}{d}\sum_{j = 1}^{k}2 \nonumber\\
     = &\left\{ 
        \begin{aligned}
            ~&2\sqrt{2}\left(\frac{d-1}{d}\right)+\frac{2}{d},&&d = \text{odd},\\
            ~&2\sqrt{2},&&d = \text{even}.
        \end{aligned}
    \right .
\end{align}
These upper bounds are indeed attainable [see Eq.~(5) of~\cite{Liang2006}], thus completing the proof that the maximal CHSH  violation by $\ket{\MESd}$ is given in the last line of \cref{maxP}.

\section{Proof of Lemma~\ref{lem2}}
\label{App:MESlocal}

Let us begin by noting that the transposition operation preserves the hermiticity of a matrix. Then, from spectral theorem, we may write :
\begin{equation}
    \begin{split}
    \begin{cases}
        &E_{0|0}\tp=\sum_{i=0}^{r^0} \lambda_i \proj{e_i},\quad0< \lambda_i \le 1,\\
        &E_{1|0}\tp=\mathbb{1}_d-E_{0|0}\tp,\\
    \end{cases}\\
    \begin{cases}
        &F_{0|0}=\sum_{i=0}^{r^1} \lambda'_i\proj{f_i},\quad 0< \lambda'_i   \le 1,\\
        &F_{1|0}=\mathbb{1}_d-F_{0|0}.
     \end{cases}
    \end{split}
\end{equation}
where $\lambda_i$ and $\lambda_i'$ are, respectively, the nonvanshing eigenvalues of $E\tp_{0|0}$ and $F_{0|0}$. 
Without loss of generality, suppose $P(0,0|0,0)=0$, then from \cref{Eq:MESCorrelation}, we get
\begin{equation}\label{Eq:VanishingSum}
    \frac{1}{d}\tr[E_{0|0}\tp F_{0|0}]= \frac{1}{d}\sum_{i=0}^{r^0}\sum_{j=0}^{r^1}\lambda_i \lambda'_j|\langle e_i | f_j \rangle|^2=0.
\end{equation}
Notice that the RHS of \cref{Eq:VanishingSum} is a sum of non-negative terms. The fact that this sum vanishes means that $|\langle e_i | f_j \rangle|=0$ for all $i,j$, which means that $E_{0|0}\tp F_{0|0}=0$, i.e., $E_{0|0}\tp$ and $F_{0|0}$ commute 
and  hence can be diagonalized in the same basis.

\end{CJK*}
\end{document}